\documentclass[10pt,twocolumn,letterpaper]{article}

\usepackage{wacv}      

\usepackage{times}
\usepackage{epsfig}
\usepackage{graphicx}
\usepackage{amsmath}
\usepackage{amssymb}
\usepackage{amsthm}
\usepackage{mathtools}
\usepackage{booktabs}
\usepackage{multirow}
\usepackage{array}
\usepackage{enumitem}
\usepackage{xcolor}
\usepackage{microtype}
\usepackage{algorithm}
\usepackage{algpseudocode}

\newtheorem{proposition}{Proposition}

\newcommand{\R}{\mathbb{R}}

\newcommand{\Bcal}{\mathcal{B}}
\newcommand{\Acal}{\mathcal{A}}

\newcommand{\FP}{\mathrm{FP}}

\newcommand{\INT}{\mathrm{INT8}}
\newcommand{\FPh}{\mathrm{FP16}}

\newcommand{\ECE}{\mathrm{ECE}}
\newcommand{\AURC}{\mathrm{AURC}}
\newcommand{\Dice}{\mathrm{Dice}}

\newcommand{\softmax}{\mathrm{softmax}}

\newcommand{\Rel}{\mathcal{D}_{\mathrm{rel}}}

\newcommand{\Unc}{\mathcal{U}}

\setlist[itemize]{leftmargin=*,nosep}
\setlist[enumerate]{leftmargin=*,nosep}

\definecolor{wacvblue}{rgb}{0.21,0.49,0.74}
\usepackage[pagebackref,breaklinks,colorlinks,allcolors=wacvblue]{hyperref}
\usepackage{multirow}
\providecommand{\best}[1]{\textbf{\underline{#1}}}
\providecommand{\second}[1]{\textbf{#1}}
\def\wacvPaperID{3654} 
\def\confName{WACV}
\def\confYear{2027}

\title{Reliability-Preserving Compression\\for Deployable Medical Image Segmentation}

\author{Sagar Kumar\\
Indian Institute of Science, India\\
{\tt\small sagarkumar1@iisc.ac.in}
\and
Shailja Sharma\\
Indian Institute of Science (IISc), India\\
{\tt\small sshailja@iisc.ac.in}
\and
Tirthajit Baruah\\
Indian Institute of Science (IISc), India\\
{\tt\small tirthajitb@iisc.ac.in}
\and
Punit Rathore\\
Indian Institute of Science (IISc), India\\
{\tt\small prathore@iisc.ac.in}
}

\begin{document}
\maketitle
\begin{abstract}
Post-training compression of pretrained medical segmentation models commonly uses quantization, pruning, or distillation to reduce model size and computational cost. 
Existing medical segmentation quantization studies mainly evaluate Dice, model size, memory, and latency, but this is incomplete for safety-critical use because compressed models should also preserve uncertainty signals that indicate unreliable predictions. 
We study this problem in compressed MRI/CT segmentation and formalize why Dice can remain stable while uncertainty changes. 
Motivated by this failure mode, we propose score-guided aggregate-validated mixed precision (SGAV-MP), a post-training compression framework that selects layer precisions under reliability constraints. 
SGAV-MP ranks quantization actions with a compact benefit-risk score and accepts them only after full-model reliability validation on a calibration set. 
Across five 2D backbones and three 3D benchmarks, SGAV-MP preserves FP32-level Dice and high uncertainty-map agreement, achieving $3.59\times$--$7.8\times$ compression in 2D and $4.73\times$--$7.82\times$ in 3D. We also report backend-supported TensorRT deployment measurements separately to distinguish runtime execution from reliability-preserving storage compression.

\end{abstract}   

\section{Introduction}
Medical image segmentation has been dominated by U-Net-style architectures and their extensions. U-Net-style CNNs, nnU-Net, transformer-based architectures, and promptable medical segmentation models, have achieved strong performance across medical imaging modalities~\cite{ronneberger2015unet,isensee2021nnunet,hatamizadeh2022swinunetr,ma2024segment}. However, these models often operate on high-resolution images or volumetric scans. This makes deployment costly in memory, latency, and computation. This motivates compression for practical use, while actual runtime gains depend on backend and hardware support.

Model compression reduces deployment cost through pruning, quantization, and distillation~\cite{han2015deepcompression,jacob2018quantization,hinton2015distilling}. Post-training quantization (PTQ) is especially attractive because it can compress a pretrained model using only a small calibration set and without full retraining~\cite{nagel2020adaround,li2021brecq}. This is useful in medical settings, where data access, annotation, and retraining can be expensive. Recent medical segmentation quantization methods, including EfficientQ~\cite{zhang2024efficientq}, MedPTQ~\cite{qu2026medptq}, and QMedSAM~\cite{lu2024qmedsam}, report strong compression or deployment-related gains while preserving segmentation accuracy. However, their evaluation is mainly based on hard-mask accuracy and efficiency metrics such as Dice, HD95, model size, memory, and latency.

\begin{table}[t]
\centering
\caption{Representative cases where INT4 quantization may preserve usable Dice while degrading uncertainty and calibration. UQ Corr. denotes uncertainty-map correlation with the FP32 model.}
\label{tab:intro_gap_examples}
\setlength{\tabcolsep}{5.5pt}
\resizebox{\linewidth}{!}{
\begin{tabular}{cccccc}
\toprule
\textbf{Dataset} & \textbf{Model} &
\textbf{Dice} &
\textbf{ECE} &
\textbf{Brier} &
\textbf{UQ Corr.} \\
\midrule

PROMISE12 & EfficientNet-B0 
& $0.8892 \rightarrow 0.882$
& $0.0982 \rightarrow 0.112$
& $0.196 \rightarrow 0.22$
& $1.000 \rightarrow 0.428$ \\

PROMISE12 & ResNet50
& $0.886 \rightarrow 0.881$
& $0.103 \rightarrow 0.114$
& $0.205 \rightarrow 0.229$
& $1.000 \rightarrow 0.324$ \\

MosMedPlus & ResNet50
& $0.785 \rightarrow 0.651$
& $0.206 \rightarrow 0.365$
& $0.412 \rightarrow 0.730$
& $1.000 \rightarrow 0.119$ \\
\bottomrule
\end{tabular}
}
\end{table}
Accuracy-only evaluation is incomplete for safety-critical segmentation. A compressed model should not only predict a mask. It should also indicate when the prediction may be unreliable. Predictive uncertainty and calibration are important for ambiguous boundaries, low-contrast lesions, artifacts, small structures, domain shift, and out-of-distribution anatomy~\cite{mehrtash2020calibration,jungo2019assessing,huang2024reviewuq}. Table~\ref{tab:intro_gap_examples} shows the practical gap. INT4 quantization can leave hard-mask performance usable in some cases, while uncertainty correlation and calibration degrade substantially.

This paper asks:\begin{quote}
\emph{Can a compressed medical segmenter preserve not only segmentation accuracy, but also uncertainty, calibration, and failure-detection behavior?}
\end{quote}

We address this question with score-guided aggregate-validated mixed precision (SGAV-MP), a post-training mixed-precision quantization framework for reliability-preserving compression. 
SGAV-MP ranks layer-precision actions using compression benefit, layer importance, and connection-aware risk, but accepts an action only if the complete compressed model satisfies aggregate reliability constraints on a calibration set. 
We report compression as the main objective and separately report backend-supported FP32/FP16/INT8 deployment-related measurements for representative 2D and 3D settings; mixed INT4/INT2 policies are treated as storage-compression policies unless directly supported by the backend.

Our contributions are:
\begin{itemize}
\item We show that compressed medical segmenters can preserve Dice while degrading calibration, uncertainty-map fidelity, and failure-detection behavior.
\item We propose SGAV-MP, a post-training mixed-precision quantization method with full-model reliability validation before accepting compression actions.
\item We evaluate reliability-preserving compression across medical segmentation datasets and models using segmentation, efficiency, calibration, uncertainty-preservation, failure-detection, and backend-supported INT8 deployment-related metrics.
\end{itemize}

\section{Related Work}

Medical image segmentation is commonly based on U-Net-style models and their extensions. 
U-Net~\cite{ronneberger2015unet}, 3D U-Net~\cite{cicek20163dunet}, and nnU-Net~\cite{isensee2021nnunet} remain strong biomedical segmentation baselines, while transformer and foundation-style models such as SwinUNETR~\cite{hatamizadeh2022swinunetr} and VISTA3D~\cite{he2025vista3d} further improve representation capacity. 
However, their memory and inference requirements motivate compression for practical use.

Model compression reduces parameter storage and can reduce deployment cost through quantization, pruning, and distillation. 
Mixed-precision quantization assigns higher precision to sensitive layers and lower precision to robust layers. 
HAWQ-V2~\cite{dong2020hawqv2} uses Hessian trace information, while Chen et al.~\cite{chen2021mckp} formulate precision assignment as a constrained multiple-choice knapsack problem. 
Other PTQ methods use calibration data and hardware-aware constraints~\cite{pandey2023practical}, and AdaRound~\cite{nagel2020adaround} learns adaptive weight rounding from calibration data. 
In medical segmentation, EfficientQ~\cite{zhang2024efficientq} and MedPTQ~\cite{qu2026medptq} show that PTQ can preserve segmentation accuracy while improving compression or backend-supported efficiency. 
These works mainly optimize accuracy and efficiency; we study whether compressed models also preserve uncertainty and reliability.

Uncertainty quantification estimates prediction reliability using Monte Carlo dropout~\cite{gal2016dropout}, deep ensembles~\cite{lakshminarayanan2017deepensembles}, probabilistic segmentation~\cite{kohl2018probunet}, and uncertainty distillation such as DUDES~\cite{landgraf2023dudes}. 
In medical segmentation, calibration and uncertainty support failure prediction, out-of-distribution detection, and reliability assessment~\cite{mehrtash2020calibration,huang2024reviewuq}. 
We connect this reliability perspective with post-training compression by evaluating whether quantized segmenters preserve calibration, uncertainty maps, and failure-detection behavior.
\section{Application Setting and Problem Setup}
\label{sec:setting_problem}

We study post-training compression of an already-trained medical image segmentation model under compression and backend-related constraints such as model size, memory, latency, and throughput. We distinguish parameter-storage compression from backend-supported runtime execution, since actual speedup depends on hardware and kernel support. Let $x:\Omega\rightarrow\R^{C_{\mathrm{in}}}$ be an input image defined on the pixel or voxel domain $\Omega$, and let $y:\Omega\rightarrow\{1,\ldots,C\}$ be the corresponding segmentation map with $C$ classes. A full-precision model $f_{\theta}^{\FP}$ produces logits $z_v^{\FP}=f_{\theta}^{\FP}(x)_v\in\R^C$ and probabilities $p_v^{\FP}=\softmax(z_v^{\FP})$ at each location $v\in\Omega$.

\begin{equation}
\begin{gathered}
q=(q_1,\ldots,q_\ell,\ldots,q_L),
\qquad
q_\ell\in\Bcal,\\
\Bcal=\{\mathrm{FP32},\FPh,\INT,\mathrm{INT4}, \mathrm{INT2}\}.
\end{gathered}
\label{eq:precision_policy}
\end{equation}
where $L$ is the number of quantizable layers or blocks. The compressed model produces probabilities
\begin{equation}
p^{q}_{v}=\softmax(f_{\theta}^{q}(x)_v),
\qquad v\in\Omega .
\label{eq:probabilities_compact}
\end{equation}

Standard compression minimizes storage or deployment-related cost while preserving hard-mask accuracy:
\begin{equation}
\begin{aligned}
\min_q \quad & \mathcal{J}(q)\\
\text{s.t.}\quad
& \Dice(f_\theta^q)
\geq
\Dice(f_\theta^{\FP})-\varepsilon_D .
\end{aligned}
\label{eq:standard_compression}
\end{equation}
Here $\mathcal{J}(q)$ denotes model size, latency, memory, or a weighted compression cost. This objective is incomplete because Dice is computed from hard segmentation masks and does not measure calibration or uncertainty reliability.

We instead require reliability-preserving compression:
\begin{equation}
\begin{aligned}
\Dice(f_\theta^q) &\geq \Dice(f_\theta^{\FP})-\varepsilon_D,\\
\ECE(f_\theta^q) &\leq \ECE(f_\theta^{\FP})+\varepsilon_C,\\
\rho(\Unc_q,\Unc_{\FP}) &\geq \tau_U,\\
\AURC(f_\theta^q) &\leq \AURC(f_\theta^{\FP})+\varepsilon_R .
\end{aligned}
\label{eq:reliability_constraints}
\end{equation}
Here $\Unc_q:\Omega\rightarrow\R$ is the uncertainty map of the compressed model, $\ECE$ measures confidence calibration~\cite{guo2017calibration}, $\rho(\cdot,\cdot)$ measures rank agreement between compressed and full-precision uncertainty maps, and $\AURC$ measures selective risk-coverage behavior~\cite{geifman2017selective}. The constants $\varepsilon_D$, $\varepsilon_C$, and $\varepsilon_R$ define allowable Dice, calibration, and risk-coverage degradation, while $\tau_U$ defines the minimum acceptable uncertainty-map agreement.

The calibration set $\mathcal{X}=\{(x_i,y_i)\}_{i=1}^{N_c}$ is disjoint from the held-out test split and is used only for range estimation, SGAV-MP scoring, and aggregate validation. 
No weights are updated during policy selection, and all final metrics are reported on the held-out test set. 
Calibration-size sensitivity is reported in the supplementary material.

\section{Reliability Gap Under Quantization}
\label{sec:reliability_gap}

We now explain why preserving segmentation accuracy under quantization does not necessarily preserve uncertainty reliability. In the previous section, a compressed model $f_\theta^q$ was defined by a precision policy $q$. For a pixel or voxel $v\in\Omega$, let $z_v\in\R^C$ denote the full-precision logits and let
\begin{equation}
z^q_v=z_v+\epsilon_v
\label{eq:perturbed_logits}
\end{equation}
denote the compressed logits, where $\epsilon_v$ is the effective perturbation induced by quantization. The corresponding probability vectors are
\begin{equation}
p_v=\softmax(z_v),
\qquad
p^q_v=\softmax(z^q_v).
\label{eq:softmax_probs}
\end{equation}
The hard segmentation label is obtained from $\arg\max_c z_{v,c}$, while uncertainty measures such as predictive entropy depend on the full probability vector:
\begin{equation}
H(p_v)=-\sum_{c=1}^{C}p_{v,c}\log p_{v,c}.
\label{eq:entropy_basic}
\end{equation}
Therefore, hard-mask stability and uncertainty stability are different requirements.

\begin{proposition}[Stable hard label can still change entropy]
\label{prop:entropy_gap}
Consider a two-class logit vector $z=(a,0)$ with $a>0$ where the first coordinate corresponds to class $1$ and the second to class $0$.
Let the quantized logits be 
\begin{equation}
z^q=(a-\delta,\delta),
\qquad
0<\delta<\frac{a}{2}.
\end{equation}
Then
\begin{equation}
\arg\max_c z_c
=
\arg\max_c z^q_c
=
1,
\end{equation}
but
\begin{equation}
H(\softmax(z^q))>H(\softmax(z)).
\end{equation}
\end{proposition}

\begin{proof}
Since $0<\delta<a/2$, we have $a-\delta>\delta$, so the argmax label is unchanged. For two classes, the softmax distribution depends only on the logit gap. Quantization reduces the gap from $a$ to $a-2\delta<a$, making the distribution less peaked and increasing binary entropy. Thus, the hard prediction is preserved while uncertainty changes.
\end{proof}

This motivates reliability-preserving compression: a compressed model should preserve not only Dice, but also calibration, uncertainty-map structure, and failure-detection behavior.

\section{SGAV-MP: Scoring and Validation}
\label{sec:method}

Given a trained FP32 segmentation model, a calibration set, and a compression objective, SGAV-MP selects a mixed-precision policy in a post-training setting. 
The goal is to obtain a compact model that preserves segmentation accuracy, uncertainty-map structure, and failure-detection behavior under explicit validation constraints.

\subsection{Compression as Reliability-Constrained Set Selection}
Let $\Acal$ be the candidate action set, where each action $a=(\ell,b)$ assigns layer or block $\ell$ to precision $b\in\Bcal$. For a selected action set $A\subseteq\Acal$, let $q_A$ be the induced precision policy and $q_{\emptyset}$ be the full-precision policy. We define
\begin{equation}
F(A)=\mathcal{J}(q_{\emptyset})-\mathcal{J}(q_A),
\qquad
R(A)=\Rel(q_A),
\label{eq:set_functions}
\end{equation}
where $F(A)$ is the compression benefit and $R(A)$ is the aggregate reliability damage. Reliability-aware compression is then
\begin{equation}
\begin{aligned}
\max_{A\subseteq\Acal}\quad & F(A)\\
\mathrm{s.t.}\quad & R(A)\leq \tau_{\mathrm{rel}} .
\end{aligned}
\label{eq:reliability_constrained_set}
\end{equation}
The search space grows exponentially with the number of quantizable layers, so SGAV-MP uses score-guided greedy selection followed by aggregate full-model reliability validation.

\subsection{Reliability-Constrained Action Scoring}
\label{subsec:rel_action_scoring}
To score whether a candidate policy preserves the reliability behavior of the FP32 model, we define the score-level degradation
\begin{equation}
\Rel(q)
=
\alpha_D\Delta_D(q)
+
\alpha_C\Delta_C(q)
+
\alpha_U\Delta_U(q)
+
\alpha_R\Delta_R(q),
\label{eq:aggregate_reliability_degradation}
\end{equation}
where $\alpha_D,\alpha_C,\alpha_U,\alpha_R\geq0$ are weighting coefficients. The degradation terms are
\begin{align}
\Delta_D(q)
&=
\max\{0,\Dice(f^{\FP})-\Dice(f^q)\}, 
\label{eq:delta_dice_compact}
\\
\Delta_C(q)
&=
\max\{0,\ECE(f^q)-\ECE(f^{\FP})\}, 
\label{eq:delta_ece_compact}
\\
\Delta_U(q)
&=
1-\rho(\Unc_q,\Unc_{\FP}), 
\label{eq:delta_unc_compact}
\\
\Delta_R(q)
&=
\max\{0,\AURC(f^q)-\AURC(f^{\FP})\}. 
\label{eq:delta_aurc_compact}
\end{align}
Thus, $D_{\mathrm{rel}}(q)$ measures degradation in segmentation accuracy, calibration, uncertainty-map agreement, and risk--coverage behavior relative to the full-precision model.
For each candidate action $a=(\ell,b)$, SGAV-MP estimates a compression benefit and a reliability risk. The compression benefit is
\begin{equation}
U_{\ell b}^{(t)}
=
\mathcal{J}(q^{(t)})
-
\mathcal{J}(q^{(t)}\oplus(\ell,b)),
\label{eq:state_deployment_utility}
\end{equation}
where $\mathcal{J}(\cdot)$ denotes model size, memory, latency, energy, or a weighted cost when the corresponding backend measurement is available. When model size is used,
\begin{equation}
U_{\ell b}
=
N_\ell(32-r_b),
\label{eq:size_utility_compact}
\end{equation}
where $N_\ell$ is the number of parameters in layer $\ell$ and $r_b$ is the bit-width of precision $b$.

The action risk is estimated using a plug-in importance score. In the main implementation, we use empirical reliability impact:
\begin{equation}
I_{\ell b}^{\mathrm{rel}}
=
\Rel(q^{(\ell,b)}),
\label{eq:empirical_reliability_importance_compact}
\end{equation}
where $q^{(\ell,b)}$ denotes the policy in which only layer $\ell$ is assigned precision $b$ and all other layers remain full precision. More generally, SGAV-MP can use any normalized sensitivity estimator:
\begin{equation}
I_{\ell b}
=
\eta_{\mathrm{rel}}\widehat{I}_{\ell b}^{\mathrm{rel}}
+
\eta_H\widehat{I}_{\ell b}^{H}
+
\eta_A\widehat{I}_{\ell b}^{A},
\label{eq:general_plugin_importance_compact}
\end{equation}
where $\widehat{I}_{\ell b}^{\mathrm{rel}}$ is the normalized empirical reliability impact, $\widehat{I}_{\ell b}^{H}$ may denote a Hessian- or Fisher-style sensitivity score, and $\widehat{I}_{\ell b}^{A}$ may denote an activation-perturbation score. In our experiments, we use only the empirical term and set $\eta_H=\eta_A=0$ to avoid additional computation.

To account for architectural sensitivity, we further define a connection-aware risk
\begin{equation}
M_\ell
=
\lambda_K K_\ell
+
\lambda_{\nu}\nu_\ell,
\label{eq:connection_risk_compact}
\end{equation}
where
\begin{equation}
K_\ell
=
\frac{d_\ell^{in}+d_\ell^{out}}
{\max_j(d_j^{in}+d_j^{out})+\epsilon}
\label{eq:degree_risk_compact}
\end{equation}
is the normalized graph degree of layer $\ell$, and
\begin{equation}
\nu_{\ell}
=
\begin{cases}
1, & \text{if  $\ell$ participates in skip or multi-scale fusion},\\
0, & \text{otherwise}.
\end{cases}
\label{eq:skip_indicator_compact}
\end{equation}

The final priority score for action $(\ell,b)$ is
\begin{equation}
S_{\ell b}
=
\frac{
\widehat{U}_{\ell b}
}
{
\widehat{I}_{\ell b}
+
\lambda_M M_\ell
+
\epsilon
}.
\label{eq:sgav_score_compact}
\end{equation}
Here $\widehat{U}_{\ell b}$ and $\widehat{I}_{\ell b}$ are normalized over candidate actions. A high score indicates that the action provides large compression benefit with low estimated reliability damage and structural risk.
The score is only used to rank candidate actions. SGAV-MP does not assume that layer-wise reliability impacts are additive. After an action is proposed by Eq.~\eqref{eq:sgav_score_compact}, the complete mixed-precision model is evaluated on the calibration set and accepted only if it satisfies the aggregate reliability constraints in Eq.~\eqref{eq:reliability_constraints}.

\subsection{Pareto Filtering and Score-Guided Priority}
Before ranking actions, we remove actions that are strictly dominated in the utility-risk plane. For two actions $a_1$ and $a_2$, action $a_1$ dominates $a_2$ if
\begin{equation}
U(a_1)\geq U(a_2),
\qquad
I(a_1)\leq I(a_2),
\label{eq:dominance}
\end{equation}
with at least one strict inequality. Pareto filtering reduces the number of aggregate validations while keeping the ordering rule simple.

\paragraph{Feasibility note.}
Because SGAV-MP accepts an action only after evaluating the complete candidate policy on $\mathcal{X}$, every accepted policy satisfies the calibration-set constraints in Eq.~\eqref{eq:reliability_constraints}. This does not guarantee reliability on unseen test data, but it prevents the selected policy from relying only on additive independent layer scores during calibration-time selection. 

\begin{algorithm}[t]
\caption{SGAV-MP}
\label{alg:sgav}
\begin{algorithmic}[1]
\Require FP32 model $f_\theta^{\FP}$, calibration set $\mathcal{X}$, action set $\Acal$, thresholds $\varepsilon_D,\varepsilon_C,\varepsilon_R,\tau_U$
\State Evaluate $f_\theta^{\FP}$ on $\mathcal{X}$ and store reference metrics.
\ForAll{$a=(\ell,b)\in\Acal$}
\State Estimate $I_{\ell b}$, $U_{\ell b}$, and $S_{\ell b}$.
\EndFor
\State Remove dominated actions using Eq.~\eqref{eq:dominance}.
\State Initialize $q^{(0)}\leftarrow(\mathrm{FP32},\ldots,\mathrm{FP32})$, $\mathcal{R}\leftarrow\emptyset$, $t\leftarrow0$.
\While{$\Acal\setminus\mathcal{R}\neq\emptyset$}
\State Select $a^*=(\ell^*,b^*)\leftarrow\arg\max_{a\in\Acal\setminus\mathcal{R}}S_a$.
\State Form $q_c\leftarrow q^{(t)}\oplus a^*$ and evaluate $f_\theta^{q_c}$ on $\mathcal{X}$.
\If{$q_c$ satisfies Eq.~\eqref{eq:reliability_constraints}}
\State $q^{(t+1)}\leftarrow q_c$; $t\leftarrow t+1$
\State $\mathcal{R}\leftarrow\mathcal{R}\cup\{(\ell^*,b):(\ell^*,b)\in\Acal\}$
\Else
\State $\mathcal{R}\leftarrow\mathcal{R}\cup\{a^*\}$
\EndIf
\EndWhile
\State \Return final policy $q^{(t)}$
\end{algorithmic}
\end{algorithm}
\section{Experimental Protocol}
\label{sec:experiments}

\subsection{Datasets and Deployment Scenarios}
\label{subsec:datasets}

We evaluate reliability-preserving compression on two 2D and three 3D medical segmentation settings. 
For 2D, we use PROMISE12~\cite{litjens2014promise12} for prostate MRI segmentation and MosMedPlus, derived from MosMedData~\cite{morozov2020mosmeddata}, for COVID-19 lung-infection CT segmentation~\cite{shan2023coarse}. 
For 3D, we use SegResNet on AMOS22~\cite{ji2022amos} for abdominal multi-organ segmentation, BraTS~\cite{menze2015brats} for brain-tumor segmentation, and BTCV~\cite{landman2015btcv} for abdominal CT multi-organ segmentation. 
These datasets cover 2D MRI/CT and volumetric MRI/CT segmentation, allowing us to test whether compression-induced reliability changes persist across modality, anatomy, and dimensionality.

For all experiments, calibration data are disjoint from the held-out test split and are used only for quantization-range estimation, SGAV-MP scoring, and aggregate validation; all final metrics are reported on the held-out test set.
For the 2D experiments, we use the original MedSegBench partitions without modifying the split: MosMedPlus has $1910/272/547$ train/calibration/test images of size $512\times512\times3$, and PROMISE12 has $1031/147/295$ images of size $512\times512$. 
The MedSegBench validation split is used as calibration data for the main 2D results. 
For 3D, we use one fixed volume-level split for all methods and ablations: AMOS22 $168/24/48$, BraTS $1134/162/325$, and BTCV $21/3/6$ train/calibration/test volumes. 
In the default SGAV-MP setting, score computation uses $10\%$ of the calibration data for efficiency; calibration-size sensitivity using $10\%$, $25\%$, $50\%$, and $100\%$ is reported in the supplement.

We treat compression as the main objective and report deployment-related measurements separately. 
For representative 2D and 3D deployment settings, we measure backend-supported GPU execution, including latency, peak memory, model size, and compression ratio. 
For SGAV-MP policies containing mixed INT4 or INT2 actions, we report storage compression and reliability metrics without claiming runtime speedup unless directly supported by the backend.

\subsection{Models}
\label{subsec:models}

For the 2D experiments, we evaluate representative CNN-based segmentation models instantiated with five encoder families: ResNet18, ResNet50~\cite{he2016resnet}, MobileNetV2~\cite{sandler2018mobilenetv2}, EfficientNet-B0~\cite{tan2019efficientnet}, and DenseNet121~\cite{huang2017densenet}. 
These models cover residual, lightweight mobile, compound-scaled, and densely connected architectures, allowing us to test whether reliability degradation under quantization is architecture-dependent.

For the 3D experiments, we use SegResNet, a residual encoder--decoder segmentation model widely used for volumetric medical segmentation and related to the BraTS-winning residual segmentation design~\cite{myronenko2018segresnet}. For the separate 3D TensorRT deployment check, we use UNet3D on AMOS22 because the MedPTQ/TensorRT path is directly compatible with this backbone, while SegResNet deployment is limited in our setting by GroupNorm and residual-path execution issues.
The 3D study evaluates whether SGAV-MP remains effective when compression is applied to volumetric models with higher memory and inference cost.

\subsection{Baselines and Deployment}
\label{subsec:baselines}

We compare SGAV-MP with fixed-precision, PTQ, and mixed-precision baselines. 
FP32 is the reference model; FP16, INT8, and INT4 are uniform fixed-bit baselines. 
We also include AdaRound~\cite{nagel2020adaround}, Hessian Greedy, MCKP, and EfficientQ~\cite{zhang2024efficientq} to compare against calibration-based rounding, sensitivity-guided selection, constrained mixed-precision assignment, and medical PTQ. EfficientQ2/4 are uniform weight-only EfficientQ baselines at 2/4 bits, run with the released code under the same calibration/test protocol.

MedPTQ is discussed as related medical PTQ, but direct reproduction is limited by code/export/backend reproducibility in our setting and by INT4 deployment support.

Tables~\ref{tab:amos22_sgav_results} and ~\ref{tab:sgav_deployment_compact} report backend-supported TensorRT measurements for representative 3D and 2D settings. All TensorRT measurements use an NVIDIA RTX 6000 Ada Generation GPU (48 GB, compute capability sm$_{89}$) with TensorRT 10.3, supporting optimized FP32, FP16, and INT8 inference.

The 2D SGAV-MP deployment policies are restricted to FP16/INT8 execution, while mixed INT4/INT2 policies are reported as storage-compression results unless directly supported by the backend. 
For 3D deployment, we report an AMOS22 UNet3D TensorRT check because direct SegResNet reproduction with MedPTQ is limited in our setting by code/export/backend compatibility, GroupNorm and residual-path execution issues, and INT4 deployment support. 
MedPTQ also does not provide a directly usable INT2 or arbitrary mixed INT2/INT4 variant for our evaluated backbones/backend.

\subsection{Evaluation Metrics}
\label{subsec:metrics}

We evaluate compressed models using segmentation, compression, calibration, uncertainty-preservation, and failure-detection metrics. 
Segmentation quality is measured by Dice~\cite{dice1945measures,taha2015metrics}. 
Storage efficiency is reported using model size and compression ratio; for TensorRT-supported deployment variants, we also report latency and peak memory.
Calibration and probability quality are measured by ECE and Brier score~\cite{guo2017calibration,brier1950verification,mehrtash2020calibration}. 
Uncertainty preservation is measured by the rank correlation $\rho_U$ between predictive-entropy maps of the compressed and FP32 models~\cite{shannon1948mathematical,spearman1904proof}. 
Failure-detection behavior is evaluated using uncertainty-error AUROC and AURC~\cite{fawcett2006introduction,geifman2017selective}. 
SGAV-MP uses Dice, ECE, $\rho_U$, and AURC as aggregate acceptance constraints, while error-AUROC is reported as an additional diagnostic.

\noindent
Additional results and analyses are provided in the supplementary material due to page constraints.

\section{Results}
\label{sec:results}

\subsection{2D Reliability-Preserving Compression}
\label{subsec:results_2d}

Tables~\ref{tab:mri_promise_results} and~\ref{tab:ct_mosmed_results} report representative 2D ResNet50 results, with full five-backbone results in the supplement. 
On PROMISE12/ResNet50, SGAV-MP gives $7.29\times$ compression with Dice $0.8858$ and $\rho_U=0.9697$, compared with INT8 at $3.98\times$ and $\rho_U=0.9441$. 
On MosMedPlus/ResNet50, SGAV-MP gives $5.94\times$ compression with Dice $0.7848$ and $\rho_U=0.9392$, while uniform INT4 drops to Dice $0.6507$ and $\rho_U=0.1192$. 
EfficientQ gives higher nominal compression in some cases, but lower uncertainty agreement: EfficientQ4 gives $\rho_U=0.7282$ on PROMISE12/ResNet50 and $\rho_U=0.5303$ on MosMedPlus/ResNet50.

We report $\rho_U$ because absolute uncertainty magnitudes are small and model-dependent; representative MosMedPlus cross-calibration runs have average absolute uncertainty around $0.011$--$0.014$. 
Thus, $\rho_U$ captures preservation of FP32 uncertainty-map structure, while ECE, Brier, AURC, and error-AUROC report absolute reliability behavior.

Figure~\ref{fig:pareto_compression_uncertainty} shows SGAV-MP policies obtained by varying uncertainty tolerance. 
On MosMedPlus/ResNet50, policies range from $4.89\times$ compression with $\rho_U=0.970$ to $5.91$--$5.92\times$ with $\rho_U\approx0.94$--$0.95$. 
On DenseNet121, policies range from $4.86\times$ with $\rho_U=0.968$ to $6.23$--$6.25\times$ with $\rho_U\approx0.947$. 
Selected policies quantize about $62$--$129$ layers; build time, peak memory, repeated-run stability, calibration sweeps, paired-layer analysis, and full five-backbone results are provided in the supplement.

\begin{figure}[tb!]
\centering
\includegraphics[width=0.85\linewidth]{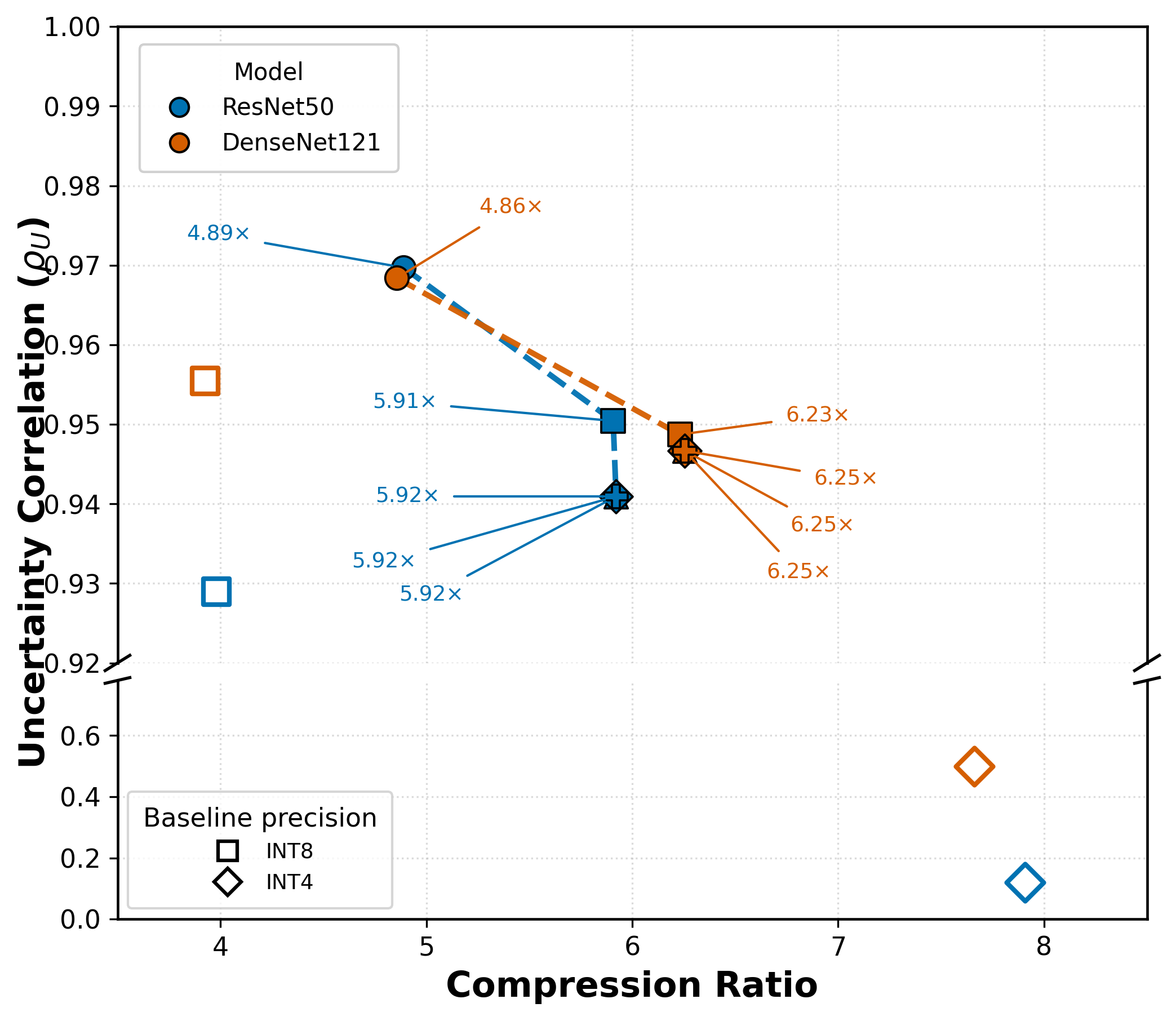}
\caption{
Compression--uncertainty Pareto analysis on MosMedPlus. 
INT8/INT4 are fixed points; SGAV-MP varies uncertainty tolerance with full-model validation. 
SGAV-MP improves compression beyond INT8 while preserving uncertainty correlation, whereas INT4 degrades it. 
The broken y-axis shows both INT4 and the SGAV-MP frontier.
}
\label{fig:pareto_compression_uncertainty}
\end{figure}

\begin{table*}[tb!]
\centering
\caption{
MRI segmentation results on PROMISE12/ResNet50. 
FP32 is the reference and is not ranked. 
Best compressed value is bold-underlined and second best is bold. 
 For SGAV-MP, Layers / Build / Peak GPU are $62$ / $2450$ s / $19915$ MB for ResNet50.
}
\label{tab:mri_promise_results}
\scriptsize
\setlength{\tabcolsep}{4pt}
\resizebox{\textwidth}{!}{
\begin{tabular}{llccccccc}
\toprule
\textbf{Model} & \textbf{Method} &
\textbf{Dice $\uparrow$} &
\textbf{ECE $\downarrow$} &
\textbf{Brier $\downarrow$} &
\textbf{AURC $\downarrow$} &
\textbf{Error AUROC $\uparrow$} &
\textbf{UQ Corr. $\uparrow$} &
\textbf{Comp. Ratio $\uparrow$} \\
\midrule

\multirow{9}{*}{ResNet50}
& FP32           & 0.8857 & 0.1025 & 0.2052 & 0.00539 & 0.5881 & 1.0000 & 1.00 \\
& INT8           & 0.8855 & 0.1032 & 0.2067 & 0.00534 & 0.5879 & 0.9441 & 3.98 \\
& INT4           & 0.8810 & 0.1144 & 0.2293 & \second{0.00524} & 0.5848 & 0.3240 & \second{7.91} \\
& MCKP           & 0.8857 & \second{0.1026} & 0.2056 & 0.00540 & 0.5880 & \best{0.9945} & 3.96 \\
& Hessian Greedy & 0.8838 & 0.1101 & 0.2206 & 0.00542 & 0.5835 & 0.4547 & 7.90 \\
& AdaRound       & \best{0.8864} & 0.1027 & \second{0.2053} & 0.00542 & 0.5869 & 0.7906 & 3.88 \\
& EfficientQ2    & 0.8829 & 0.1060 & 0.2125 & \best{0.00496} & \best{0.6046} & 0.2193 & \best{15.60} \\
& EfficientQ4    & \second{0.8859} & \best{0.1021} & \best{0.2044} & 0.00525 & \second{0.5898} & 0.7282 & \second{7.91} \\
& SGAV-MP        & 0.8858 & \second{0.1026} & 0.2055 & 0.00539 & 0.5883 & \second{0.9697} & 7.29 \\

\bottomrule
\end{tabular}
}
\end{table*}

\begin{table*}[tb!]
\centering
\caption{
CT segmentation results on MosMedPlus/ResNet50. 
FP32 is the reference and is not ranked. 
Best compressed value is bold-underlined and second best is bold. 
For SGAV-MP, Layers / Build / Peak GPU are $62$ / $510.1$ s / $4719.9$ MB.
}
\label{tab:ct_mosmed_results}
\scriptsize
\setlength{\tabcolsep}{4pt}
\resizebox{\textwidth}{!}{
\begin{tabular}{llccccccc}
\toprule
\textbf{Model} & \textbf{Method} &
\textbf{Dice $\uparrow$} &
\textbf{ECE $\downarrow$} &
\textbf{Brier $\downarrow$} &
\textbf{AURC $\downarrow$} &
\textbf{Error AUROC $\uparrow$} &
\textbf{UQ Corr. $\uparrow$} &
\textbf{Comp. Ratio $\uparrow$} \\
\midrule

\multirow{9}{*}{ResNet50}
& FP32           & 0.7850 & 0.2064 & 0.4123 & 0.00464 & 0.5829 & 1.0000 & 1.00 \\
& INT8           & \second{0.7849} & 0.2059 & 0.4114 & 0.00465 & \second{0.5846} & 0.9290 & 3.98 \\
& INT4           & 0.6507 & 0.3653 & 0.7304 & 0.00653 & 0.5458 & 0.1192 & \second{7.91} \\
& MCKP           & \best{0.7850} & 0.2067 & 0.4128 & \second{0.00464} & 0.5829 & \best{0.9936} & 3.97 \\
& Hessian Greedy & 0.7717 & 0.2137 & 0.4270 & 0.00491 & 0.5784 & 0.2549 & 7.89 \\
& AdaRound       & 0.7835 & 0.2089 & 0.4163 & \best{0.00463} & 0.5845 & 0.7667 & 3.96 \\
& EfficientQ2    & 0.6765 & \best{0.1569} & \best{0.3134} & 0.02215 & 0.5801 & 0.2312 & \best{15.60} \\
& EfficientQ4    & 0.7823 & \second{0.2008} & \second{0.4011} & 0.00478 & \best{0.5851} & 0.5303 & \second{7.91} \\
& SGAV-MP        & 0.7848 & 0.2075 & 0.4144 & \second{0.00464} & 0.5825 & \second{0.9392} & 5.94 \\

\bottomrule
\end{tabular}
}
\end{table*}

\begin{table*}[tb!]
\centering
\caption{
3D SegResNet results on AMOS22 and BraTS. 
FP32 is the reference and is not ranked. 
Best compressed value is bold-underlined and second best is bold. 
 For SGAV-MP, build time / peak GPU are $8875$ s / $5981$ MB on AMOS22 and $28252.6$ s / $3430.9$ MB on BraTS.
}
\label{tab:segresnet_3d_results}
\scriptsize
\setlength{\tabcolsep}{4pt}
\resizebox{\textwidth}{!}{
\begin{tabular}{llccccccc}
\toprule
\textbf{Dataset} & \textbf{Precision} &
\textbf{Dice $\uparrow$} &
\textbf{ECE $\downarrow$} &
\textbf{Brier $\downarrow$} &
\textbf{AURC $\downarrow$} &
\textbf{Err-AUROC $\uparrow$} &
\textbf{UC $\uparrow$} &
\textbf{CR $\uparrow$} \\
\midrule

\multirow{5}{*}{AMOS22}
& FP32    & 0.7778 & 0.0527 & 0.1755 & 0.00043 & 0.8719 & 1.0000 & 1.00 \\
& FP16    & \second{0.7778} & 0.0527 & 0.1755 & \second{0.00043} & 0.8719 & \best{1.0000} & 2.00 \\
& INT8    & \best{0.7780} & \second{0.0525} & \second{0.1750} & \best{0.00042} & \second{0.8722} & \second{0.9997} & 4.00 \\
& INT4    & 0.7450 & 0.0703 & 0.2109 & 0.00054 & 0.8603 & 0.9121 & \best{7.98} \\
& SGAV-MP & 0.7717 & \best{0.0478} & \best{0.1679} & 0.00053 & \best{0.8806} & 0.9857 & \second{7.72} \\

\midrule

\multirow{5}{*}{BraTS}
& FP32    & 0.7334 & 0.0866 & 0.2095 & 0.00013 & 0.8723 & 1.0000 & 1.00 \\
& FP16    & \second{0.7334} & 0.0866 & 0.2095 & \best{0.00013} & 0.8723 & \best{1.0000} & 2.00 \\
& INT8    & \best{0.7335} & \best{0.0854} & \best{0.2063} & \second{0.00015} & \second{0.8751} & \second{0.9986} & 4.00 \\
& INT4    & 0.6458 & 0.1196 & 0.2769 & 0.00039 & 0.8471 & 0.5348 & \best{7.98} \\
& SGAV-MP & 0.7293 & \second{0.0858} & \second{0.2068} & 0.00017 & \best{0.8757} & 0.9755 & \second{7.82} \\

\bottomrule
\end{tabular}
}
\end{table*}

\subsection{3D Results}
\label{subsec:results_3d}

Table~\ref{tab:segresnet_3d_results} evaluates SGAV-MP on 3D SegResNet across AMOS22 and BraTS. 
The trend is consistent with 2D results: INT8 is a strong fixed baseline, while SGAV-MP approaches low-bit compression in some settings with explicit reliability validation. 
On AMOS22, SGAV-MP reaches $7.72\times$ compression with $\rho_U=0.9857$ and improves ECE/Brier relative to FP32. 
On BraTS, it obtains $7.82\times$ compression with Dice $0.7293$ and $\rho_U=0.9755$, while INT4 degrades Dice and uncertainty agreement. 
BTCV results, including the restricted no-INT2 variant, are reported in the supplementary material.

Uniform INT4 gives high nominal compression but is unreliable in volumetric segmentation; on BraTS, $\rho_U$ drops to $0.5348$ and Dice from $0.7334$ to $0.6458$. 
Thus, storage reduction alone is insufficient without preserving calibration, uncertainty maps, and risk--coverage behavior. 
The selected 3D SGAV-MP runs require build times of about $8875$--$28252$ s and peak GPU memory of $3.4$--$6.0$ GB. 
Overall, these results support aggregate-validated mixed precision for compressing volumetric segmenters while avoiding the reliability collapse of aggressive uniform quantization.

\begin{table*}[tbh!]
\centering
\caption{
Ablation study on PROMISE12/ResNet50. 
$D$, $U$, $R$, and $C$ denote Dice, uncertainty, AURC, and ECE degradation terms. 
NC removes the connection term. The ablation reruns policy selection under fixed score-term settings. Best compressed value is bold-underlined and second best is bold.
}
\label{tab:ablation_promise12}
\scriptsize
\setlength{\tabcolsep}{2.5pt}
\resizebox{\textwidth}{!}{
\begin{tabular}{lccccc|cccccccccc}
\toprule
\textbf{Variant} &
$\boldsymbol{\alpha_D}$ & $\boldsymbol{\alpha_U}$ & $\boldsymbol{\alpha_R}$ & $\boldsymbol{\alpha_C}$ & $\boldsymbol{\lambda_M}$ &
\textbf{Dice $\uparrow$} & \textbf{ECE $\downarrow$} & \textbf{Brier $\downarrow$} & \textbf{AURC $\downarrow$} &
\textbf{Err-AUROC $\uparrow$} & \textbf{U-Corr. $\uparrow$} & \textbf{CR $\uparrow$} &
\textbf{LQ} & \textbf{Time(s)} & \textbf{Mem(MB)} \\
\midrule
Dice              & 1 & 0 & 0 & 0 & 1 & \best{0.8857} & \best{0.1026} & \best{0.2055} & 0.0054 & \second{0.5883} & \second{0.9608} & 5.07 & 32 & 2444.6 & 19914.2 \\
Dice only (NC)    & 1 & 0 & 0 & 0 & 0 & \best{0.8857} & \best{0.1026} & \best{0.2055} & 0.0054 & \second{0.5883} & \second{0.9608} & 5.07 & 32 & 2440.4 & 19914.2 \\
Dice+UQ+risk      & 1 & 1 & 1 & 0 & 1 & \second{0.8856} & 0.1029 & 0.2061 & 0.0054 & \best{0.5884} & 0.9531 & \best{7.29} & 62 & 2601.2 & 19912.4 \\
UQ                & 0 & 1 & 0 & 0 & 1 & 0.8853 & 0.1037 & 0.2077 & 0.0054 & 0.5882 & \second{0.9545} & \second{7.17} & 63 & 2455.5 & 19913.9 \\
UQ only (NC)      & 0 & 1 & 0 & 0 & 0 & 0.8853 & 0.1037 & 0.2077 & 0.0054 & 0.5882 & \second{0.9545} & \second{7.17} & 63 & 2442.6 & 19913.9 \\
Risk+ECE          & 0 & 0 & 1 & 1 & 1 & \best{0.8857} & \best{0.1026} & \best{0.2055} & 0.0054 & \second{0.5883} & \second{0.9608} & 5.07 & 32 & 2443.6 & 19914.7 \\
Risk+ECE (NC)     & 0 & 0 & 1 & 1 & 0 & \best{0.8857} & \best{0.1026} & \best{0.2055} & 0.0054 & \second{0.5883} & \second{0.9608} & 5.07 & 32 & 2536.2 & 19914.7 \\
Dice+risk+ECE     & 1 & 0 & 1 & 1 & 1 & \best{0.8857} & \best{0.1026} & \best{0.2055} & 0.0054 & \second{0.5883} & \second{0.9608} & 5.07 & 32 & 2454.8 & 19913.4 \\
Full SGAV-MP (NC) & 1 & 1 & 1 & 1 & 0 & \second{0.8856} & \second{0.1027} & \second{0.2057} & 0.0054 & \second{0.5883} & 0.9500 & 7.16 & 55 & 2692.5 & 19914.7 \\
Full SGAV-MP      & 1 & 1 & 1 & 1 & 1 & \second{0.8856} & \best{0.1026} & \best{0.2055} & 0.0054 & \second{0.5883} & \best{0.9697} & \best{7.29} & 55 & 2450.7 & 19914.7 \\
\bottomrule
\end{tabular}
}
\end{table*}

\subsection{Ablation and Cross-Dataset Calibration}
\label{subsec:ablation_drift}

Table~\ref{tab:ablation_promise12} reports the ablation on PROMISE12/ResNet50. Full SGAV-MP reaches $7.29\times$ compression with $\rho_U=0.9697$, compared with $7.17\times$/$0.9545$ for UQ-only. Removing the connection-aware term reduces the full policy to $7.16\times$/$0.9500$, while Dice+UQ+risk reaches the same $7.29\times$ compression but lower $\rho_U=0.9531$. Thus, uncertainty-aware scoring drives much of the compression gain, while the complete score provides the strongest reliability-preserving operating point in this ablation.

Cross-dataset calibration is reported in Table~\ref{tab:sgavmp_mosmed_models}, where SGAV-MP is calibrated on PROMISE12 and evaluated on MosMedPlus. 
Despite this shift, Dice remains close to the MosMedPlus FP32/SGAV-MP range, $\rho_U$ stays high ($0.942$--$0.968$), and compression remains $4.06$--$5.73\times$. 
This suggests limited overfitting to one calibration split, although stronger domain shift still depends on FP32 reliability and calibration representativeness. 
Additional calibration-size sweeps, paired-layer analysis, and qualitative examples are provided in the supplement.
\subsection{Deployment-Related Measurements}
\label{subsec:deployment_results}

Table~\ref{tab:amos22_sgav_results} gives a representative 3D AMOS22 TensorRT deployment check using UNet3D. 
MedPTQ provides the largest backend-supported compression and lowest latency ($3.90\times$, $0.843$ ms), while SGAV-MP gives $1.58\times$ compression and preserves Dice ($0.7274$) and uncertainty agreement ($\rho_U=0.9972$) with slightly better ECE/Brier. 
This supports our deployment framing: runtime efficiency depends on backend support, while SGAV-MP focuses on reliability-preserving compression.

Table~\ref{tab:sgav_deployment_compact} reports TensorRT measurements for representative 2D models using FP32, uniform INT8, and a TensorRT-compatible SGAV-MP policy. 
The SGAV-MP deployment policy is restricted to FP16/INT8 execution, while INT4/INT2 policies are reported separately as storage-compression results. 
Uniform INT8 gives larger weight compression, while SGAV-MP gives intermediate compression with similar Dice and higher uncertainty agreement. 
On MosMedPlus/EfficientNet-B0, SGAV-MP gives $2.25\times$ compression, $0.791$ ms latency, and $\rho_U=0.9833$, compared with uniform INT8 at $3.81\times$, $0.903$ ms, and $\rho_U=0.9600$. 
On PROMISE12/ResNet50, SGAV-MP gives $2.53\times$, $0.594$ ms, and $\rho_U=0.9698$, compared with uniform INT8 at $4.08\times$, $0.659$ ms, and $\rho_U=0.9426$.

\begin{table*}[tb!]
\centering
\caption{SGAV-MP is calibrated on PROMISE12 and evaluated on MosMedPlus to assess sensitivity to calibration-domain shift}
\label{tab:sgavmp_mosmed_models}
\scriptsize
\setlength{\tabcolsep}{4pt}
\resizebox{\textwidth}{!}{
\begin{tabular}{lccccccccccc}
\toprule
\textbf{Model} &
\textbf{Dice $\uparrow$} &
\textbf{ECE $\downarrow$} &
\textbf{Brier $\downarrow$} &
\textbf{AURC $\downarrow$} &
\textbf{Error AUROC $\uparrow$} &
\textbf{UQ Corr. $\uparrow$} &
\textbf{Abs. Unc.} &
\textbf{Comp. Ratio $\uparrow$} &
\textbf{Layers Quant.} &
\textbf{Build Time (s)} &
\textbf{Peak GPU (MB)} \\
\midrule

ResNet18
& 0.779 & 0.212 & 0.423 & 0.004 & 0.586 & 0.954 & 0.014 & 4.153 & 29 & 1261.513 & 12775.755 \\

MobileNetV2
& 0.784 & 0.199 & 0.398 & 0.004 & 0.603 & 0.942 & 0.012 & 4.065 & 43 & 2441.863 & 10874.005 \\

EfficientNet-B0
& 0.788 & 0.192 & 0.383 & 0.004 & 0.606 & 0.952 & 0.012 & 5.734 & 77 & 3614.041 & 11168.406 \\

DenseNet121
& 0.791 & 0.196 & 0.392 & 0.004 & 0.620 & 0.968 & 0.013 & 4.545 & 120 & 5355.454 & 15615.749 \\

ResNet50
& 0.785 & 0.206 & 0.411 & 0.005 & 0.583 & 0.948 & 0.011 & 5.454 & 61 & 2537.617 & 16074.700 \\

\bottomrule
\end{tabular}
}
\end{table*}

\begin{table*}[tb!]
\centering
\caption{
Representative 3D AMOS22 TensorRT deployment measurements using the UNet3D. 
FP32 is the backend reference and is not ranked. 
Best compressed value is bold-underlined and second best is bold. 
Build time / peak GPU are $35.7$ s / $152$ MB for FP16, $19.7$ s / $170$ MB for MedPTQ, and $30.7$ s / $224$ MB for SGAV-MP.
}
\label{tab:amos22_sgav_results}
\scriptsize
\setlength{\tabcolsep}{4pt}
\resizebox{\textwidth}{!}{
\begin{tabular}{lccccccccc}
\toprule
\textbf{Method} &
\textbf{Dice $\uparrow$} &
\textbf{ECE $\downarrow$} &
\textbf{Brier $\downarrow$} &
\textbf{AURC $\downarrow$} &
\textbf{Err-AUROC $\uparrow$} &
\textbf{U-Corr. $\uparrow$} &
\textbf{Lat.(ms) $\downarrow$} &
\textbf{Mem.(MB) $\downarrow$} &
\textbf{CR $\uparrow$} \\
\midrule
FP32
& 0.7273 & 0.0658 & 0.2039 & 0.0007 & 0.8828 & 1.0000 & 2.186 & 141.50 & 1.00 \\

FP16
& \best{0.7274} & \second{0.0657} & \second{0.2036} & 0.0007 & \best{0.8828} & \best{1.0000} & \second{0.862} & \best{70.76} & \second{2.00} \\

MedPTQ
& 0.7270 & 0.0661 & 0.2046 & 0.0007 & 0.8823 & 0.9946 & \best{0.843} & \second{88.95} & \best{3.90} \\

SGAV-MP
& \best{0.7274} & \best{0.0654} & \best{0.2033} & 0.0007 & \second{0.8825} & \second{0.9972} & 1.580 & 136.95 & 1.58 \\
\bottomrule
\end{tabular}
}
\end{table*}

\begin{table*}[tb!]
\centering
\caption{
Representative 2D TensorRT deployment results. 
CR denotes compression ratio, UC denotes uncertainty-map correlation, and Con-Mem denotes TensorRT context memory. 
FP32 is the reference and is not ranked. 
Best value is bold-underlined and second best is bold.
}
\label{tab:sgav_deployment_compact}
\scriptsize
\setlength{\tabcolsep}{3pt}
\resizebox{\textwidth}{!}{
\begin{tabular}{lllccccccccccccc}
\toprule
\textbf{Dataset} &
\textbf{Model} &
\textbf{Precision} &
\textbf{CR $\uparrow$} &
\textbf{Lat. $\downarrow$} &
\textbf{Con-Mem $\downarrow$} &
\textbf{Dice $\uparrow$} &
\textbf{ECE $\downarrow$} &
\textbf{Brier $\downarrow$} &
\textbf{AURC $\downarrow$} &
\textbf{Err-AUROC $\uparrow$} &
\textbf{UC $\uparrow$} &
\textbf{Weights $\downarrow$} &
\textbf{Total Mem. $\downarrow$} \\
\midrule

\multirow{3}{*}{MosMedPlus}
& \multirow{3}{*}{EfficientNet-B0}
& FP32
& 1.000 & 1.331  & 94.563
& 0.7882 & 0.1915 & 0.3825 & 0.0041 & 0.6067 & 1
& 35.04 & 129.60 \\

& & SGAV-MP
& \second{2.250} & \best{0.791}  & \best{34.718}
& \best{0.7882} & \best{0.1914} & \best{0.3823} & 0.0041 & \best{0.6067} & \best{0.9833}
& \second{15.58} & \second{50.30} \\

& & Uniform INT8
& \best{3.812} & \second{0.903}  & \second{37.125}
& \second{0.7881} & \second{0.1917} & \second{0.3828} & 0.0041 & \second{0.6065} & \second{0.9600}
& \best{9.20} & \best{46.33} \\

\midrule

\multirow{3}{*}{PROMISE12}

& \multirow{3}{*}{ResNet50}
& FP32
& 1.000 & 0.982  & 88.438
& 0.8857 & 0.1026 & 0.2055 & 0.0023 & 0.5883 & 1
& 94.11 & 182.55 \\

& & SGAV-MP
& \second{2.527} & \best{0.594}  & \best{37.125}
& 0.8857 & 0.1026 & 0.2055 & 0.0023 & 0.5883 & \best{0.9698}
& \second{37.25} & \second{74.18} \\

& & Uniform INT8
& \best{4.082} & \second{0.659}  & \second{40.813}
& 0.8857 & 0.1026 & 0.2055 & 0.0023 & 0.5883 & \second{0.9426}
& \best{23.06} & \best{65.81} \\

\bottomrule
\end{tabular}
}
\end{table*}
\section{Discussion}
\label{sec:discussion}

SGAV-MP shows that medical segmentation compression should be evaluated beyond Dice and model size. 
Across 2D and 3D experiments, fixed low-bit baselines can preserve usable Dice while changing calibration, uncertainty maps, or risk--coverage behavior. 
This motivates aggregate validation: layer scores rank candidate precision changes, but acceptance is decided by evaluating the complete compressed model on a calibration set. 
The Pareto analysis also shows that reliability tolerance defines a compression--uncertainty frontier rather than a single fixed INT8 or INT4 point.

SGAV-MP is compression-oriented and does not imply automatic runtime speedup. 
We separate storage compression from backend-supported measurements: TensorRT FP32/FP16/INT8 results are reported for representative 2D and 3D settings, while mixed INT4/INT2 policies are reported only by storage compression and reliability metrics unless directly supported by the backend. 
Limitations include dependence on the FP32 model's reliability, use of labeled calibration data, sensitivity to calibration representativeness, nonzero search cost, and the need for broader evaluation on foundation-style segmenters, additional modalities, and real mixed INT4/INT2 deployment backends.

\section{Conclusion}
\label{sec:conclusion}

We studied reliability-preserving post-training compression for medical image segmentation. 
Dice-preserving quantization can still degrade uncertainty, calibration, and failure-detection behavior, motivating SGAV-MP, a score-guided mixed-precision method with aggregate full-model reliability validation. 
Across 2D MRI/CT and 3D volumetric benchmarks, SGAV-MP preserves FP32-like segmentation and uncertainty behavior while achieving storage compression. 
Overall, compressed medical segmenters should be evaluated using reliability signals in addition to accuracy and compression ratio.

{
    \small
    \bibliographystyle{ieeenat_fullname}
    \bibliography{main}
}

\end{document}


\maketitle

This supplementary material provides the complete quantitative results, ablation results, additional qualitative examples, and implementation details that could not be included in the main paper due to page constraints.


\section{Implementation Details}
\label{sec:supp_implementation}

SGAV-MP uses fixed stopping criteria across all datasets and models. During the mixed-precision search, a candidate quantization action is accepted only if the complete compressed model satisfies the calibration-set reliability constraints. As illustrated in Figure~\hyperref[fig:subadditive_layer_effects]{\ref*{fig:subadditive_layer_effects}}, the measured joint uncertainty degradation can be lower than the summed layer-wise degradation, reflecting sub-additive overlapping layer effects. The allowed degradations are $3$ percentage points (pp) for Dice uncertainty-map correlation and AURC. Equivalently, an accepted candidate model must satisfy
\[
\Dice(f^q) \geq \Dice(f^{\FP}) - 0.03,
\enspace
\rho_U(f^q,f^{\FP}) \geq \rho_U(f^{\FP},f^{\FP}) - 0.03,
\enspace
\AURC(f^q) \leq \AURC(f^{\FP}) + 0.03 .
\]
Since $\rho_U(f^{\FP},f^{\FP})=1$, the uncertainty-map correlation constraint is equivalent to $\rho_U(f^q,f^{\FP}) \geq 0.97$. These tolerances can be varied to control the compression--reliability trade-off, as examined in the tolerance-sensitivity results discussed in section~\ref{sec:tol_sen}. The quantization loop stops when no remaining candidate action satisfies all three criteria.

\begin{figure}[t]
\centering
\includegraphics[width=0.4\columnwidth]{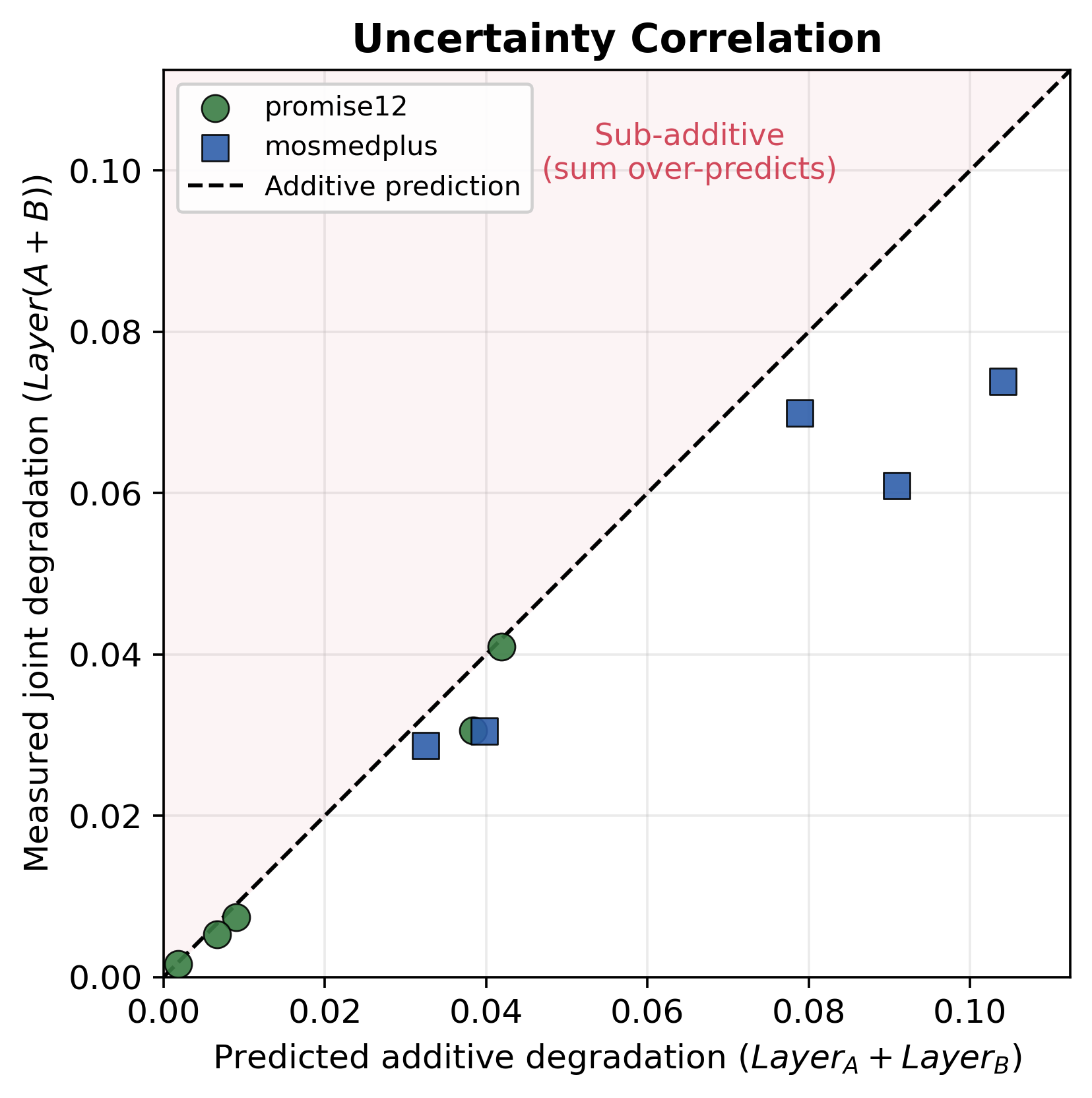}
\caption{
Paired-layer uncertainty degradation. Points below the diagonal indicate that the measured joint degradation is lower than the summed layer-wise degradation, demonstrating sub-additive overlapping layer effects.
}
\label{fig:subadditive_layer_effects}
\end{figure}

\section{Additional Quantitative Results}

\label{sec:supp_quantitative_results}

\subsection{Full Segmentation and Compression Results}

Table~\hyperref[tab:supp_3d_full]{\ref*{tab:supp_3d_full}} reports the 3D SegResNet results on AMOS22, BraTS, and BTCV. FP16 and INT8 remain close to FP32 across the reported metrics, whereas uniform INT4 provides higher compression at the cost of larger degradation. On AMOS22, INT4 increases compression from $1.0$ to $7.98$, while Dice drops from $0.7778$ to $0.7450$ and UQ correlation from $1.0000$ to $0.9121$; on BraTS, UQ correlation drops to $0.5348$. In contrast, SGAV-MP achieves compression ratios of $5.48$, $7.76$, and $4.73$ on AMOS22, BraTS, and BTCV, respectively, while retaining UQ correlations of $0.9962$, $0.9864$, and $0.9984$. SGAV-MP-2 further increases compression on AMOS22 and BraTS, with UQ correlation decreasing to $0.9857$ and $0.9755$. Overall, SGAV-MP provides a higher-compression alternative to uniform INT8 while preserving the reported uncertainty measures more closely than uniform INT4.

Table~\hyperref[tab:supp_promise_full]{\ref*{tab:supp_promise_full}} compares SGAV-MP with FP32, uniform INT8/INT4, MCKP, Hessian Greedy, and AdaRound across five encoders on PROMISE12 from the MedSegBench Database. SGAV-MP-4 maintains Dice close to FP32 while achieving $3.32$--$4.56\times$ compression and UQ correlation near the $0.97$ acceptance threshold. Allowing INT2 through SGAV-MP-2 increases compression, e.g., from $4.11$ to $6.51$ for ResNet18 and from $4.56$ to $7.02$ for ResNet50, with UQ correlation decreasing to $0.9707$ and $0.9698$. The relaxed SGAV-MP-2R further increases compression to $7.77\times$ and $7.29\times$, with UQ correlations of $0.9260$ and $0.9697$, respectively. Thus, the SGAV-MP variants expose a compression--uncertainty trade-off: relaxing the constraints enables higher compression while producing progressively larger changes in UQ correlation.

Table~\hyperref[tab:supp_mosmed_full]{\ref*{tab:supp_mosmed_full}} reports the corresponding CT segmentation results on MosMedPlus. Uniform INT8 produces compression ratios close to $4\times$ while generally retaining Dice near FP32, but the associated uncertainty correlation varies across architectures. The degradation becomes substantially larger for INT4 in several models: for example, MobileNetV2 decreases from a Dice of $0.7845$ to $0.6322$ and its UQ correlation decreases from $1.0000$ to $0.1261$, while ResNet50 decreases from $0.7850$ to $0.6507$ with UQ correlation of $0.1192$. SGAV-MP-4 is more conservative on this dataset, with compression ratios ranging from $2.75$ to $3.77$ and UQ correlations between $0.9845$ and $0.9853$. The less constrained SGAV-MP-2 configurations increase compression for all five architectures, reaching $4.89\times$ for ResNet50 and $4.86\times$ for DenseNet121, while UQ correlation decreases to approximately $0.97$. The relaxed SGAV-MP-2R configurations provide additional compression, with the largest reported value of $6.21\times$ for DenseNet121, while UQ correlation decreases further. Importantly, Dice remains numerically very close to the FP32 value for the SGAV variants even when compression is increased. Thus, the results indicate that the effect of increasing compression is more visible in the uncertainty-correlation metric than in Dice for the SGAV configurations reported in this table.

\subsection{Acceptance-Tolerance Sensitivity}
\label{sec:tol_sen}

Tables~\hyperref[tab:rho_tolerance_promise12]{\ref*{tab:rho_tolerance_promise12}} and~\hyperref[tab:rho_tolerance_mosmedplus]{\ref*{tab:rho_tolerance_mosmedplus}} evaluate the effect of varying the uncertainty-map correlation tolerance $\rho_{\mathrm{tol}}$ in SGAV-MP-2, alongside fixed-precision FP32, INT8, INT4, and INT2 references. On PROMISE12, increasing $\rho_{\mathrm{tol}}$ generally allows higher compression for several architectures while reducing UQ correlation; for example, ResNet18 increases from $6.5016\times$ to $7.4655\times$ compression as $\rho_{\mathrm{tol}}$ increases from $0.03$ to $0.05$, while UQ correlation decreases from $0.9699$ to $0.9556$. DenseNet121 shows greater sensitivity, with UQ correlation decreasing from $0.9693$ to $0.8767$ across the sweep. On MosMedPlus, a similar trend is observed for MobileNetV2 and EfficientNet-B0, where increasing $\rho_{\mathrm{tol}}$ from $0.03$ to $0.15$ reduces UQ correlation from $0.9695$ to $0.8806$ and from $0.9682$ to $0.8757$, respectively, while compression increases. In contrast, Dice remains comparatively stable across the SGAV-MP-2 configurations on both datasets. These results directly illustrate the role of the stopping criterion defined above: relaxing $\rho_{\mathrm{tol}}$ permits candidates with greater uncertainty degradation to be accepted, thereby exposing a controllable compression--uncertainty trade-off whose magnitude depends on the encoder and dataset.

Table~\hyperref[tab:tolerance_ablation]{\ref*{tab:tolerance_ablation}} examines the effect of varying the uncertainty-correlation tolerance $\tau_{\rho_U}$ and risk tolerance $\tau_R$ on MosMedPlus with MobileNetV2. The results show that changing $\tau_R$ while keeping $\tau_{\rho_U}$ fixed produces identical values for all reported metrics and the compression ratio. For example, at $\tau_{\rho_U}=0.015$, all three tested values of $\tau_R$ produce a compression ratio of $2.82$, UQ correlation of $0.9851$, and the same Dice, ECE, Brier, AURC, and Error-AUROC values. In contrast, increasing $\tau_{\rho_U}$ from $0.015$ to $0.020$ increases the compression ratio from $2.82$ to $3.11$, while UQ correlation decreases from $0.9851$ to $0.9800$. Increasing it further to $0.030$ raises the compression ratio to $3.73$ and decreases UQ correlation to $0.9695$. Dice changes only from $0.7844$ to $0.7843$, while the remaining reported metrics also exhibit only small numerical changes. Hence, within the tested settings, the uncertainty-correlation tolerance is the parameter associated with the observed changes in the selected compression policy, whereas varying the risk tolerance does not alter the resulting policy or metrics.

\subsection{Calibration-Data Sensitivity}

Table~\hyperref[tab:sgav_calibration_sweep_mosmed]{\ref*{tab:sgav_calibration_sweep_mosmed}} studies the effect of calibration-set size on SGAV-MP-2 across five encoders on MosMedPlus. The reported segmentation and uncertainty metrics remain largely stable from $10\%$ to $100\%$ calibration, while build time and peak GPU memory increase substantially. For example, DenseNet121 shows Dice variation of only $0.7915$--$0.7913$ and UQ correlation of $0.9678$--$0.9689$, whereas build time increases from $1056.5$~s to $8628.2$~s. Similar metric stability is observed across the other encoders, with no consistent improvement from using more calibration data. Thus, within the evaluated settings, smaller calibration fractions achieve metrics comparable to full calibration at substantially lower search cost.

Table~\hyperref[tab:sgav_calibration_sweep_promise12]{\ref*{tab:sgav_calibration_sweep_promise12}} reports the corresponding sweep on PROMISE12. Dice remains stable across calibration fractions, ranging from $0.8943$--$0.8962$ for DenseNet121 and $0.8946$--$0.8949$ for MobileNetV2, while UQ correlation shows similarly small variations. In contrast, calibration cost increases substantially; for ResNet50, build time rises from $311.6$~s at $10\%$ calibration to $2448.0$~s at $100\%$, with peak GPU memory increasing from $2856.2$~MB to $19912.4$~MB. Overall, increasing the calibration fraction does not yield a consistent improvement in the reported metrics, but substantially increases the computational cost.

Table~\hyperref[tab:sgav_3d_calibration_sweep]{\ref*{tab:sgav_3d_calibration_sweep}} extends the calibration-data analysis to 3D datasets with SGAV-MP-2 pipeline on AMOS22, BraTS, and BTCV. On AMOS22, increasing calibration from $10\%$ to $100\%$ changes Dice from $0.7748$ to $0.7697$ and the compression ratio from $7.34$ to $7.72$, while UQ correlation changes from $0.9754$ to $0.9857$. BraTS exhibits even smaller Dice variation between the $10\%$ and $50\%$ settings, followed by a lower value at $75\%$ and a return to $0.7293$ at $100\%$. BTCV shows a similar non-monotonic pattern, with Dice ranging from $0.6022$ to $0.6067$. In contrast to the relatively small changes in the evaluation metrics, build time increases markedly with calibration size. For AMOS22, build time increases from $3001.3$ s at $10\%$ to $31755.7$ s at $100\%$, while BraTS increases from $3979.7$ s to $40445.6$ s. Peak GPU memory also increases, although the magnitude differs across datasets. The results therefore do not show a monotonic improvement in segmentation or uncertainty metrics with additional calibration data. Instead, the primary consistent effect of increasing the calibration fraction in this sweep is increased calibration/search cost.

\subsection{Scoring-Term Ablation}

Table~\hyperref[tab:ablation]{\ref*{tab:ablation}} evaluates the contribution of the individual SGAV-MP-2 scoring terms on MosMedPlus with ResNet18 and PROMISE12 with ResNet50. On MosMedPlus, using Dice alone produces a compression ratio of only $1.37$, while adding UQ and risk terms increases the ratio to $4.29$ with UQ correlation increasing from $0.9497$ to $0.9620$. The UQ-only configuration also reaches a compression ratio of $4.15$, whereas the Risk+ECE configuration remains much more conservative at $1.28$. The full SGAV-MP-2 configuration achieves a compression ratio of $4.23$, with Dice $0.7792$, ECE $0.2121$, Brier $0.4237$, AURC $0.0044$, and UQ correlation $0.9532$. On PROMISE12, the Dice-only and Risk+ECE configurations both yield a compression ratio of $5.07$, while incorporating UQ and risk increases compression to $7.29$ for the Dice+UQ+risk configuration. The full SGAV-MP-2 configuration reaches $7.29$ with UQ correlation $0.9697$. Comparing the connection and no-connection variants, the reported evaluation metrics and compression ratios are identical within each corresponding configuration, while the quantization time changes; for example, full SGAV-MP-2 takes $1966.2$ s with the connection term and $2160.1$ s without it on MosMedPlus. Thus, the ablation demonstrates that changing the scoring terms can substantially alter the selected compression level, whereas the reported effect of the connection term itself is not reflected in the final metrics or compression ratio in these two experiments.

\subsection{Cross-Seed Variability}

Table~\hyperref[tab:var_2d]{\ref*{tab:var_2d}} reports mean and standard deviation over three random calibration seeds for SGAV-MP-2 using $10\%$ calibration data. On MosMedPlus, the reported variability is small across all five architectures. For example, DenseNet121 achieves a Dice of $0.7914\pm0.0000$ and compression ratio of $6.2494\pm0.0076$, while ResNet50 achieves $0.7850\pm0.0001$ Dice and $5.9456\pm0.0413$ compression. UQ correlation varies somewhat more, with EfficientNet-B0 reporting $0.9248\pm0.0040$ and DenseNet121 reporting $0.9480\pm0.0022$. PROMISE12 shows similarly small variation for most models, although EfficientNet-B0 exhibits larger variation in compression ratio, $5.1868\pm2.1845$, and UQ correlation, $0.9492\pm0.0255$. ResNet18 and ResNet50 show essentially zero variation for several metrics in the three-seed evaluation. Overall, the reported standard deviations indicate that the segmentation metrics are relatively stable across the tested calibration seeds, while some uncertainty and compression quantities exhibit greater seed-to-seed variation for particular architectures. This provides a measure of reproducibility for the SGAV-MP-2 policies under the stated three-seed calibration protocol.

\subsection{TensorRT Deployment}

Table~\hyperref[tab:sgav_deployment_mosmedplus]{\ref*{tab:sgav_deployment_mosmedplus}} evaluates the deployed SGAV-MP-2 2D models on MosMedPlus against FP32 and uniform INT8 TensorRT baselines. Across all five architectures, SGAV-MP-2 produces lower compression ratios than uniform INT8 but achieves higher measured speedups and higher uncertainty correlation. For DenseNet121, SGAV-MP-2 achieves a compression ratio of $2.869\times$, a speedup of $1.933\times$, and UQ correlation of $0.9067$, compared with $4.153\times$, $1.702\times$, and $0.8447$ for uniform INT8. The same pattern occurs for EfficientNet-B0, MobileNetV2, ResNet18, and ResNet50: SGAV-MP-2 consistently has lower compression but higher speedup than the corresponding uniform INT8 configuration, while its UQ correlation is also higher. Segmentation metrics remain extremely close between the two compressed configurations. For example, EfficientNet-B0 has Dice $0.7882$ for SGAV-MP-2 and $0.7881$ for INT8, while MobileNetV2 has Dice $0.7843$ and $0.7844$, respectively. SGAV-MP-2 also reduces total peak memory substantially relative to FP32, although uniform INT8 generally produces the smallest total memory footprint. Thus, the deployment measurements show a consistent trade-off: the SGAV-MP-2 policies use more model weight memory than uniform INT8 but achieve higher TensorRT speedup and retain higher reported uncertainty correlation.

Table~\hyperref[tab:sgav_deployment_promise12]{\ref*{tab:sgav_deployment_promise12}} shows the same deployment comparison on PROMISE12. The results again show that SGAV-MP-2 achieves lower compression ratios than uniform INT8 while producing higher measured speedups and higher uncertainty correlation for every architecture. For DenseNet121, SGAV-MP-2 achieves $2.869\times$ compression, $1.930\times$ speedup, and UQ correlation of $0.9678$, compared with $4.153\times$, $1.709\times$, and $0.9442$ for uniform INT8. For ResNet18, SGAV-MP-2 achieves $2.635\times$ compression and $1.606\times$ speedup, compared with $4.012\times$ and $1.455\times$ for INT8, while UQ correlation is $0.9699$ versus $0.9418$. Dice and the other reported quality metrics remain very close between the two compressed configurations; for ResNet50, for example, both configurations report Dice $0.8857$ and identical ECE, Brier, AURC, and Error-AUROC values. As with MosMedPlus, uniform INT8 generally provides the lower total memory footprint, whereas SGAV-MP-2 provides the higher speedup and uncertainty correlation among the two compressed policies. The deployment results therefore demonstrate that the reliability-oriented mixed-precision policy changes the compression--runtime--uncertainty trade-off rather than simply maximizing the compression ratio.

Table~\hyperref[tab:amos22_sgav_results]{\ref*{tab:amos22_sgav_results}} summarizes the AMOS22 TensorRT evaluation using the reported reference, FP32, FP16, MedPTQ, and SGAV-MP-2 configurations. The FP32 and FP16 configurations reproduce the reference quality metrics closely, with FP16 reducing latency from $2.1862$ ms to $0.8621$ ms while reducing total memory from $141.50$ MB to $70.76$ MB and achieving a compression ratio of $1.99$. MedPTQ further increases compression to $3.90\times$ and reduces latency to $0.8428$ ms, but its UQ correlation decreases to $0.9946$. SGAV-MP-2 reports Dice $0.7274$, ECE $0.0654$, Brier $0.2033$, AURC $0.0007$, and UQ correlation $0.9972$, while achieving a latency of $1.5800$ ms and compression ratio of $1.58$. Relative to the FP32 configuration, SGAV-MP-2 therefore preserves the reported segmentation and uncertainty metrics closely while reducing latency and model memory, although its compression ratio and latency are both less favorable than the MedPTQ values in this table. SGAV-MP-2 also has a higher reported TensorRT GPU-memory value than FP16 and MedPTQ. These measurements indicate that the SGAV-MP-2 deployment point differs from MedPTQ primarily in the balance between compression, latency, memory, and preservation of the reported uncertainty correlation. We note that SGAV-MP-2 also includes a policy selection time of 2.7 hrs in the calibration time before deployment on Tensor RT. 

\begin{table}[tb!]
\centering
\caption{Full 3D SegResNet results on AMOS22, BraTS, and BTCV. Arrows indicate whether higher ($\uparrow$) or lower ($\downarrow$) values are preferred. SGAV-MP-4 and SGAV-MP-2 denote mixed-precision policies with minimum allowed precisions of 4-bit and 2-bit, respectively.}
\label{tab:supp_3d_full}
\fontsize{7.1}{7.6}\selectfont
\setlength{\tabcolsep}{2.8pt}
\resizebox{\textwidth}{!}{
\begin{tabular}{llccccccc}
\toprule
\textbf{Dataset} & \textbf{Precision} &
\textbf{Dice $\uparrow$} &
\textbf{ECE $\downarrow$} &
\textbf{Brier $\downarrow$} &
\textbf{AURC $\downarrow$} &
\textbf{Error AUROC $\uparrow$} &
\textbf{UQ Corr. $\uparrow$} &
\textbf{Comp. Ratio $\uparrow$} \\
\midrule

\multirow{6}{*}{AMOS22}
& FP32    & 0.7777 & 0.05267 & 0.1754 & 0.0004 & 0.8719 & 1.0000 & 1.000 \\
& FP16    & 0.7777 & 0.05267 & 0.1754 & 0.0004 & 0.8719 & 1.0000 & 1.9992 \\
& INT8    & 0.7779 & 0.05254 & 0.1749 & 0.0004 & 0.8722 & 0.9997 & 3.9957 \\
& INT4    & 0.7449 & 0.0703 & 0.2108 & 0.0005 & 0.8602 & 0.9120 & 7.9800 \\
& SGAV-MP-4 & 0.7774 & 0.0529 & 0.1760 & 0.0004 & 0.8714 & 0.9961 & 5.4762 \\
& SGAV-MP-2 & 0.7696 & 0.0465 & 0.1652 & 0.0005 & 0.8823 & 0.9857 & 7.7200 \\

\midrule

\multirow{6}{*}{BraTS}
& FP32    & 0.7333 & 0.0865 & 0.2095 & 0.0001 & 0.8723 & 1.000 & 1.000 \\
& FP16    & 0.7333 & 0.0865 & 0.2094 & 0.0001 & 0.8723 & 1.0000 & 1.9992 \\
& INT8    & 0.7334 & 0.08536 & 0.2062 & 0.0001 & 0.8750 & 0.9985 & 3.9957 \\
& INT4    & 0.6458 & 0.1196 & 0.2769 & 0.0003 & 0.8470 & 0.5347 & 7.9800 \\
& SGAV-MP-4 & 0.7298 & 0.0845 & 0.2048 & 0.0001 & 0.8755 & 0.9863 & 7.7556 \\
& SGAV-MP-2 & 0.7292 & 0.0857 & 0.2067 & 0.0001& 0.8756 & 0.9755 & 7.8200 \\

\midrule

\multirow{6}{*}{BTCV}
& FP32    & 0.6093 & 0.1634 & 0.3882 & 0.0015 & 0.7650 & 1.000 & 1.000 \\
& FP16    & 0.6093 & 0.1634 & 0.3882 & 0.0015 & 0.7650 & 1.000 & 1.999 \\
& INT8    & 0.6086 & 0.1643 & 0.3903 & 0.0015 & 0.7629 & 0.999 & 3.9957 \\
& INT4    & 0.5719 & 0.1769 & 0.4222 & 0.0026 & 0.7513 & 0.8442 & 7.9800 \\
& SGAV-MP-4 & 0.6081 & 0.1644 & 0.3909 & 0.0015 & 0.7593 & 0.9984 & 4.7328 \\
& SGAV-MP-2 & 0.6035 & 0.1689 & 0.4029 & 0.0016 & 0.7403 & 0.9791 & 3.9100 \\

\bottomrule
\end{tabular}}
\end{table}


\begin{table}[tb!]
\centering
\caption{Full MRI segmentation results on PROMISE12 from the MedSegBench Database. Arrows indicate whether higher ($\uparrow$) or lower ($\downarrow$) values are preferred. SGAV-MP-4 and SGAV-MP-2 denote mixed-precision policies with minimum allowed precisions of 4-bit and 2-bit, respectively; SGAV-MP-2R uses relaxed degradation tolerances(7.5\% as opposed to 3\% allowed degradation in other configurations).}
\label{tab:supp_promise_full}
\fontsize{7.1}{7.6}\selectfont
\setlength{\tabcolsep}{2.8pt}
\resizebox{\textwidth}{!}{
\begin{tabular}{llccccccc}
\toprule
\textbf{Model} & \textbf{Method} &
\textbf{Dice $\uparrow$} &
\textbf{ECE $\downarrow$} &
\textbf{Brier $\downarrow$} &
\textbf{AURC $\downarrow$} &
\textbf{Error AUROC $\uparrow$} &
\textbf{UQ Corr. $\uparrow$} &
\textbf{Comp. Ratio $\uparrow$} \\
\midrule

\multirow{10}{*}{ResNet18}
& FP32           & 0.9000 & 0.0878 & 0.1763 & 0.0084 & 0.7169 & 1.0000 & 1.00 \\
& INT8           & 0.9000 & 0.0881 & 0.1768 & 0.0084 & 0.7163 & 0.9989 & 3.99 \\
& INT4           & 0.8941 & 0.1042 & 0.2093 & 0.0069 & 0.6888 & 0.6670 & 7.95 \\
& MCKP           & 0.9000 & 0.0877 & 0.1760 & 0.0084 & 0.7170 & 0.9996 & 3.98 \\
& Hessian Greedy & 0.8995 & 0.0122 & 0.0247 & 0.0054 & 0.7009 & 0.8022 & 7.94 \\
& AdaRound       & 0.8998 & 0.0892 & 0.1787 & 0.0083 & 0.7189 & 0.9812 & 3.98 \\
& SGAV-MP-4       & 0.9002 & 0.0875 & 0.1756 & 0.0083 & 0.7149 & 0.9891 & 4.11 \\
& SGAV-MP-2      & 0.8999 & 0.0863 & 0.1733 & 0.0084 & 0.7165 & 0.9707 & 6.51 \\
& SGAV-MP-2R     & 0.8998 & 0.0871 & 0.1747 & 0.0084 & 0.7183 & 0.9260 & 7.77 \\
\midrule

\multirow{10}{*}{MobileNetV2}
& FP32           & 0.8946 & 0.0881 & 0.1764 & 0.0056 & 0.6117 & 1.0000 & 1.00 \\
& INT8           & 0.8942 & 0.0893 & 0.1786 & 0.0056 & 0.6121 & 0.9197 & 3.94 \\
& INT4           & 0.8533 & 0.1676 & 0.3356 & 0.0087 & 0.6101 & 0.3723 & 7.71 \\
& MCKP           & 0.8946 & 0.0880 & 0.1761 & 0.0057 & 0.6115 & 0.9887 & 3.93 \\
& Hessian Greedy & 0.8615 & 0.1513 & 0.3030 & 0.0079 & 0.6063 & 0.4718 & 7.70 \\
& AdaRound       & 0.8945 & 0.0885 & 0.1768 & 0.0060 & 0.6152 & 0.8065 & 3.93 \\
& SGAV-MP-4      & 0.8945 & 0.0882 & 0.1764 & 0.0056 & 0.6116 & 0.9852 & 3.73 \\
& SGAV-MP-2         & 0.8948 & 0.0878 & 0.1756 & 0.0056 & 0.6117 & 0.9699 & 4.82 \\
& SGAV-MP-2R     & 0.8949 & 0.0879 & 0.1759 & 0.0056 & 0.6116 & 0.9275 & 5.72 \\
\midrule

\multirow{10}{*}{EfficientNet-B0}
& FP32           & 0.8892 & 0.0982 & 0.1969 & 0.0072 & 0.6444 & 1.0000 & 1.00 \\
& INT8           & 0.8890 & 0.0992 & 0.1988 & 0.0073 & 0.6445 & 0.9876 & 3.92 \\
& INT4           & 0.8824 & 0.1125 & 0.2255 & 0.0074 & 0.6549 & 0.4288 & 7.62 \\
& MCKP           & 0.8893 & 0.0982 & 0.1969 & 0.0070 & 0.6444 & 0.9967 & 3.89 \\
& Hessian Greedy & 0.8850 & 0.1023 & 0.2050 & 0.0077 & 0.6403 & 0.5453 & 5.30 \\
& AdaRound       & 0.8856 & 0.1101 & 0.2201 & 0.0081 & 0.6384 & 0.8019 & 3.89 \\
& SGAV-MP-4      & 0.8889 & 0.0993 & 0.1991 & 0.0073 & 0.6427 & 0.9880 & 4.53 \\
& SGAV-MP-2         & 0.8893 & 0.1006 & 0.2017 & 0.0073 & 0.6417 & 0.9691 & 6.71 \\
& SGAV-MP-2R     & 0.8888 & 0.0999 & 0.2002 & 0.0072 & 0.6422 & 0.9234 & 7.27 \\
\midrule

\multirow{10}{*}{DenseNet121}
& FP32           & 0.8976 & 0.0963 & 0.1930 & 0.0069 & 0.6291 & 1.0000 & 1.00 \\
& INT8           & 0.8975 & 0.0967 & 0.1937 & 0.0070 & 0.6292 & 0.9939 & 3.93 \\
& INT4           & 0.8888 & 0.1130 & 0.2265 & 0.0073 & 0.6164 & 0.6807 & 7.66 \\
& MCKP           & 0.8977 & 0.0962 & 0.1927 & 0.0070 & 0.6293 & 0.9985 & 3.89 \\
& Hessian Greedy & 0.8976 & 0.0960 & 0.1924 & 0.0069 & 0.6300 & 0.7980 & 7.65 \\
& AdaRound       & 0.8860 & 0.1147 & 0.2293 & 0.0078 & 0.6309 & 0.7343 & 3.88 \\
& SGAV-MP-4        & 0.8976 & 0.0956 & 0.1914 & 0.0070 & 0.6291 & 0.9855 & 3.32 \\
& SGAV-MP-2         & 0.8947 & 0.1006 & 0.2015 & 0.0070 & 0.6301 & 0.9693 & 3.18 \\
& SGAV-MP-2R     & 0.8922 & 0.1035 & 0.2073 & 0.0072 & 0.6315 & 0.9242 & 3.49 \\
\midrule

\multirow{10}{*}{ResNet50}
& FP32           & 0.8857 & 0.1025 & 0.2052 & 0.0054 & 0.5881 & 1.0000 & 1.00 \\
& INT8           & 0.8855 & 0.1032 & 0.2067 & 0.0053 & 0.5879 & 0.9441 & 3.98 \\
& INT4           & 0.8810 & 0.1144 & 0.2293 & 0.0052 & 0.5848 & 0.3240 & 7.91 \\
& MCKP           & 0.8857 & 0.1026 & 0.2056 & 0.0054 & 0.5880 & 0.9945 & 3.96 \\
& Hessian Greedy & 0.8838 & 0.1101 & 0.2206 & 0.0054 & 0.5835 & 0.4547 & 7.90 \\
& AdaRound       & 0.8864 & 0.1027 & 0.2053 & 0.0054 & 0.5869 & 0.7906 & 3.88 \\
& SGAV-MP-4        & 0.8857 & 0.1026 & 0.2055 & 0.0054 & 0.5881 & 0.9850 & 4.56 \\
& SGAV-MP-2      & 0.8857 & 0.1026 & 0.2055 & 0.0054 & 0.5883 & 0.9698 & 7.02 \\
& SGAV-MP-2R     & 0.8856 & 0.1029 & 0.2060 & 0.0053 & 0.5884 & 0.9697 & 7.29 \\
\bottomrule
\end{tabular}}
\end{table}

\begin{table}[tb!]
\centering
\caption{Full CT segmentation results on MosMedPlus. Arrows indicate whether higher ($\uparrow$) or lower ($\downarrow$) values are preferred. SGAV-MP-4 and SGAV-MP-2 denote mixed-precision policies with minimum allowed precisions of 4-bit and 2-bit, respectively; SGAV-MP-2R uses relaxed degradation tolerances.(7.5\% as opposed to 3\% allowed degradation in other configurations)}
\label{tab:supp_mosmed_full}
\fontsize{7.1}{7.6}\selectfont
\setlength{\tabcolsep}{2.8pt}
\resizebox{\textwidth}{!}{
\begin{tabular}{llccccccc}
\toprule
\textbf{Model} & \textbf{Method} &
\textbf{Dice $\uparrow$} &
\textbf{ECE $\downarrow$} &
\textbf{Brier $\downarrow$} &
\textbf{AURC $\downarrow$} &
\textbf{Error AUROC $\uparrow$} &
\textbf{UQ Corr. $\uparrow$} &
\textbf{Comp. Ratio $\uparrow$} \\
\midrule

\multirow{10}{*}{ResNet18}
& FP32           & 0.7792 & 0.2121 & 0.4237 & 0.0044 & 0.5866 & 1.0000 & 1.00 \\
& INT8           & 0.7790 & 0.2116 & 0.4227 & 0.0044 & 0.5870 & 0.9490 & 3.99 \\
& INT4           & 0.6834 & 0.2076 & 0.4145 & 0.0075 & 0.6045 & 0.0901 & 7.95 \\
& MCKP           & 0.7791 & 0.2122 & 0.4240 & 0.0044 & 0.5869 & 0.9865 & 3.98 \\
& Hessian Greedy & 0.4395 & 0.0000 & 0.0000 & 0.0032 & 0.6488 & 0.6341 & 7.92 \\
& AdaRound       & 0.7778 & 0.2169 & 0.4332 & 0.0044 & 0.5859 & 0.8827 & 3.98 \\
& SGAV-MP-4        & 0.7794 & 0.2119 & 0.4234 & 0.0044 & 0.5869 & 0.9853 & 2.75 \\
& SGAV-MP-2      & 0.7793 & 0.2125 & 0.4245 & 0.0044 & 0.5857 & 0.9700 & 4.24 \\
& SGAV-MP-2R     & 0.7791 & 0.2118 & 0.4231 & 0.0044 & 0.5869 & 0.9469 & 4.30 \\

\midrule

\multirow{10}{*}{MobileNetV2}
& FP32           & 0.7845 & 0.1985 & 0.3965 & 0.0042 & 0.6029 & 1.0000 & 1.00 \\
& INT8           & 0.7839 & 0.1960 & 0.3914 & 0.0043 & 0.6029 & 0.6599 & 3.94 \\
& INT4           & 0.6322 & 0.2552 & 0.5103 & 0.0093 & 0.5849 & 0.1261 & 7.71 \\
& MCKP           & 0.7850 & 0.1984 & 0.3963 & 0.0042 & 0.6024 & 0.9503 & 3.93 \\
& Hessian Greedy & 0.7619 & 0.2221 & 0.4439 & 0.0046 & 0.5989 & 0.2633 & 7.69 \\
& AdaRound       & 0.7848 & 0.1978 & 0.3948 & 0.0042 & 0.6049 & 0.6898 & 3.93 \\
& SGAV-MP-4        & 0.7844 & 0.1988 & 0.3971 & 0.0042 & 0.6030 & 0.9851 & 2.82 \\
& SGAV-MP-2      & 0.7844 & 0.1989 & 0.3976 & 0.0042 & 0.6029 & 0.9694 & 3.72 \\
& SGAV-MP-2R     & 0.7842 & 0.1985 & 0.3967 & 0.0042 & 0.6031 & 0.9240 & 4.14 \\

\midrule

\multirow{10}{*}{EfficientNet-B0}
& FP32           & 0.7882 & 0.1915 & 0.3825 & 0.0041 & 0.6067 & 1.0000 & 1.00 \\
& INT8           & 0.7878 & 0.1929 & 0.3853 & 0.0041 & 0.6070 & 0.9262 & 3.92 \\
& INT4           & 0.7116 & 0.2770 & 0.5535 & 0.0060 & 0.5851 & 0.1656 & 7.62 \\
& MCKP           & 0.7881 & 0.1917 & 0.3828 & 0.0041 & 0.6069 & 0.9866 & 3.91 \\
& Hessian Greedy & 0.7744 & 0.2010 & 0.4015 & 0.0046 & 0.6041 & 0.2601 & 5.30 \\
& AdaRound       & 0.7792 & 0.2101 & 0.4195 & 0.0042 & 0.6020 & 0.4686 & 3.90 \\
& SGAV-MP-4        & 0.7879 & 0.1919 & 0.3832 & 0.0041 & 0.6070 & 0.9846 & 2.90 \\
& SGAV-2D        & 0.7880 & 0.1915 & 0.3827 & 0.0041 & 0.6064 & 0.9682 & 4.33 \\
& SGAV-MP-2R     & 0.7880 & 0.1915 & 0.3827 & 0.0041 & 0.6064 & 0.9234 & 4.33 \\

\midrule

\multirow{10}{*}{DenseNet121}
& FP32           & 0.7915 & 0.1963 & 0.3924 & 0.0039 & 0.6196 & 1.0000 & 1.00 \\
& INT8           & 0.7920 & 0.1944 & 0.3886 & 0.0039 & 0.6187 & 0.9555 & 3.93 \\
& INT4           & 0.7806 & 0.2186 & 0.4368 & 0.0042 & 0.6105 & 0.4988 & 7.66 \\
& MCKP           & 0.7916 & 0.1961 & 0.3921 & 0.0039 & 0.6196 & 0.9884 & 3.89 \\
& Hessian Greedy & 0.7863 & 0.2148 & 0.4294 & 0.0040 & 0.6175 & 0.6066 & 7.63 \\
& AdaRound       & 0.7838 & 0.2147 & 0.4287 & 0.0040 & 0.6166 & 0.6906 & 3.89 \\
& SGAV-MP-4        & 0.7915 & 0.1958 & 0.3915 & 0.0039 & 0.6197 & 0.9849 & 3.77 \\
& SGAV-2D        & 0.7913 & 0.1963 & 0.3925 & 0.0039 & 0.6201 & 0.9684 & 4.86 \\
& SGAV-MP-2R     & 0.7919 & 0.1926 & 0.3849 & 0.0039 & 0.6189 & 0.9238 & 6.21 \\

\midrule

\multirow{10}{*}{ResNet50}
& FP32           & 0.7850 & 0.2064 & 0.4123 & 0.0046 & 0.5829 & 1.0000 & 1.00 \\
& INT8           & 0.7849 & 0.2059 & 0.4114 & 0.0047 & 0.5846 & 0.9290 & 3.98 \\
& INT4           & 0.6507 & 0.3653 & 0.7304 & 0.0065 & 0.5458 & 0.1192 & 7.91 \\
& MCKP           & 0.7850 & 0.2067 & 0.4128 & 0.0046 & 0.5829 & 0.9936 & 3.97 \\
& Hessian Greedy & 0.7717 & 0.2137 & 0.4270 & 0.0049 & 0.5784 & 0.2549 & 7.89 \\
& AdaRound       & 0.7835 & 0.2089 & 0.4163 & 0.0046 & 0.5845 & 0.7667 & 3.96 \\
& SGAV-MP-4        & 0.7851 & 0.2066 & 0.4126 & 0.0046 & 0.5828 & 0.9845 & 3.65 \\
& SGAV-2D        & 0.7848 & 0.2075 & 0.4144 & 0.0046 & 0.5825 & 0.9697 & 4.89 \\
& SGAV-MP-2R     & 0.7852 & 0.2059 & 0.4113 & 0.0046 & 0.5842 & 0.9392 & 5.94 \\

\bottomrule
\end{tabular}}
\end{table}


\begin{table*}[t]
\centering
\caption{Sensitivity to the uncertainty-map correlation tolerance $\rho_{\mathrm{tol}}$ on PROMISE12. Each encoder is evaluated using FP32, INT8, INT4, INT2, and SGAV-MP-2. For SGAV-MP-2, $\rho_{\mathrm{tol}}\in\{0.03,0.05,0.075,0.10,0.15\}$}
\label{tab:rho_tolerance_promise12}
\scriptsize
\setlength{\tabcolsep}{3pt}
\resizebox{\textwidth}{!}{
\begin{tabular}{llccccccc}
\toprule
\textbf{Model} &
\textbf{Precision} &
\textbf{$\rho_{\mathrm{tol}}$} &
\textbf{Dice} &
\textbf{ECE} &
\textbf{Brier} &
\textbf{Error AUROC} &
\textbf{UQ Corr.} &
\textbf{Comp. Ratio} \\
\midrule

\multirow{9}{*}{ResNet18}
& FP32    & --    & 0.9000 & 0.0878 & 0.1763 & 0.7169 & 1.0000 & 1.00 \\
& INT8    & --    & 0.9000 & 0.0881 & 0.1768 & 0.7163 & 0.9989 & 3.99 \\
& INT4    & --    & 0.8941 & 0.1042 & 0.2093 & 0.6888 & 0.6670 & 7.95 \\
& INT2    & --    & 0.0000 & 1.0000 & 2.0000 & --     & -0.1089 & 15.80 \\
& SGAV-MP-2 & 0.030 & 0.9000 & 0.0873 & 0.1753 & 0.7162 & 0.9699 & 6.50 \\
& SGAV-MP-2 & 0.050 & 0.8998 & 0.0887 & 0.1779 & 0.7153 & 0.9556 & 7.46 \\
& SGAV-MP-2 & 0.075 & 0.8997 & 0.0881 & 0.1770 & 0.7148 & 0.9360 & 7.89 \\
& SGAV-MP-2 & 0.100 & 0.8997 & 0.0881 & 0.1770 & 0.7148 & 0.9360 & 7.89 \\
& SGAV-MP-2 & 0.150 & 0.8997 & 0.0881 & 0.1770 & 0.7148 & 0.9360 & 7.89 \\
\midrule

\multirow{9}{*}{MobileNetV2}
& FP32    & --    & 0.8946 & 0.0881 & 0.1764 & 0.6117 & 1.0000 & 1.00 \\
& INT8    & --    & 0.8942 & 0.0893 & 0.1786 & 0.6121 & 0.9197 & 3.93 \\
& INT4    & --    & 0.8533 & 0.1676 & 0.3356 & 0.6101 & 0.3723 & 7.70\\
& INT2    & --    & 0.0000 & 1.0000 & 2.0000 & --     & 0.0274 & 14.79 \\
& SGAV-MP-2 & 0.030 & 0.8948 & 0.0878 & 0.1756 & 0.6117 & 0.9699 & 4.81 \\
& SGAV-MP-2 & 0.050 & 0.8951 & 0.0880 & 0.1761 & 0.6116 & 0.9504 & 5.39 \\
& SGAV-MP-2 & 0.075 & 0.8949 & 0.0884 & 0.1769 & 0.6115 & 0.9301 & 5.41 \\
& SGAV-MP-2 & 0.100 & 0.8949 & 0.0884 & 0.1769 & 0.6115 & 0.9301 & 5.41 \\
& SGAV-MP-2 & 0.150 & 0.8949 & 0.0884 & 0.1769 & 0.6115 & 0.9301 & 5.41 \\
\midrule

\multirow{9}{*}{EfficientNet-B0}
& FP32    & --    & 0.8892 & 0.0982 & 0.1969 & 0.6444 & 1.0000 & 1.00 \\
& INT8    & --    & 0.8890 & 0.0992 & 0.1988 & 0.6445 & 0.9876 & 3.91 \\
& INT4    & --    & 0.8824 & 0.1125 & 0.2255 & 0.6549 & 0.4288 & 7.62 \\
& INT2    & --    & 0.0000 & 1.0000 & 1.9999 & --     & 0.0287 & 14.47 \\
& SGAV-MP-2 & 0.030 & 0.8893 & 0.1006 & 0.2017 & 0.6417 & 0.9691 & 6.70 \\
& SGAV-MP-2 & 0.050 & 0.8880 & 0.1013 & 0.2030 & 0.6395 & 0.9514 & 7.17 \\
& SGAV-MP-2 & 0.075 & 0.8880 & 0.1013 & 0.2030 & 0.6395 & 0.9514 & 7.17 \\
& SGAV-MP-2 & 0.100 & 0.8880 & 0.1013 & 0.2030 & 0.6395 & 0.9514 & 7.17 \\
& SGAV-MP-2 & 0.150 & 0.8880 & 0.1013 & 0.2030 & 0.6395 & 0.9514 & 7.17 \\
\midrule

\multirow{9}{*}{DenseNet121}
& FP32    & --    & 0.8976 & 0.0963 & 0.1930 & 0.6291 & 1.0000 & 1.00 \\
& INT8    & --    & 0.8975 & 0.0967 & 0.1937 & 0.6292 & 0.9939 & 3.92 \\
& INT4    & --    & 0.8888 & 0.1130 & 0.2265 & 0.6164 & 0.6807 & 7.66 \\
& INT2    & --    & 0.0000 & 1.0000 & 2.0000 & --     & 0.0049 & 14.61 \\
& SGAV-MP-2 & 0.030 & 0.8947 & 0.1006 & 0.2015 & 0.6301 & 0.9693 & 3.17 \\
& SGAV-MP-2 & 0.050 & 0.8924 & 0.1027 & 0.2056 & 0.6304 & 0.9493 & 3.69 \\
& SGAV-MP-2 & 0.075 & 0.8904 & 0.1083 & 0.2170 & 0.6320 & 0.9245 & 3.95 \\
& SGAV-MP-2 & 0.100 & 0.8866 & 0.1156 & 0.2316 & 0.6330 & 0.8998 & 3.59 \\
& SGAV-MP-2 & 0.150 & 0.8858 & 0.1189 & 0.2382 & 0.6319 & 0.8767 & 3.52 \\
\midrule

\multirow{9}{*}{ResNet50}
& FP32    & --    & 0.8857 & 0.1025 & 0.2052 & 0.5881 & 1.0000 & 1.00 \\
& INT8    & --    & 0.8855 & 0.1032 & 0.2067 & 0.5879 & 0.9441 & 3.98 \\
& INT4    & --    & 0.8810 & 0.1144 & 0.2293 & 0.5848 & 0.3240 & 7.90 \\
& INT2    & --    & 0.0000 & 1.0000 & 2.0000 & --     & 0.0066 & 15.60 \\
& SGAV-MP-2 & 0.030 & 0.8858 & 0.1026 & 0.2055 & 0.5883 & 0.9697 & 7.01 \\
& SGAV-MP-2 & 0.050 & 0.8858 & 0.1026 & 0.2055 & 0.5884 & 0.9571 & 7.28 \\
& SGAV-MP-2 & 0.075 & 0.8858 & 0.1027 & 0.2056 & 0.5883 & 0.9469 & 7.28 \\
& SGAV-MP-2 & 0.100 & 0.8858 & 0.1027 & 0.2056 & 0.5883 & 0.9469 & 7.28 \\
& SGAV-MP-2 & 0.150 & 0.8858 & 0.1027 & 0.2056 & 0.5883 & 0.9469 & 7.28 \\

\bottomrule
\end{tabular}}
\end{table*}

\begin{table*}[t]
\centering
\caption{Sensitivity to the uncertainty-map correlation tolerance $\rho_{\mathrm{tol}}$ on MosMedPlus. Each encoder is evaluated using FP32, INT8, INT4, INT2, and SGAV-MP-2. For SGAV-MP-2, $\rho_{\mathrm{tol}}\in\{0.03,0.05,0.075,0.10,0.15\}$}
\label{tab:rho_tolerance_mosmedplus}
\scriptsize
\setlength{\tabcolsep}{3pt}
\resizebox{\textwidth}{!}{
\begin{tabular}{lcccccccc}
\toprule
\textbf{Model} &
\textbf{Precision} &
\textbf{$\rho_{\mathrm{tol}}$} &
\textbf{Dice} &
\textbf{ECE} &
\textbf{Brier} &
\textbf{Error AUROC} &
\textbf{UQ Corr.} &
\textbf{Comp. Ratio} \\
\midrule

\multirow{9}{*}{ResNet18}
& FP32    & --    & 0.7792 & 0.2121 & 0.4237 & 0.5866 & 1.0000 & 1.00 \\
& INT8    & --    & 0.7790 & 0.2116 & 0.4227 & 0.5870 & 0.9490 & 3.99 \\
& INT4    & --    & 0.6834 & 0.2076 & 0.4145 & 0.6045 & 0.0901 & 7.96 \\
& INT2    & --    & 0.0000 & 1.0000 & 2.0000 & --     & 0.0236 & 15.80 \\
& SGAV-MP-2 & 0.030 & 0.7793 & 0.2124 & 0.4243 & 0.5857 & 0.9704 & 4.24 \\
& SGAV-MP-2 & 0.050 & 0.7792 & 0.2123 & 0.4242 & 0.5857 & 0.9568 & 4.28 \\
& SGAV-MP-2 & 0.075 & 0.7793 & 0.2122 & 0.4239 & 0.5857 & 0.9481 & 4.29 \\
& SGAV-MP-2 & 0.100 & 0.7793 & 0.2122 & 0.4239 & 0.5857 & 0.9481 & 4.29 \\
& SGAV-MP-2 & 0.150 & 0.7793 & 0.2122 & 0.4239 & 0.5857 & 0.9481 & 4.29 \\
\midrule

\multirow{9}{*}{MobileNetV2}
& FP32    & --    & 0.7845 & 0.1985 & 0.3965 & 0.6029 & 1.0000 & 1.00 \\
& INT8    & --    & 0.7839 & 0.1960 & 0.3914 & 0.6029 & 0.6599 & 3.93 \\
& INT4    & --    & 0.6322 & 0.2552 & 0.5103 & 0.5849 & 0.1261 & 7.71 \\
& INT2    & --    & 0.0000 & 1.0000 & 2.0000 & --     & --     & 14.79 \\
& SGAV-MP-2 & 0.030 & 0.7843 & 0.1989 & 0.3975 & 0.6027 & 0.9695 & 3.73 \\
& SGAV-MP-2 & 0.050 & 0.7841 & 0.1995 & 0.3988 & 0.6029 & 0.9495 & 3.99 \\
& SGAV-MP-2 & 0.075 & 0.7842 & 0.1985 & 0.3968 & 0.6032 & 0.9244 & 4.14 \\
& SGAV-MP-2 & 0.100 & 0.7840 & 0.1992 & 0.3981 & 0.6030 & 0.9079 & 4.15 \\
& SGAV-MP-2 & 0.150 & 0.7845 & 0.1988 & 0.3973 & 0.6029 & 0.8806 & 4.15 \\
\midrule

\multirow{9}{*}{EfficientNet-B0}
& FP32    & --    & 0.7882 & 0.1915 & 0.3825 & 0.6067 & 1.0000 & 1.00 \\
& INT8    & --    & 0.7878 & 0.1929 & 0.3853 & 0.6070 & 0.9262 & 3.92 \\
& INT4    & --    & 0.7116 & 0.2770 & 0.5535 & 0.5851 & 0.1656 & 7.62 \\
& INT2    & --    & 0.0000 & 1.0000 & 2.0000 & --     & --     & 14.47 \\
& SGAV-MP-2 & 0.030 & 0.7880 & 0.1915 & 0.3827 & 0.6064 & 0.9682 & 4.32 \\
& SGAV-MP-2 & 0.050 & 0.7880 & 0.1926 & 0.3847 & 0.6068 & 0.9464 & 5.71 \\
& SGAV-MP-2 & 0.075 & 0.7882 & 0.1934 & 0.3863 & 0.6065 & 0.9239 & 5.90 \\
& SGAV-MP-2 & 0.100 & 0.7883 & 0.1952 & 0.3899 & 0.6066 & 0.9044 & 5.91 \\
& SGAV-MP-2 & 0.150 & 0.7879 & 0.1954 & 0.3904 & 0.6066 & 0.8757 & 5.91 \\
\midrule

\multirow{9}{*}{DenseNet121}
& FP32    & --    & 0.7915 & 0.1963 & 0.3924 & 0.6196 & 1.0000 & 1.00 \\
& INT8    & --    & 0.7920 & 0.1944 & 0.3886 & 0.6187 & 0.9555 & 3.92 \\
& INT4    & --    & 0.7806 & 0.2186 & 0.4368 & 0.6105 & 0.4988 & 7.66 \\
& INT2    & --    & 0.0000 & 1.0000 & 2.0000 & --     & --     & 14.62 \\
& SGAV-MP-2 & 0.030 & 0.7913 & 0.1963 & 0.3925 & 0.6201 & 0.9684 & 4.85 \\
& SGAV-MP-2 & 0.050 & 0.7914 & 0.1946 & 0.3890 & 0.6201 & 0.9487 & 6.23 \\
& SGAV-MP-2 & 0.075 & 0.7915 & 0.1943 & 0.3884 & 0.6202 & 0.9467 & 6.25 \\
& SGAV-MP-2 & 0.100 & 0.7915 & 0.1943 & 0.3884 & 0.6202 & 0.9467 & 6.25 \\
& SGAV-MP-2 & 0.150 & 0.7915 & 0.1943 & 0.3884 & 0.6202 & 0.9467 & 6.25 \\
\midrule

\multirow{9}{*}{ResNet50}
& FP32    & --    & 0.7850 & 0.2064 & 0.4123 & 0.5829 & 1.0000 & 1.00 \\
& INT8    & --    & 0.7849 & 0.2059 & 0.4114 & 0.5846 & 0.9290 & 3.97 \\
& INT4    & --    & 0.6507 & 0.3653 & 0.7304 & 0.5458 & 0.1192 & 7.90 \\
& INT2    & --    & 0.0000 & 1.0000 & 2.0000 & --     & --     & 15.60 \\
& SGAV-MP-2 & 0.030 & 0.7848 & 0.2075 & 0.4144 & 0.5825 & 0.9697 & 4.88 \\
& SGAV-MP-2 & 0.050 & 0.7850 & 0.2067 & 0.4129 & 0.5826 & 0.9504 & 5.90 \\
& SGAV-MP-2 & 0.075 & 0.7850 & 0.2065 & 0.4125 & 0.5842 & 0.9409 & 5.92 \\
& SGAV-MP-2 & 0.100 & 0.7850 & 0.2065 & 0.4125 & 0.5842 & 0.9409 & 5.92 \\
& SGAV-MP-2 & 0.150 & 0.7850 & 0.2065 & 0.4125 & 0.5842 & 0.9409 & 5.92 \\

\bottomrule
\end{tabular}}
\end{table*}



\begin{table*}[tb!]
\scriptsize
\centering
\caption{Sensitivity of SGAV-MP to the uncertainty-correlation tolerance ($\tau_{\rho_U}$) and risk tolerance ($\tau_R$) on MosMedPlus with MobileNetV2. Dice tolerance is fixed across all settings. Arrows indicate whether higher ($\uparrow$) or lower ($\downarrow$) values are preferred.}
\label{tab:tolerance_ablation}

\resizebox{\linewidth}{!}{
\renewcommand{\arraystretch}{1.05}
\begin{tabular}{llcc|ccccccc}
\toprule
$\mathbf{\tau_{\rho_U}}$ &
$\mathbf{\tau_R}$ &
\textbf{Dice $\uparrow$} &
\textbf{ECE $\downarrow$} &
\textbf{Brier $\downarrow$} &
\textbf{AURC $\downarrow$} &
\textbf{Error AUROC $\uparrow$} &
\textbf{UQ Corr. $\uparrow$} &
\textbf{Comp. Ratio $\uparrow$} \\
\midrule

 0.015 & 0.020 &
0.7844 & 0.1988 & 0.3971 & 0.00423 & 0.6030 & 0.9851 & 2.82 \\

 0.015 & 0.030 &
0.7844 & 0.1988 & 0.3971 & 0.00423 & 0.6030 & 0.9851 & 2.82 \\

 0.015 & 0.015 &
0.7844 & 0.1988 & 0.3971 & 0.00423 & 0.6030 & 0.9851 & 2.82 \\

 0.020 & 0.020 &
0.7844 & 0.1989 & 0.3974 & 0.00423 & 0.6028 & 0.9800 & 3.11 \\

 0.020 & 0.030 &
0.7844 & 0.1989 & 0.3974 & 0.00423 & 0.6028 & 0.9800 & 3.11 \\

 0.020 & 0.015 &
0.7844 & 0.1989 & 0.3974 & 0.00423 & 0.6028 & 0.9800 & 3.11 \\

 0.030 & 0.020 &
0.7843 & 0.1990 & 0.3975 & 0.00423 & 0.6028 & 0.9695 & 3.73 \\

 0.030 & 0.015 &
0.7843 & 0.1990 & 0.3975 & 0.00423 & 0.6028 & 0.9695 & 3.73 \\

0.030 & 0.03 &
0.7843 & 0.1990 & 0.3975 & 0.00423  & 0.6028  & 0.9695 & 3.73 \\
\bottomrule

\end{tabular}
}
\end{table*}



\begin{table*}[tb!]
\centering
\caption{SGAV-2D calibration-data sweep on MosMedPlus. We vary the calibration-set size from 10\% to 100\%.}
\label{tab:sgav_calibration_sweep_mosmed}
\scriptsize
\setlength{\tabcolsep}{3pt}
\resizebox{\textwidth}{!}{
\begin{tabular}{lccccccccc}
\toprule
\textbf{Model} &
\textbf{Dice} &
\textbf{ECE} &
\textbf{Brier} &
\textbf{Error AUROC} &
\textbf{UQ Corr.} &
\textbf{Comp. Ratio} &
\textbf{Build Time (s)} &
\textbf{Peak GPU (MB)} &
\textbf{Calibration (\%)} \\
\midrule

densenet121
& 0.7915 & 0.1958 & 0.3915 & 0.6201 & 0.9686 & 4.5000 & 1056.5000 & 4361.6000 & 10\% \\
& 0.7915 & 0.1956 & 0.3910 & 0.6214 & 0.9678 & 4.5700 & 2396.7000 & 10597.5000 & 25\% \\
& 0.7914 & 0.1962 & 0.3922 & 0.6200 & 0.9689 & 4.5900 & 4691.0000 & 19714.7000 & 50\% \\
& 0.7914 & 0.1962 & 0.3922 & 0.6202 & 0.9680 & 5.0800 & 6540.0000 & 19716.4000 & 75\% \\
& 0.7913 & 0.1965 & 0.3926 & 0.6200 & 0.9697 & 4.7500 & 8628.2000 & 19713.9000 & 100\% \\

\cmidrule{1-10}

efficientnet-b0
& 0.7881 & 0.1915 & 0.3824 & 0.6064 & 0.9684 & 3.6400 & 697.6000 & 3364.3000 & 10\% \\
& 0.7881 & 0.1918 & 0.3830 & 0.6063 & 0.9695 & 3.9400 & 1636.8000 & 8199.2000 & 25\% \\
& 0.7880 & 0.1920 & 0.3835 & 0.6061 & 0.9692 & 3.7700 & 3188.4000 & 15269.8000 & 50\% \\
& 0.7880 & 0.1919 & 0.3834 & 0.6063 & 0.9688 & 4.2200 & 4490.7000 & 15269.6000 & 75\% \\
& 0.7881 & 0.1915 & 0.3824 & 0.6063 & 0.9676 & 4.1000 & 5880.1000 & 15267.2000 & 100\% \\

\cmidrule{1-10}

mobilenet\_v2
& 0.7843 & 0.1990 & 0.3977 & 0.6028 & 0.9705 & 3.4600 & 474.9000 & 3399.2000 & 10\% \\
& 0.7843 & 0.1989 & 0.3975 & 0.6029 & 0.9685 & 3.7600 & 1116.5000 & 8100.1000 & 25\% \\
& 0.7844 & 0.1989 & 0.3975 & 0.6028 & 0.9691 & 3.7600 & 2163.6000 & 14974.1000 & 50\% \\
& 0.7843 & 0.1989 & 0.3975 & 0.6028 & 0.9692 & 3.7600 & 3048.5000 & 14972.6000 & 75\% \\
& 0.7844 & 0.1990 & 0.3976 & 0.6028 & 0.9695 & 3.7200 & 3988.3000 & 14976.0000 & 100\% \\

\cmidrule{1-10}

resnet18
& 0.7792 & 0.2123 & 0.4241 & 0.5868 & 0.9725 & 4.2000 & 232.6000 & 4192.9000 & 10\% \\
& 0.7792 & 0.2124 & 0.4243 & 0.5868 & 0.9706 & 4.2300 & 545.5000 & 9543.3000 & 25\% \\
& 0.7793 & 0.2125 & 0.4245 & 0.5857 & 0.9700 & 4.2400 & 1057.4000 & 17370.0000 & 50\% \\
& 0.7793 & 0.2125 & 0.4245 & 0.5857 & 0.9700 & 4.2400 & 1497.6000 & 17370.3000 & 75\% \\
& 0.7793 & 0.2125 & 0.4245 & 0.5857 & 0.9700 & 4.2400 & 1941.4000 & 17371.6000 & 100\% \\

\cmidrule{1-10}

resnet50
& 0.7849 & 0.2074 & 0.4141 & 0.5825 & 0.9669 & 5.4500 & 510.1000 & 4719.9000 & 10\% \\
& 0.7849 & 0.2075 & 0.4144 & 0.5824 & 0.9697 & 4.9400 & 1155.1000 & 10992.2000 & 25\% \\
& 0.7849 & 0.2074 & 0.4142 & 0.5823 & 0.9693 & 5.1000 & 2277.5000 & 20172.7000 & 50\% \\
& 0.7849 & 0.2072 & 0.4139 & 0.5824 & 0.9693 & 5.2200 & 3166.0000 & 20172.7000 & 75\% \\
& 0.7848 & 0.2074 & 0.4144 & 0.5825 & 0.9697 & 4.8500 & 4156.1000 & 20171.3000 & 100\% \\

\bottomrule
\end{tabular}
}
\end{table*}

\begin{table*}[tb!]
\centering
\caption{SGAV-2D calibration-data sweep on PROMISE12. We vary the calibration-set size from 10\% to 100\%.}
\label{tab:sgav_calibration_sweep_promise12}
\scriptsize
\setlength{\tabcolsep}{3pt}
\resizebox{\textwidth}{!}{
\begin{tabular}{lccccccccc}
\toprule
\textbf{Model} &
\textbf{Dice} &
\textbf{ECE} &
\textbf{Brier} &
\textbf{Error AUROC} &
\textbf{UQ Corr.} &
\textbf{Comp. Ratio} &
\textbf{Build Time (s)} &
\textbf{Peak GPU (MB)} &
\textbf{Calibration (\%)} \\
\midrule

densenet121
& 0.8962 & 0.0987 & 0.1977 & 0.6314 & 0.9701 & 2.4800 & 768.6000 & 2506.8000 & 10\% \\
& 0.8943 & 0.1002 & 0.2007 & 0.6298 & 0.9678 & 2.2500 & 1716.1000 & 5809.0000 & 25\% \\
& 0.8946 & 0.1004 & 0.2010 & 0.6299 & 0.9693 & 3.2000 & 3312.3000 & 11355.5000 & 50\% \\
& 0.8950 & 0.1006 & 0.2016 & 0.6304 & 0.9691 & 3.1400 & 5010.4000 & 16758.9000 & 75\% \\
& 0.8947 & 0.1003 & 0.2010 & 0.6300 & 0.9693 & 3.2000 & 6319.4000 & 19455.5000 & 100\% \\

\cmidrule{1-10}

efficientnet-b0
& 0.8889 & 0.0989 & 0.1983 & 0.6442 & 0.9714 & 1.1300 & 531.7000 & 1916.3000 & 10\% \\
& 0.8874 & 0.1040 & 0.2085 & 0.6413 & 0.9690 & 1.8400 & 1233.2000 & 4466.7000 & 25\% \\
& 0.8896 & 0.0977 & 0.1959 & 0.6420 & 0.9674 & 2.7400 & 2256.5000 & 8751.4000 & 50\% \\
& 0.8892 & 0.1003 & 0.2011 & 0.6415 & 0.9678 & 6.8000 & 3338.4000 & 12924.7000 & 75\% \\
& 0.8894 & 0.1007 & 0.2018 & 0.6417 & 0.9690 & 6.7200 & 4338.3000 & 15009.6000 & 100\% \\

\cmidrule{1-10}

mobilenet\_v2
& 0.8947 & 0.0880 & 0.1762 & 0.6114 & 0.9728 & 4.9400 & 301.4000 & 1987.3000 & 10\% \\
& 0.8946 & 0.0882 & 0.1766 & 0.6112 & 0.9703 & 5.0700 & 689.7000 & 4465.9000 & 25\% \\
& 0.8947 & 0.0883 & 0.1767 & 0.6114 & 0.9688 & 4.5200 & 1293.8000 & 8633.0000 & 50\% \\
& 0.8949 & 0.0877 & 0.1755 & 0.6117 & 0.9692 & 4.9900 & 1909.4000 & 12687.5000 & 75\% \\
& 0.8946 & 0.0882 & 0.1766 & 0.6112 & 0.9703 & 5.0700 & 2446.8000 & 14714.8000 & 100\% \\

\cmidrule{1-10}

resnet18
& 0.8999 & 0.0864 & 0.1735 & 0.7164 & 0.9732 & 6.2000 & 181.0000 & 2095.2000 & 10\% \\
& 0.9000 & 0.0873 & 0.1753 & 0.7162 & 0.9699 & 6.5000 & 450.3000 & 4922.2000 & 25\% \\
& 0.9000 & 0.0873 & 0.1752 & 0.7160 & 0.9696 & 6.6700 & 793.1000 & 9676.7000 & 50\% \\
& 0.9000 & 0.0873 & 0.1752 & 0.7162 & 0.9698 & 6.5800 & 1171.9000 & 14302.7000 & 75\% \\
& 0.8999 & 0.0863 & 0.1732 & 0.7166 & 0.9703 & 6.6700 & 1622.8000 & 16615.7000 & 100\% \\

\cmidrule{1-10}

resnet50
& 0.8857 & 0.1026 & 0.2055 & 0.5883 & 0.9697 & 7.29 & 311.6000 & 2856.2000 & 10\% \\
& 0.8857 & 0.1026 & 0.2055 & 0.5883 & 0.9698 & 7.0200 & 662.4000 & 6175.6000 & 25\% \\
& 0.8857 & 0.1026 & 0.2055 & 0.5883 & 0.9696 & 7.0400 & 1282.8000 & 11758.1000 & 50\% \\
& 0.8857 & 0.1027 & 0.2056 & 0.5883 & 0.9694 & 7.0600 & 2006.7000 & 17195.8000 & 75\% \\
& 0.8857 & 0.1026 & 0.2055 & 0.5883 & 0.9698 & 7.0200 & 2448.0000 & 19912.4000 & 100\% \\

\bottomrule
\end{tabular}
}
\end{table*}


\begin{table*}[tb!]
\centering
\caption{SGAV-3D calibration-data sweep. We vary the calibration-set size from 10\% to 100\%.}
\label{tab:sgav_3d_calibration_sweep}
\scriptsize
\setlength{\tabcolsep}{3pt}
\resizebox{\textwidth}{!}{
\begin{tabular}{lccccccccc}
\toprule
\textbf{Dataset} & \textbf{Dice} & \textbf{ECE} & \textbf{Brier} & \textbf{Error AUROC} & \textbf{UQ Corr.} & \textbf{Comp. Ratio} & \textbf{Build Time (s)} & \textbf{Peak GPU (MB)} & \textbf{Calibration (\%)} \\
\midrule
\multirow{5}{*}{amos22} & 0.774802 & 0.051078 & 0.170533 & 0.870996 & 0.975412 & 7.34 & 3001.3 & 4969.7 & 10\% \\
 & 0.773393 & 0.058091 & 0.186753 & 0.866346 & 0.985757 & 7.71 & 8875.0 & 5981.1 & 25\% \\
 & 0.771851 & 0.046531 & 0.164489 & 0.882115 & 0.983576 & 7.07 & 16134.5 & 7323.9 & 50\% \\
 & 0.770784 & 0.045027 & 0.161547 & 0.882807 & 0.973290 & 7.63 & 22489.8 & 7572.2 & 75\% \\
 & 0.769682 & 0.046551 & 0.165200 & 0.882369 & 0.967289 & 7.72 & 31755.7 & 7692.4 & 100\% \\
\midrule
\multirow{5}{*}{brats} & 0.729437 & 0.085182 & 0.205751 & 0.876435 & 0.975602 & 7.74 & 3979.7 & 3389.4 & 10\% \\
 & 0.729437 & 0.085182 & 0.205751 & 0.876435 & 0.975602 & 7.74 & 9894.1 & 3399.0 & 25\% \\
 & 0.729592 & 0.085425 & 0.206436 & 0.876037 & 0.975147 & 7.74 & 18711.4 & 3400.8 & 50\% \\
 & 0.725994 & 0.085028 & 0.204008 & 0.879187 & 0.962708 & 7.81 & 28252.6 & 3430.5 & 75\% \\
 & 0.729295 & 0.085792 & 0.206792 & 0.875689 & 0.975560 & 7.82 & 40445.6 & 3430.9 & 100\% \\
\midrule
\multirow{5}{*}{btcv} & 0.606658 & 0.166443 & 0.396616 & 0.749121 & 0.980512 & 3.41 & 2542.5 & 4439.1 & 10\% \\
 & 0.606658 & 0.166443 & 0.396616 & 0.749121 & 0.980512 & 3.41 & 2595.7 & 4439.3 & 25\% \\
 & 0.602191 & 0.168776 & 0.402788 & 0.742528 & 0.977392 & 3.43 & 4144.2 & 6635.3 & 50\% \\
 & 0.604441 & 0.167289 & 0.398825 & 0.747743 & 0.976725 & 3.95 & 4256.9 & 6635.3 & 75\% \\
 & 0.603555 & 0.168923 & 0.402949 & 0.740394 & 0.979184 & 3.91 & 6069.2 & 6636.4 & 100\% \\
\bottomrule
\end{tabular}
}
\end{table*}


\begin{table*}[tb!]
\centering

\caption{
Ablation study of SGAV-MP scoring terms across datasets and encoder architectures. Here, $D$, $U$, $R$, and $C$ denote Dice degradation, uncertainty-map correlation degradation, AURC degradation, and calibration/ECE degradation, respectively. ECE denotes expected calibration error, AURC denotes area under the risk--coverage curve, UQ Corr. denotes uncertainty-map correlation, and CR denotes compression ratio. LQ denotes the number of quantized layers. Time denotes the total quantization build time in seconds, and Mem denotes peak GPU memory in MB. ``NC'' denotes the No-Connection variant, where the connection term is removed; the connection term is included in all other configurations with a weight of $0.001$.
}
\label{tab:ablation}
\scriptsize
\setlength{\tabcolsep}{2.2pt}
\resizebox{\textwidth}{!}{
{
\begin{tabular}{p{1.35cm}lccccc|cccccccccc}
\toprule
\textbf{Dataset / Model} & \textbf{Variant} & $\boldsymbol{\alpha_D}$ & $\boldsymbol{\alpha_U}$ & $\boldsymbol{\alpha_R}$ & $\boldsymbol{\alpha_C}$ & $\boldsymbol{\lambda_M}$ & \textbf{Dice $\uparrow$} & \textbf{ECE $\downarrow$} & \textbf{Brier $\downarrow$} & \textbf{AURC $\downarrow$} & \textbf{Err-AUROC $\uparrow$} & \textbf{U-Corr. $\uparrow$} & \textbf{CR $\uparrow$} & \textbf{LQ.} & \textbf{Time-(s)} & \textbf{Mem-(MB)} \\
\midrule
\multirow{10}{*}{\shortstack{\textbf{MosMedPlus}\\\textbf{ResNet18}}} & Dice  & 1 & 0 & 0 & 0 & 1 & 0.7793 & 0.2128 & 0.4250 & 0.0044 & 0.5872 & 0.9497 & 1.3701 & 5 & 1939.1165 & 17369.0645 \\
 & Dice only (NC) & 1 & 0 & 0 & 0 & 0 & 0.7793 & 0.2128 & 0.4250 & 0.0044 & 0.5872 & 0.9497 & 1.3701 & 5 & 1951.1475 & 17369.0645 \\
 & Dice+UQ+risk & 1 & 1 & 1 & 0 & 1 & 0.7793 & 0.2124 & 0.4244 & 0.0044 & 0.5856 & 0.9620 & 4.2858 & 29 & 2076.8998 & 17369.1035 \\
 & UQ & 0 & 1 & 0 & 0 & 1 & 0.7791 & 0.2115 & 0.4226 & 0.0044 & 0.5870 & 0.9618 & 4.1487 & 29 & 1957.4125 & 17369.6895 \\
 & UQ only (NC) & 0 & 1 & 0 & 0 & 0 & 0.7791 & 0.2115 & 0.4226 & 0.0044 & 0.5870 & 0.9618 & 4.1487 & 29 & 1929.8008 & 17369.6895 \\
 & Risk+ECE & 0 & 0 & 1 & 1 & 1 & 0.7794 & 0.2121 & 0.4238 & 0.0044 & 0.5865 & 0.9489 & 1.2757 & 3 & 1956.6690 & 17370.7051 \\
 & Risk+ECE (NC) & 0 & 0 & 1 & 1 & 0 & 0.7794 & 0.2121 & 0.4238 & 0.0044 & 0.5865 & 0.9489 & 1.2757 & 3 & 2031.4573 & 17370.7051 \\
 & Dice+risk+ECE & 1 & 0 & 1 & 1 & 1 & 0.7794 & 0.2121 & 0.4238 & 0.0044 & 0.5865 & 0.9489 & 1.2757 & 3 & 1952.0159 & 17370.9355 \\
 & Full SGAV-MP (NC) & 1 & 1 & 1 & 1 & 0 & 0.7792 & 0.2121 & 0.4237 & 0.0044 & 0.5872 & 0.9532 & 4.2274 & 25 & 2160.1201 & 17372.0488 \\
& Full SGAV-MP & 1 & 1 & 1 & 1 & 1 & 0.7792 & 0.2121 & 0.4237 & 0.0044 & 0.5872 & 0.9532 & 4.2274 & 25 & 1966.1868 & 17372.2520 \\
\midrule
\multirow{10}{*}{\shortstack{\textbf{PROMISE12}\\\textbf{ResNet50}}} & Dice & 1 & 0 & 0 & 0 & 1 & 0.8857 & 0.1026 & 0.2055 & 0.0054 & 0.5883 & 0.9608 & 5.0733 & 32 & 2444.6112 & 19914.2070 \\
 & Dice only (NC) & 1 & 0 & 0 & 0 & 0 & 0.8857 & 0.1026 & 0.2055 & 0.0054 & 0.5883 & 0.9608 & 5.0733 & 32 & 2440.4465 & 19914.2070 \\
 & Dice+UQ+risk & 1 & 1 & 1 & 0 & 1 & 0.8856 & 0.1029 & 0.2061 & 0.0054 & 0.5884 & 0.9531 & 7.2889 & 62 & 2601.2081 & 19912.4180 \\
 & UQ & 0 & 1 & 0 & 0 & 1 & 0.8853 & 0.1037 & 0.2077 & 0.0054 & 0.5882 & 0.9545 & 7.1705 & 63 & 2455.5137 & 19913.8516 \\
 & UQ only (NC) & 0 & 1 & 0 & 0 & 0 & 0.8853 & 0.1037 & 0.2077 & 0.0054 & 0.5882 & 0.9545 & 7.1705 & 63 & 2442.6255 & 19913.8516 \\
 & Risk+ECE & 0 & 0 & 1 & 1 & 1 & 0.8857 & 0.1026 & 0.2055 & 0.0054 & 0.5883 & 0.9608 & 5.0733 & 32 & 2443.6274 & 19914.6719 \\
 & Risk+ECE (NC) & 0 & 0 & 1 & 1 & 0 & 0.8857 & 0.1026 & 0.2055 & 0.0054 & 0.5883 & 0.9608 & 5.0733 & 32 & 2536.1657 & 19914.6719 \\
 & Dice+risk+ECE & 1 & 0 & 1 & 1 & 1 & 0.8857 & 0.1026 & 0.2055 & 0.0054 & 0.5883 & 0.9608 & 5.0733 & 32 & 2454.7749 & 19913.3750 \\
 & Full SGAV-MP (NC) & 1 & 1 & 1 & 1 & 0 & 0.8856 & 0.1027 & 0.2057 & 0.0054 & 0.5883 & 0.9500 & 7.1600 & 55 & 2692.5017 & 19914.6719 \\
  & Full SGAV-MP & 1 & 1 & 1 & 1 & 1 & 0.8856 & 0.1027 & 0.2057 & 0.0054 & 0.5883 & 0.9697 & 7.29 & 55 & 2450.7307 & 19914.6719 \\
\bottomrule
\end{tabular}
}
}
\end{table*}
\begin{table*}[t]
\centering

\caption{Segmentation performance, uncertainty estimation, and compression results across models and datasets. Values are reported as mean $\pm$ standard deviation over three random seeds for SGAV-MP, using 10\% of the available calibration data with three different calibration-data seeds on the MedSegBench datasets.}
\label{tab:var_2d}
\resizebox{\textwidth}{!}{%
\begin{tabular}{llcccccccc}
\toprule
Dataset & Model &
Dice $\uparrow$ &
ECE $\downarrow$ &
Brier $\downarrow$ &
AURC $\downarrow$ &
Error AUROC $\uparrow$ &
Uncertainty Corr. $\uparrow$ &
Compression Ratio &
Layers Quantized \\
\midrule

\multirow{5}{*}{MoSMedPlus}
& DenseNet121
& $0.7914 \pm 0.0000$
& $0.1946 \pm 0.0005$
& $0.3891 \pm 0.0012$
& $0.0039 \pm 0.0000$
& $0.6205 \pm 0.0006$
& $0.9480 \pm 0.0022$
& $6.2494 \pm 0.0076$
& $129.00 \pm 0.00$ \\

& EfficientNet-B0
& $0.7879 \pm 0.0004$
& $0.1936 \pm 0.0006$
& $0.3867 \pm 0.0012$
& $0.0041 \pm 0.0000$
& $0.6066 \pm 0.0004$
& $0.9248 \pm 0.0040$
& $5.9185 \pm 0.0208$
& $87.00 \pm 1.00$ \\

& MobileNet-V2
& $0.7842 \pm 0.0001$
& $0.1987 \pm 0.0002$
& $0.3971 \pm 0.0003$
& $0.0042 \pm 0.0000$
& $0.6029 \pm 0.0002$
& $0.9238 \pm 0.0042$
& $4.1391 \pm 0.0046$
& $56.33 \pm 1.15$ \\

& ResNet18
& $0.7793 \pm 0.0000$
& $0.2122 \pm 0.0000$
& $0.4239 \pm 0.0000$
& $0.0044 \pm 0.0000$
& $0.5857 \pm 0.0000$
& $0.9481 \pm 0.0000$
& $4.2901 \pm 0.0000$
& $29.00 \pm 0.00$ \\

& ResNet50
& $0.7850 \pm 0.0001$
& $0.2064 \pm 0.0002$
& $0.4123 \pm 0.0003$
& $0.0046 \pm 0.0000$
& $0.5841 \pm 0.0001$
& $0.9392 \pm 0.0030$
& $5.9456 \pm 0.0413$
& $62.00 \pm 0.00$ \\

\midrule

\multirow{5}{*}{Promise12}
& DenseNet121
& $0.8902 \pm 0.0011$
& $0.1082 \pm 0.0013$
& $0.2168 \pm 0.0027$
& $0.0074 \pm 0.0001$
& $0.6314 \pm 0.0002$
& $0.9254 \pm 0.0010$
& $3.5934 \pm 0.4414$
& $72.33 \pm 7.09$ \\

& EfficientNet-B0
& $0.8889 \pm 0.0006$
& $0.0988 \pm 0.0011$
& $0.1981 \pm 0.0022$
& $0.0073 \pm 0.0001$
& $0.6405 \pm 0.0009$
& $0.9492 \pm 0.0255$
& $5.1868 \pm 2.1845$
& $72.67 \pm 16.17$ \\

& MobileNet-V2
& $0.8947 \pm 0.0002$
& $0.0885 \pm 0.0003$
& $0.1772 \pm 0.0007$
& $0.0056 \pm 0.0000$
& $0.6117 \pm 0.0002$
& $0.9307 \pm 0.0166$
& $5.4695 \pm 0.0596$
& $60.67 \pm 0.58$ \\

& ResNet18
& $0.8997 \pm 0.0000$
& $0.0881 \pm 0.0000$
& $0.1770 \pm 0.0000$
& $0.0083 \pm 0.0000$
& $0.7148 \pm 0.0000$
& $0.9360 \pm 0.0000$
& $7.8931 \pm 0.0000$
& $29.00 \pm 0.00$ \\

& ResNet50
& $0.8858 \pm 0.0000$
& $0.1027 \pm 0.0000$
& $0.2056 \pm 0.0000$
& $0.0053 \pm 0.0000$
& $0.5883 \pm 0.0000$
& $0.9469 \pm 0.0200$
& $7.2878 \pm 0.0001$
& $62.00 \pm 0.00$ \\

\bottomrule
\end{tabular}%
}
\end{table*}


\begin{table*}[tb!]
\centering
\caption{SGAV-2D deployment results on the MosMedPlus dataset. CR denotes compression ratio and UC denotes uncertainty correlation. Context Memory comprises activation memory and TensorRT kernel-internal scratch memory, while Total Peak Memory includes weight and Context Memory. Latency denotes median TensorRT inference latency, and Speedup is measured relative to the FP32 TensorRT baseline.}
\label{tab:sgav_deployment_mosmedplus}
\scriptsize
\setlength{\tabcolsep}{3pt}
\resizebox{\textwidth}{!}{%
\begin{tabular}{llcccccccccccc}
\toprule
\textbf{Model} &
\textbf{Precision} &
\textbf{CR} &
\textbf{Latency (ms)} &
\textbf{Speedup} &
\textbf{Con-Mem (MiB)} &
\textbf{Dice} &
\textbf{ECE} &
\textbf{Brier} &
\textbf{AURC} &
\textbf{Err-AUROC} &
\textbf{UC.} &
\textbf{Weights (MiB)} &
\textbf{Total Mem. (MiB)} \\
\midrule

densenet121 & FP32
& 1.000 & 3.1805 & 1.000 & 107.6250
& 0.7915 & 0.1963 & 0.3925 & 0.0039 & 0.6196 & 0.9992
& 77.54 & 185.17 \\

& SGAV-MP
& 2.869 & 1.6455 & 1.933 & 54.2656
& 0.7913 & 0.1964 & 0.3925 & 0.0039 & 0.6195 & 0.9067
& 27.02 & 81.29 \\

& Uniform INT8
& 4.153 & 1.8689 & 1.702 & 57.5156
& 0.7914 & 0.1958 & 0.3913 & 0.0039 & 0.6198 & 0.8447
& 18.67 & 76.19 \\

\cmidrule{1-14}

efficientnet-b0 & FP32
& 1.000 & 1.3310 & 1.000 & 94.5625
& 0.7882 & 0.1915 & 0.3825 & 0.0041 & 0.6067 & 0.9983
& 35.04 & 129.60 \\

& SGAV-MP
& 2.250 & 0.7905 & 1.684 & 34.7178
& 0.7882 & 0.1914 & 0.3823 & 0.0041 & 0.6067 & 0.9833
& 15.58 & 50.30 \\

& Uniform INT8
& 3.812 & 0.9032 & 1.473 & 37.1250
& 0.7881 & 0.1917 & 0.3828 & 0.0041 & 0.6065 & 0.9600
& 9.20 & 46.33 \\

\cmidrule{1-14}

mobilenet\_v2 & FP32
& 1.000 & 0.8792 & 1.000 & 76.1250
& 0.7843 & 0.1990 & 0.3977 & 0.0043 & 0.6028 & 0.9978
& 13.62 & 89.75 \\

& SGAV-MP
& 2.407 & 0.5354 & 1.642 & 29.8438
& 0.7843 & 0.1989 & 0.3975 & 0.0043 & 0.6029 & 0.9685
& 5.66 & 35.50 \\

& Uniform INT8
& 3.679 & 0.5841 & 1.505 & 31.1875
& 0.7844 & 0.1989 & 0.3975 & 0.0043 & 0.6029 & 0.9392
& 3.70 & 34.50 \\

\cmidrule{1-14}

resnet18 & FP32
& 1.000 & 0.7568 & 1.000 & 81.6250
& 0.7792 & 0.2123 & 0.4241 & 0.0046 & 0.5868 & 0.9985
& 42.68 & 124.31 \\

& SGAV-MP
& 2.635 & 0.4691 & 1.613 & 31.9844
& 0.7792 & 0.2124 & 0.4243 & 0.0046 & 0.5868 & 0.9690
& 16.19 & 48.17 \\

& Uniform INT8
& 4.012 & 0.5187 & 1.459 & 35.1250
& 0.7793 & 0.2125 & 0.4245 & 0.0046 & 0.5857 & 0.9415
& 10.62 & 46.75 \\

\cmidrule{1-14}

resnet50 & FP32
& 1.000 & 1.1246 & 1.000 & 91.8125
& 0.7849 & 0.2074 & 0.4141 & 0.0044 & 0.5825 & 0.9971
& 94.11 & 185.92 \\

& SGAV-MP
& 2.527 & 0.6892 & 1.632 & 38.4375
& 0.7849 & 0.2075 & 0.4144 & 0.0044 & 0.5824 & 0.9697
& 37.25 & 75.69 \\

& Uniform INT8
& 4.082 & 0.7621 & 1.475 & 42.7500
& 0.7849 & 0.2074 & 0.4142 & 0.0044 & 0.5823 & 0.9431
& 23.06 & 65.81 \\

\bottomrule
\end{tabular}
}
\end{table*}


\begin{table*}[tb!]
\centering
\caption{SGAV-2D deployment results on the PROMISE12 dataset. CR denotes compression ratio and UC denotes uncertainty correlation. Context Memory comprises activation memory and TensorRT kernel-internal scratch memory, while Total Peak Memory includes weight and Context Memory. Latency denotes median TensorRT inference latency, and Speedup is measured relative to the FP32 TensorRT baseline.}
\label{tab:sgav_deployment_promise12}
\scriptsize
\setlength{\tabcolsep}{3pt}
\resizebox{\textwidth}{!}{%
\begin{tabular}{llcccccccccccc}
\toprule
\textbf{Model} &
\textbf{Precision} &
\textbf{CR} &
\textbf{Latency (ms)} &
\textbf{Speedup} &
\textbf{Con-Mem (MiB)} &
\textbf{Dice} &
\textbf{ECE} &
\textbf{Brier} &
\textbf{AURC} &
\textbf{Err-AUROC} &
\textbf{UC.} &
\textbf{Weights (MiB)} &
\textbf{Total Mem. (MiB)} \\
\midrule

densenet121 & FP32
& 1.000 & 2.8741 & 1.000 & 91.8750
& 0.8962 & 0.0987 & 0.1977 & 0.0020 & 0.6314 & 0.9989
& 77.54 & 169.41 \\

& SGAV-MP
& 2.869 & 1.4886 & 1.930 & 46.2656
& 0.8943 & 0.1002 & 0.2007 & 0.0021 & 0.6298 & 0.9678
& 27.02 & 73.29 \\

& Uniform INT8
& 4.153 & 1.6815 & 1.709 & 49.5156
& 0.8946 & 0.1004 & 0.2010 & 0.0021 & 0.6299 & 0.9442
& 18.67 & 67.81 \\

\cmidrule{1-14}

efficientnet-b0 & FP32
& 1.000 & 1.2042 & 1.000 & 84.8125
& 0.8889 & 0.0989 & 0.1983 & 0.0022 & 0.6442 & 0.9981
& 35.04 & 119.85 \\

& SGAV-MP
& 2.250 & 0.7236 & 1.664 & 31.2188
& 0.8874 & 0.1040 & 0.2085 & 0.0023 & 0.6413 & 0.9690
& 15.58 & 46.80 \\

& Uniform INT8
& 3.812 & 0.8165 & 1.475 & 33.6875
& 0.8896 & 0.0977 & 0.1959 & 0.0021 & 0.6420 & 0.9445
& 9.20 & 43.75 \\

\cmidrule{1-14}

mobilenet\_v2 & FP32
& 1.000 & 0.8236 & 1.000 & 70.1875
& 0.8947 & 0.0880 & 0.1762 & 0.0018 & 0.6114 & 0.9987
& 13.62 & 83.81 \\

& SGAV-MP
& 2.407 & 0.5032 & 1.638 & 27.7344
& 0.8946 & 0.0882 & 0.1766 & 0.0018 & 0.6112 & 0.9703
& 5.66 & 33.40 \\

& Uniform INT8
& 3.679 & 0.5598 & 1.471 & 29.4688
& 0.8947 & 0.0883 & 0.1767 & 0.0018 & 0.6114 & 0.9688
& 3.70 & 33.17 \\

\cmidrule{1-14}

resnet18 & FP32
& 1.000 & 0.6734 & 1.000 & 78.1875
& 0.8999 & 0.0864 & 0.1735 & 0.0017 & 0.7164 & 0.9988
& 42.68 & 120.87 \\

& SGAV-MP
& 2.635 & 0.4195 & 1.606 & 30.0625
& 0.9000 & 0.0873 & 0.1753 & 0.0017 & 0.7162 & 0.9699
& 16.19 & 46.25 \\

& Uniform INT8
& 4.012 & 0.4628 & 1.455 & 32.5625
& 0.9000 & 0.0873 & 0.1752 & 0.0017 & 0.7160 & 0.9418
& 10.62 & 43.75 \\

\cmidrule{1-14}

resnet50 & FP32
& 1.000 & 0.9815 & 1.000 & 88.4375
& 0.8857 & 0.1026 & 0.2055 & 0.0023 & 0.5883 & 0.9982
& 94.11 & 182.55 \\

& SGAV-MP
& 2.527 & 0.5942 & 1.652 & 37.1250
& 0.8857 & 0.1026 & 0.2055 & 0.0023 & 0.5883 & 0.9698
& 37.25 & 74.18 \\

& Uniform INT8
& 4.082 & 0.6587 & 1.491 & 40.8125
& 0.8857 & 0.1026 & 0.2055 & 0.0023 & 0.5883 & 0.9426
& 23.06 & 65.81 \\

\bottomrule
\end{tabular}
}
\end{table*}



\begin{table*}[htbp!]
\centering
\caption{AMOS22 Dataset: SGAV and TensorRT Evaluation Results. ECE: Expected Calibration Error; AURC: Area Under the Risk--Coverage curve; Err-AUROC: Error Area Under the Receiver Operating Characteristic curve; U-Corr.: Uncertainty Correlation; Lat.: median inference latency; Mem.: total memory; GPU: peak GPU memory; CR: compression ratio; Build: TensorRT engine build time.}
\label{tab:amos22_sgav_results}

\fontsize{9}{10}\selectfont
\setlength{\tabcolsep}{2.5pt}

\begin{tabular}{@{}lccccccccccr@{}}
\toprule
Prec. & Dice & ECE & Brier & AURC & E-AUROC & U-Corr. & Lat.(ms) & Mem.(MB) & GPU(MB) & CR & Build(s) \\
\midrule
Ref 
& 0.7274 & 0.0657 & 0.2036 & 0.0007 & 0.8828 & 1.00 
& -- & -- & -- & -- & -- \\

fp32 
& 0.7273 & 0.0658 & 0.2039 & 0.0007 & 0.8828 & 1.00 
& 2.1862 & 141.4966 & 222 & 1.00 & 16.5689 \\

fp16 
& 0.7274 & 0.0657 & 0.2036 & 0.0007 & 0.8828 & 1.00 
& 0.8621 & 70.7607 & 152 & 1.99 & 35.7283 \\

MedPTQ 
& 0.7270 & 0.0661 & 0.2046 & 0.0007 & 0.8823 & 0.9946 
& 0.8428 & 88.9481 & 170 & 3.90 & 19.7355 \\

SGAV-MP 
& 0.7274 & 0.0654 & 0.2033 & 0.0007 & 0.8825 & 0.9972 
& 1.5800 & 136.9484 & 224 & 1.58 & 30.7248 \\

\bottomrule
\end{tabular}
\end{table*}
\begin{figure}[tb!]
\centering
\setlength{\tabcolsep}{2.5pt}

\begin{tabular}{ccccc}
\scriptsize Input & \scriptsize Ground truth & \scriptsize FP32 mask & \scriptsize INT4 mask & \scriptsize SGAV-MP mask \\
\includegraphics[width=0.185\textwidth]{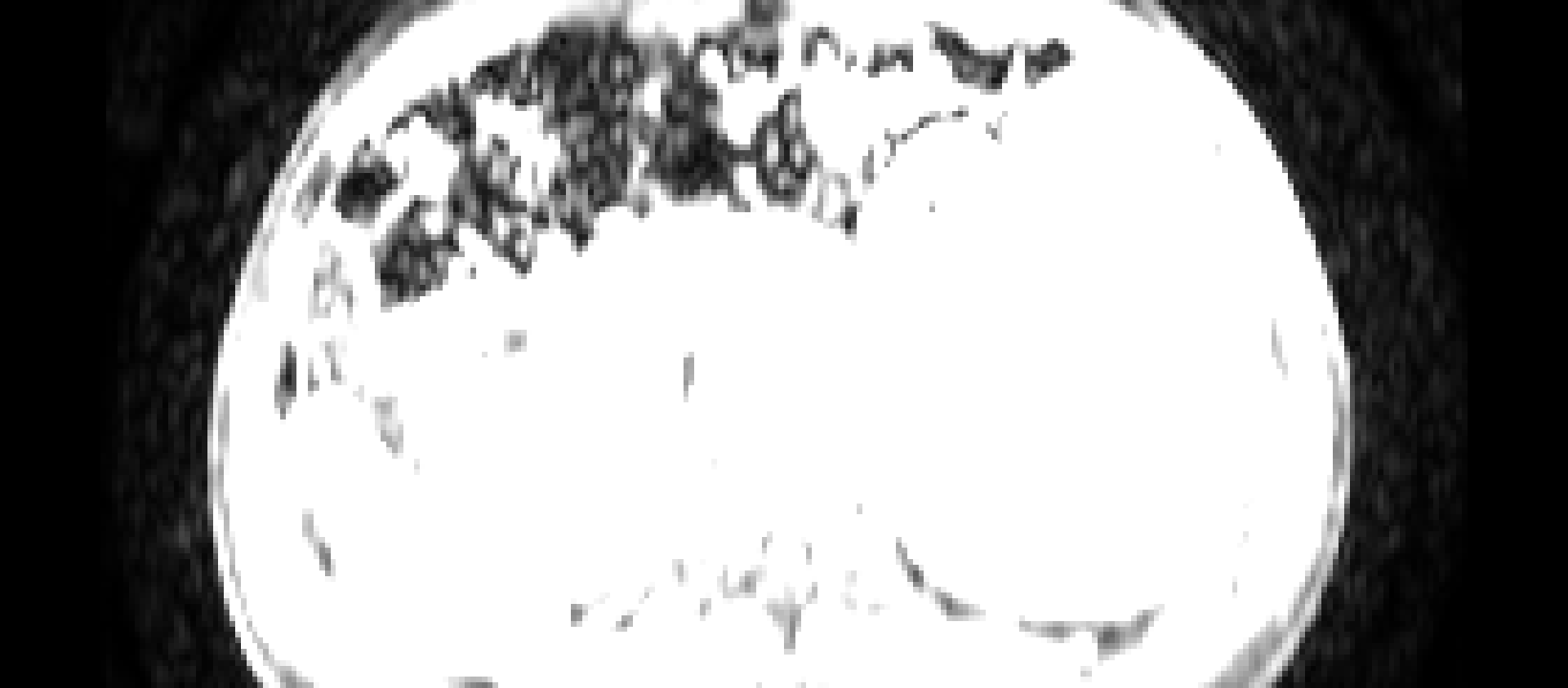} &
\includegraphics[width=0.185\textwidth]{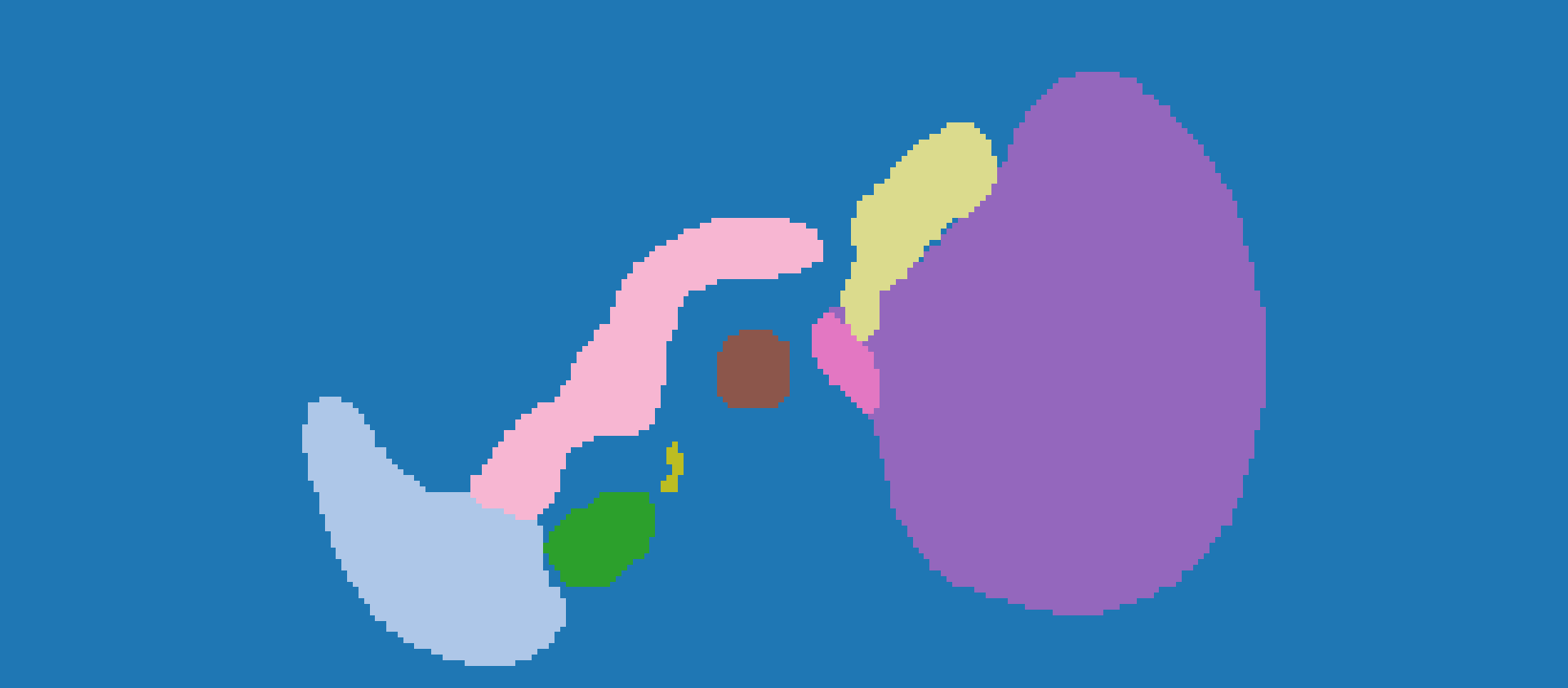} &
\includegraphics[width=0.185\textwidth]{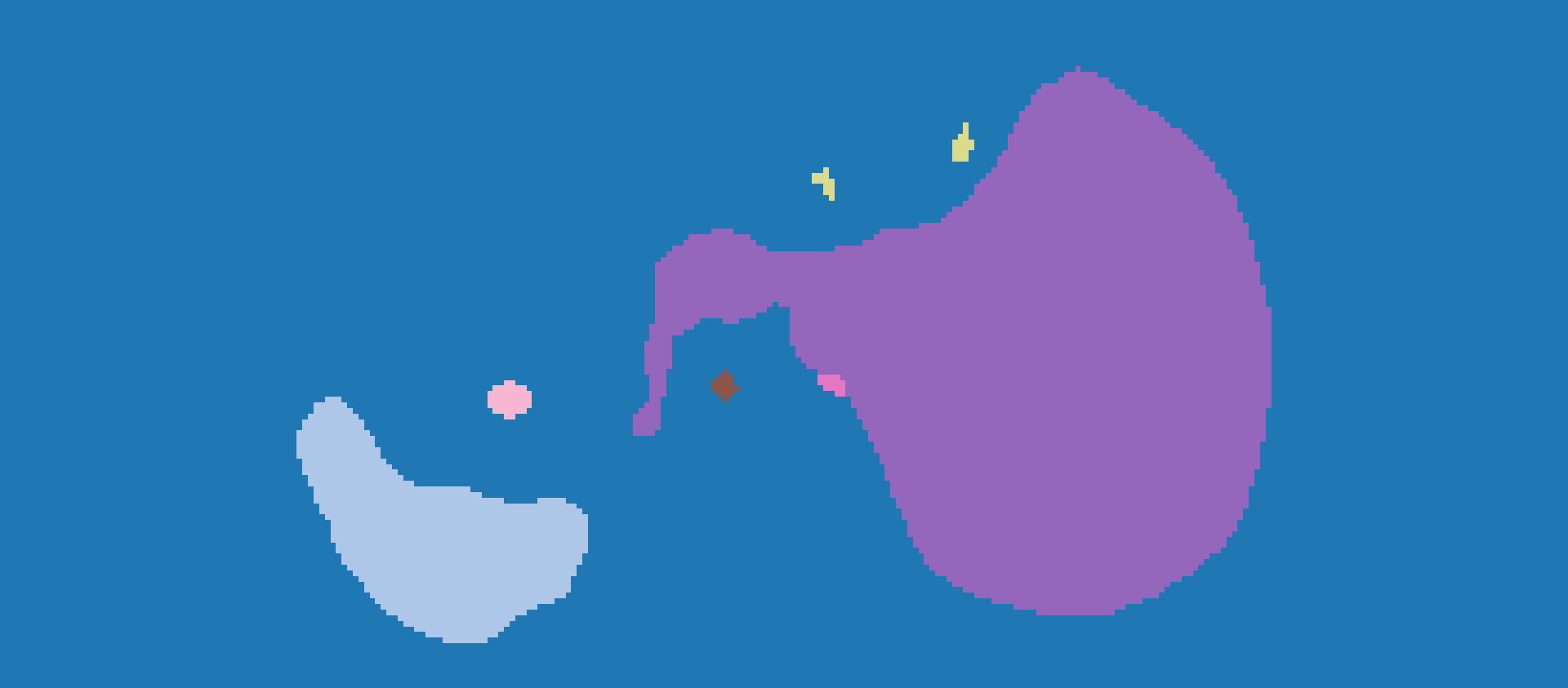} &
\includegraphics[width=0.185\textwidth]{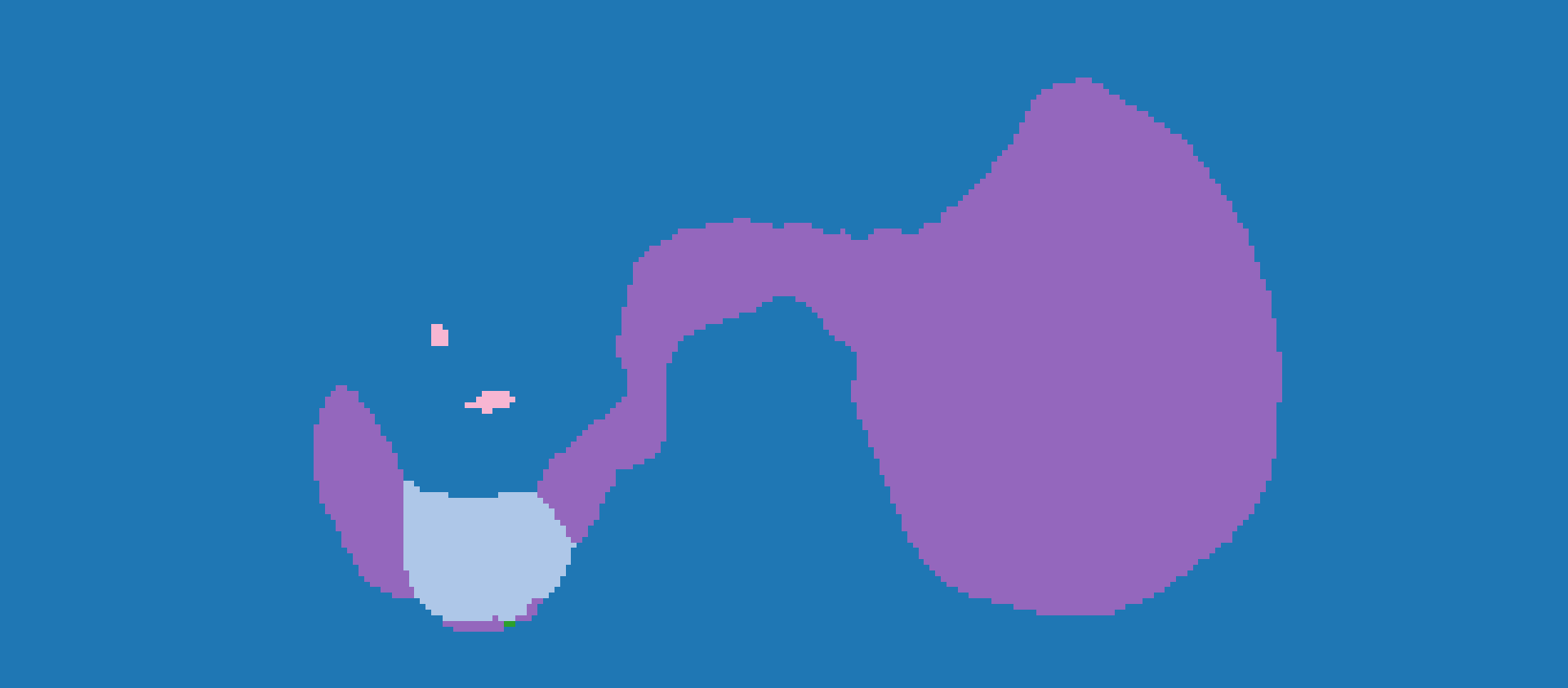} &
\includegraphics[width=0.185\textwidth]{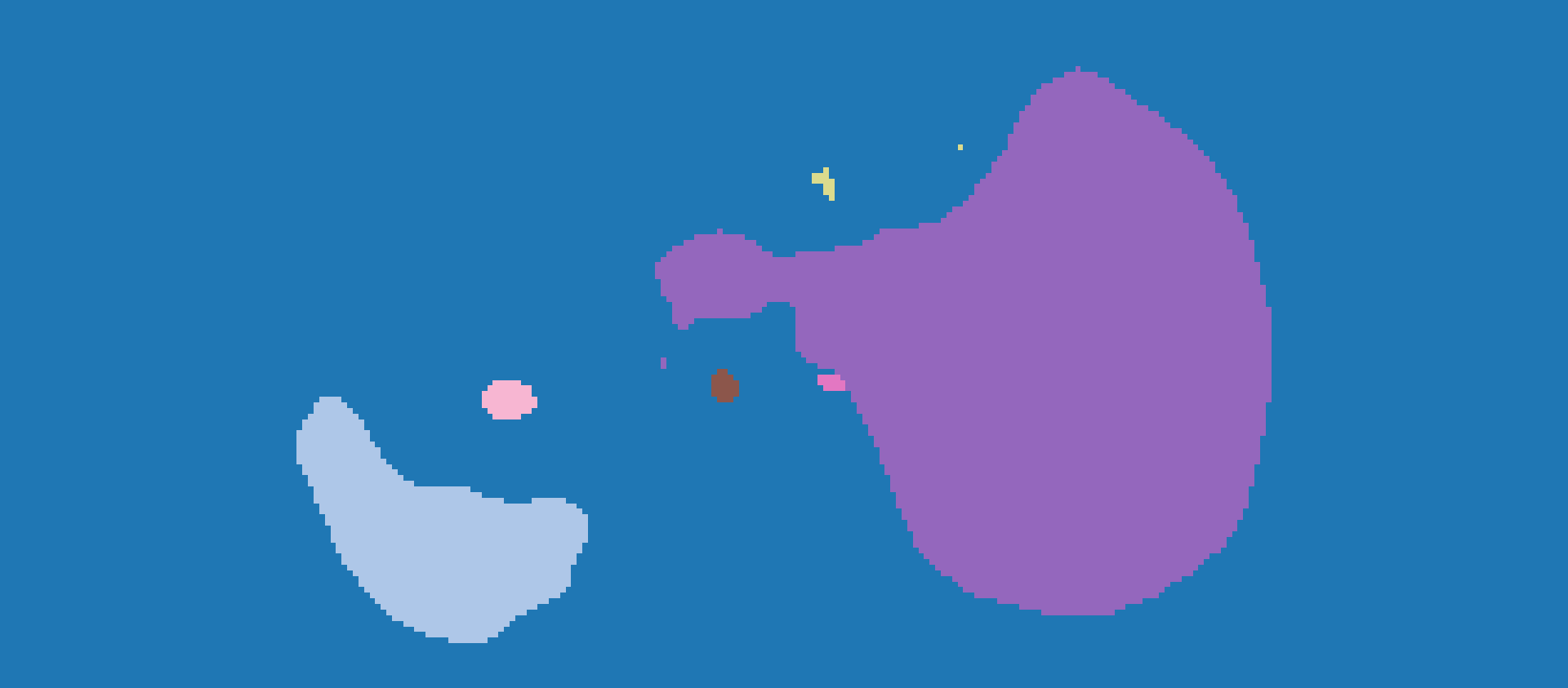} \\
\\[-0.45em]
\scriptsize SGAV-MP error & & \scriptsize FP32 uncertainty & \scriptsize INT4 uncertainty & \scriptsize SGAV-MP uncertainty \\
\includegraphics[width=0.185\textwidth]{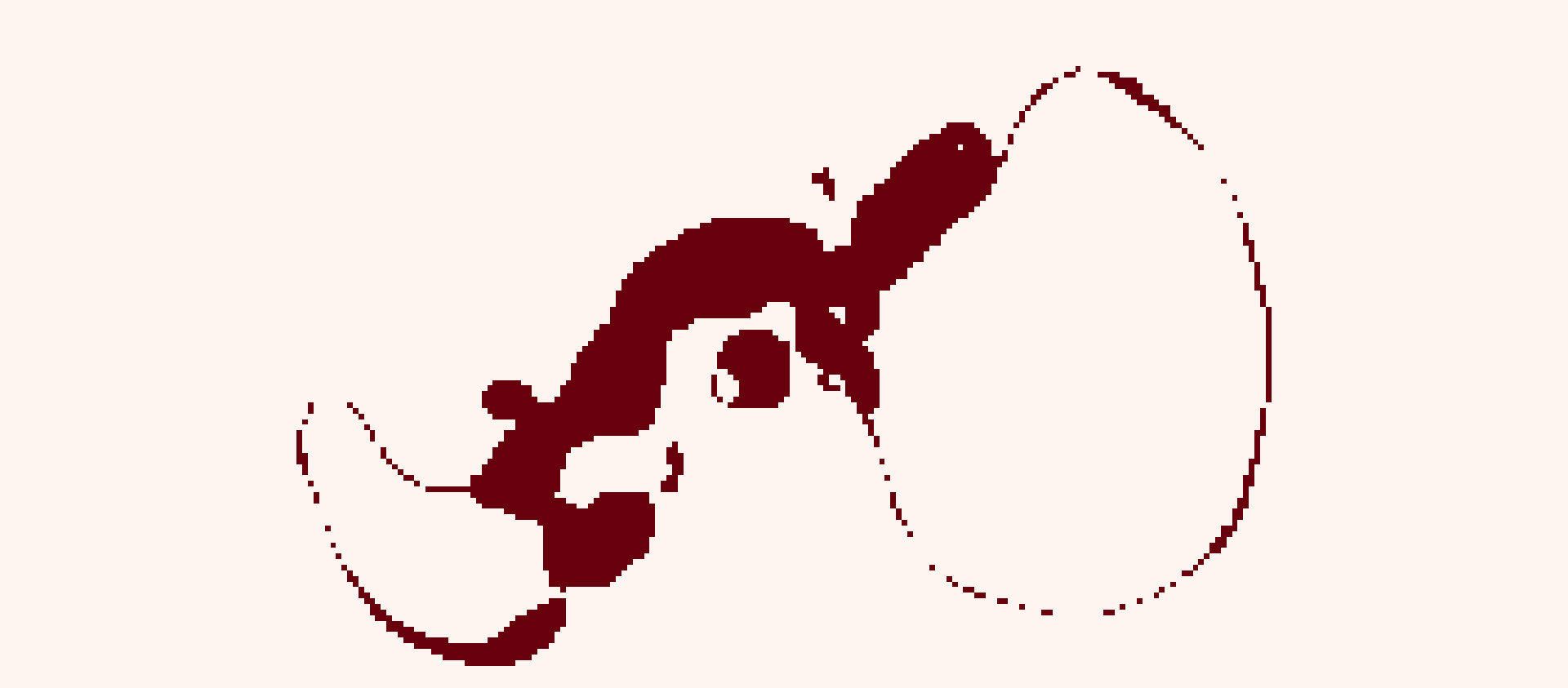} &
&
\includegraphics[width=0.185\textwidth]{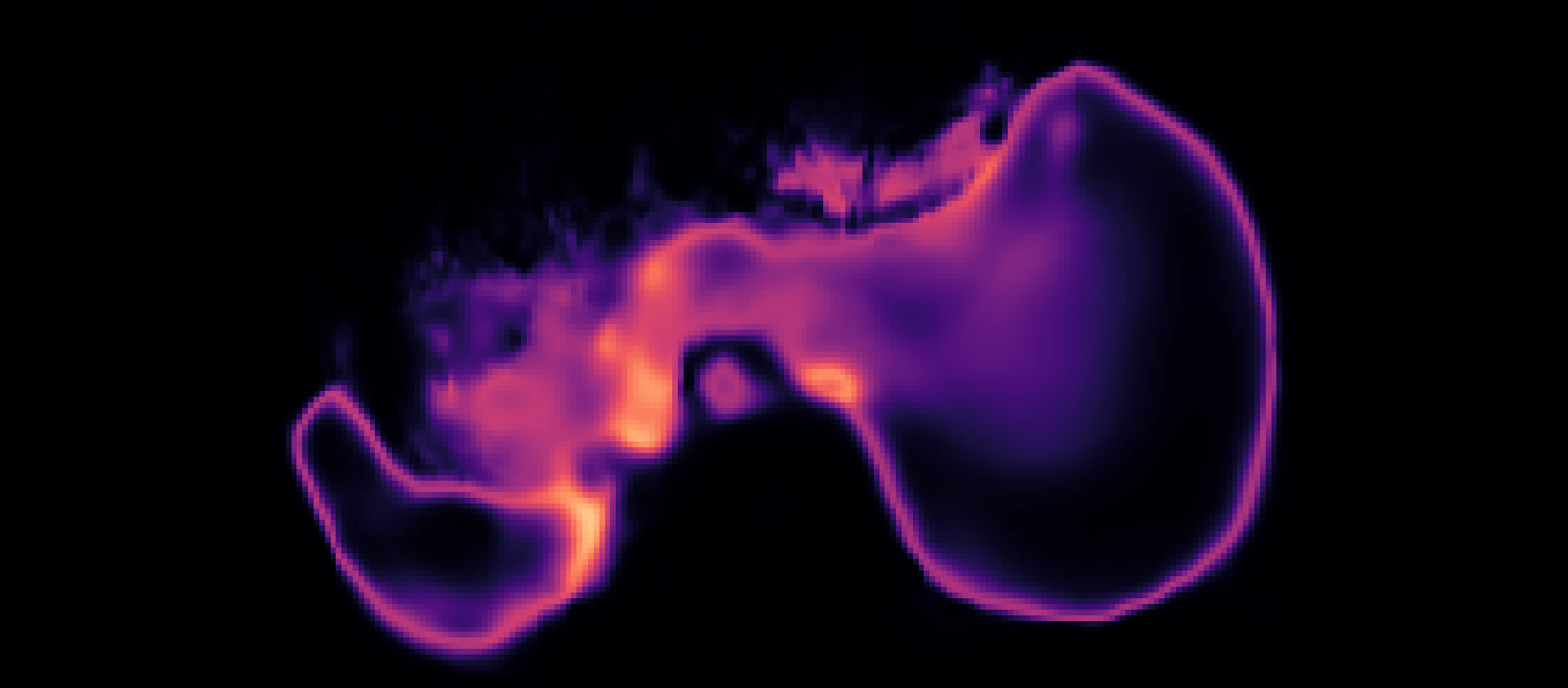} &
\includegraphics[width=0.185\textwidth]{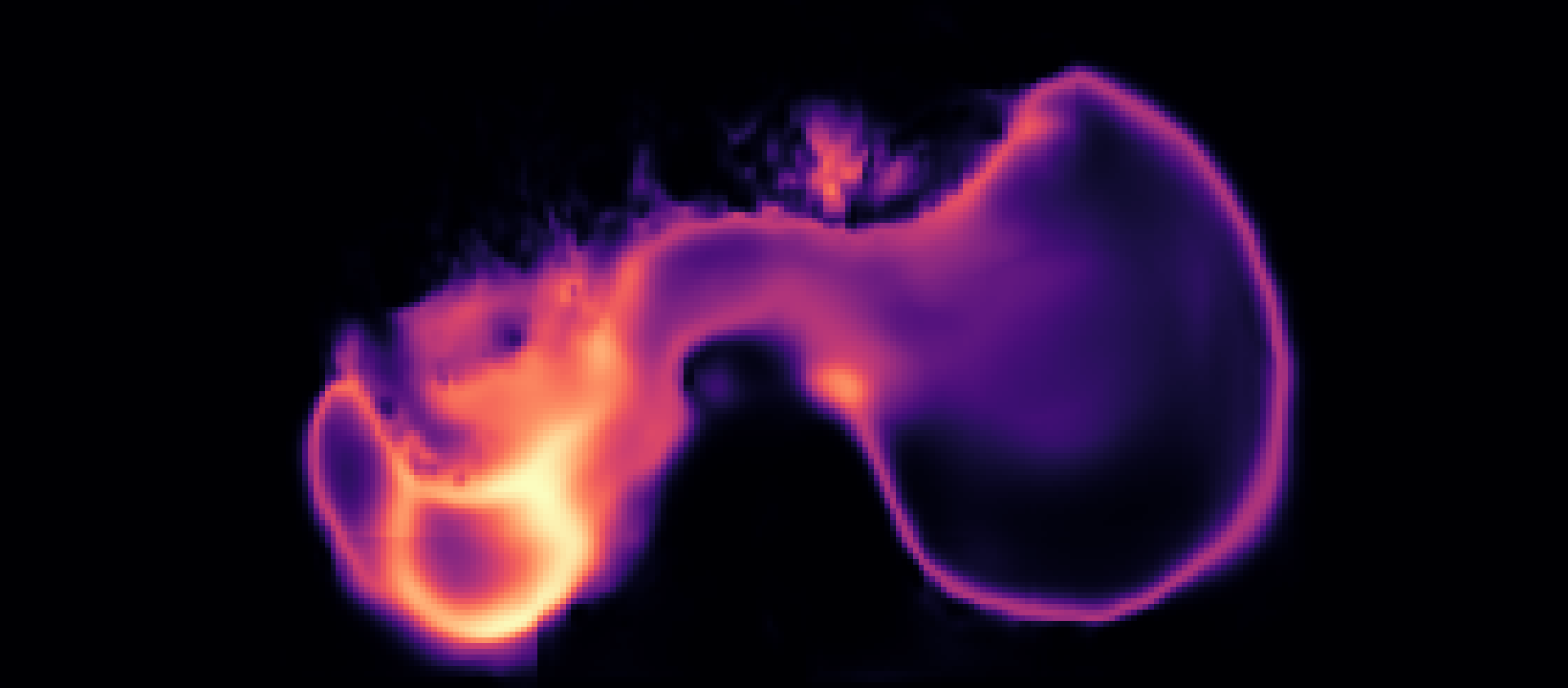} &
\includegraphics[width=0.185\textwidth]{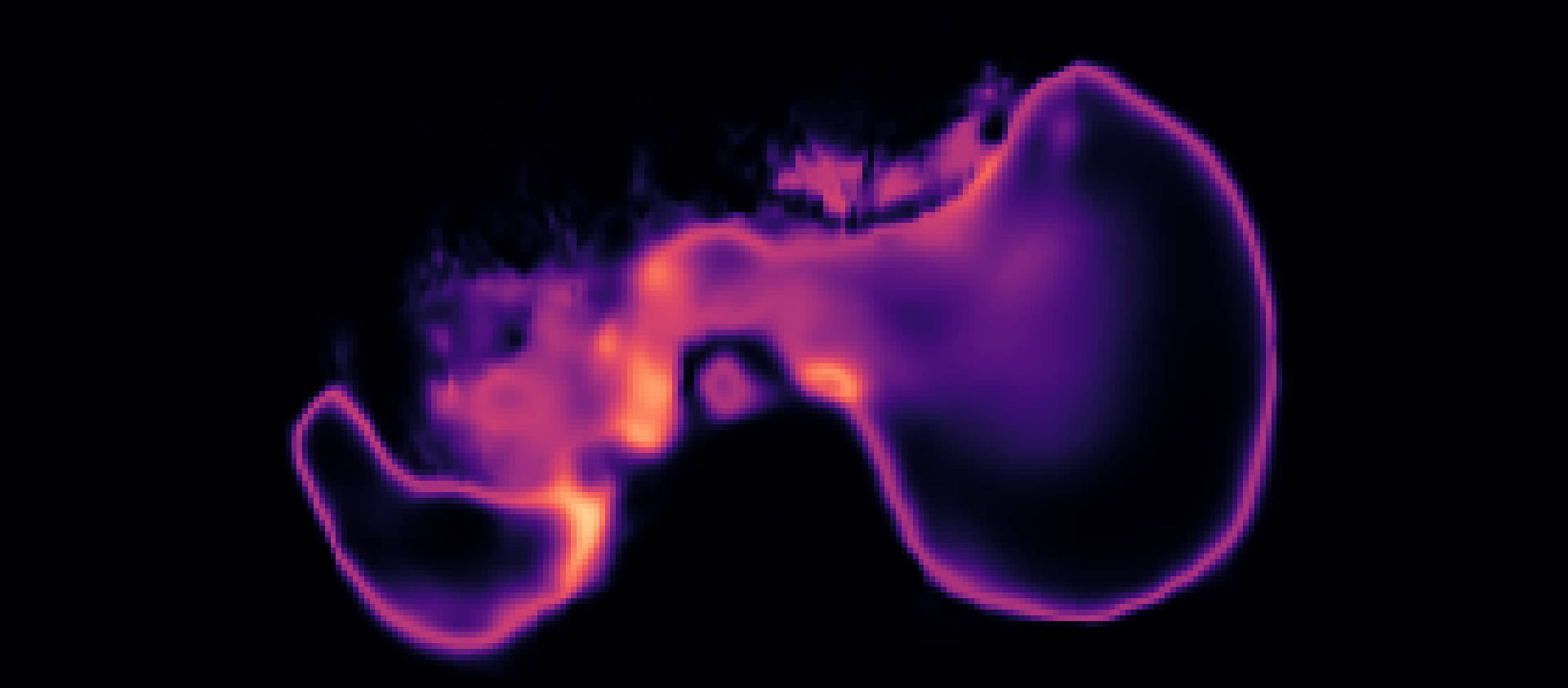}
\end{tabular}
\caption{Additional qualitative AMOS22 example using SegResNet, case 4, slice 124.}
\label{fig:supp_amos22_case4_z124}
\end{figure}

\begin{figure}[tb!]
\centering
\setlength{\tabcolsep}{1.5pt}
\begin{tabular}{ccccc}
\scriptsize Input & \scriptsize Ground truth & \scriptsize FP32 mask & \scriptsize INT4 mask & \scriptsize SGAV-MP mask \\
\includegraphics[width=0.185\textwidth]{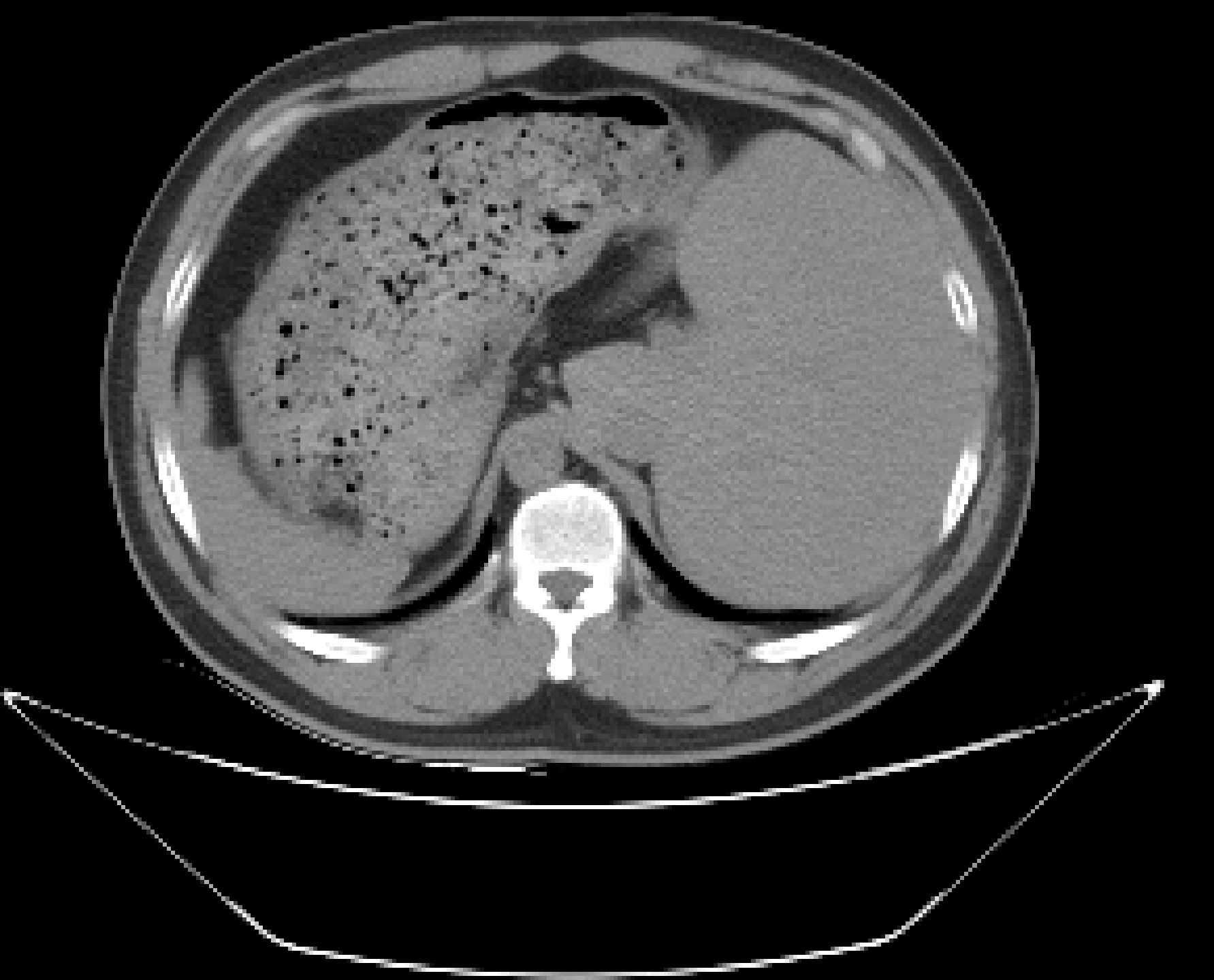} &
\includegraphics[width=0.185\textwidth]{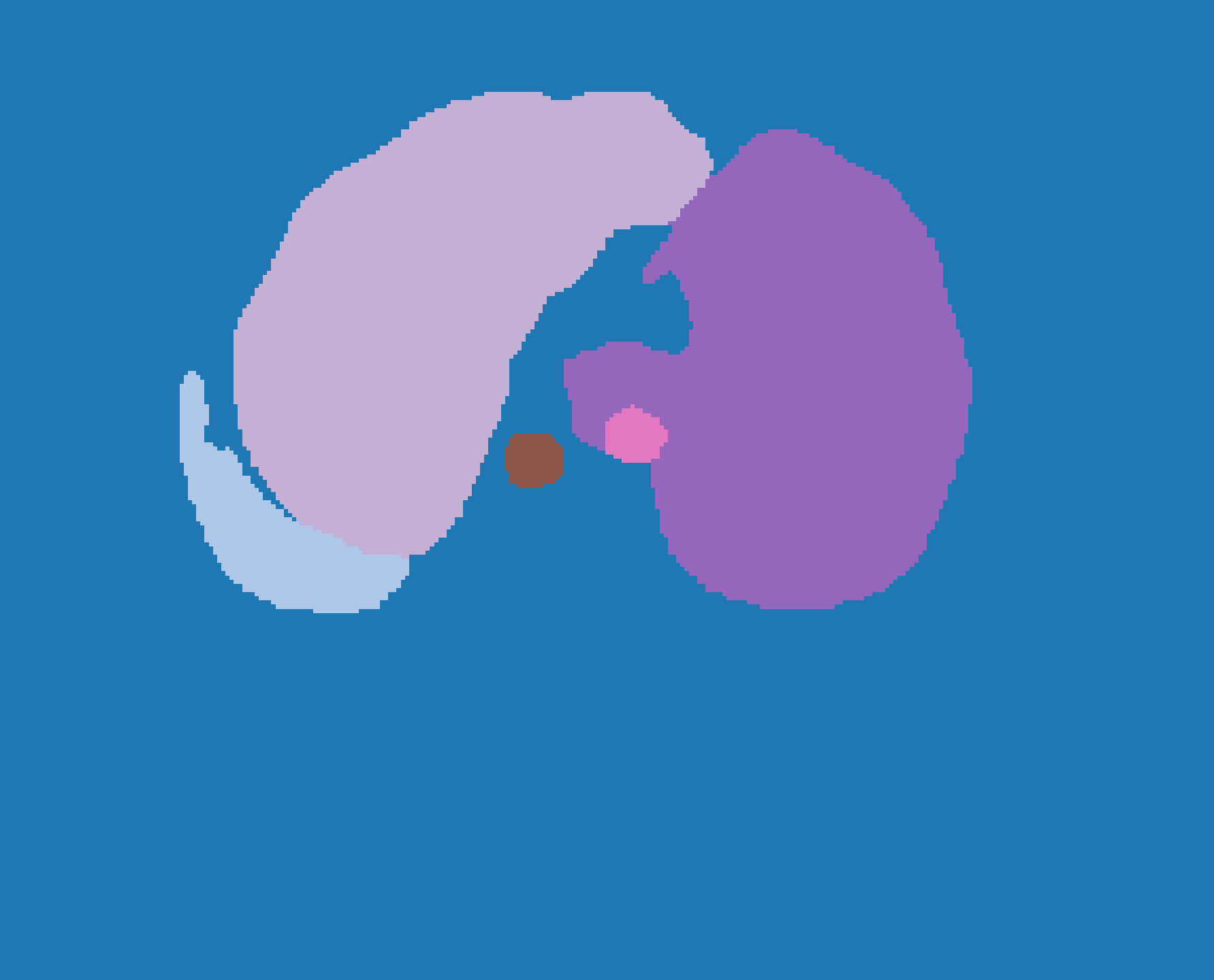} &
\includegraphics[width=0.185\textwidth]{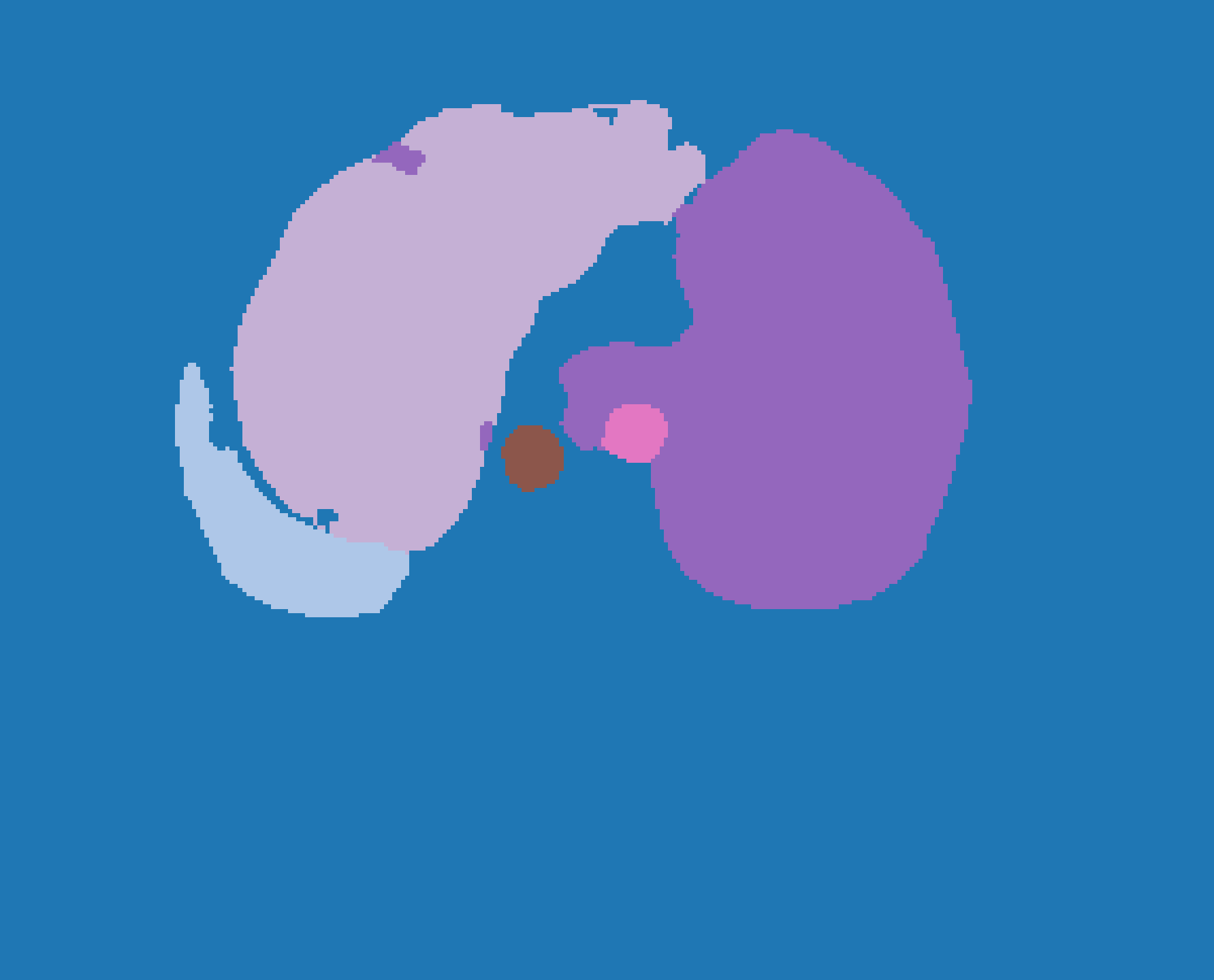} &
\includegraphics[width=0.185\textwidth]{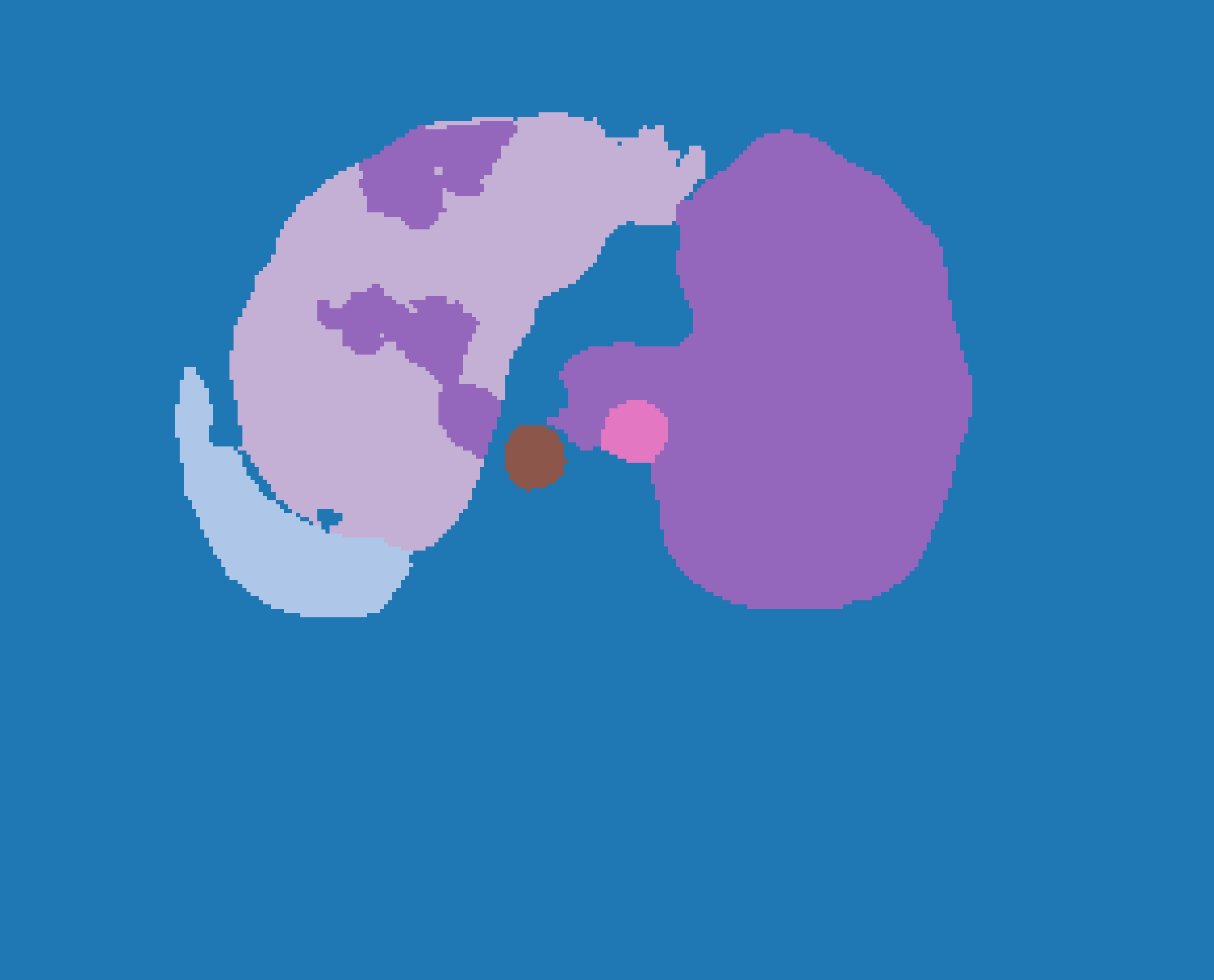} &
\includegraphics[width=0.185\textwidth]{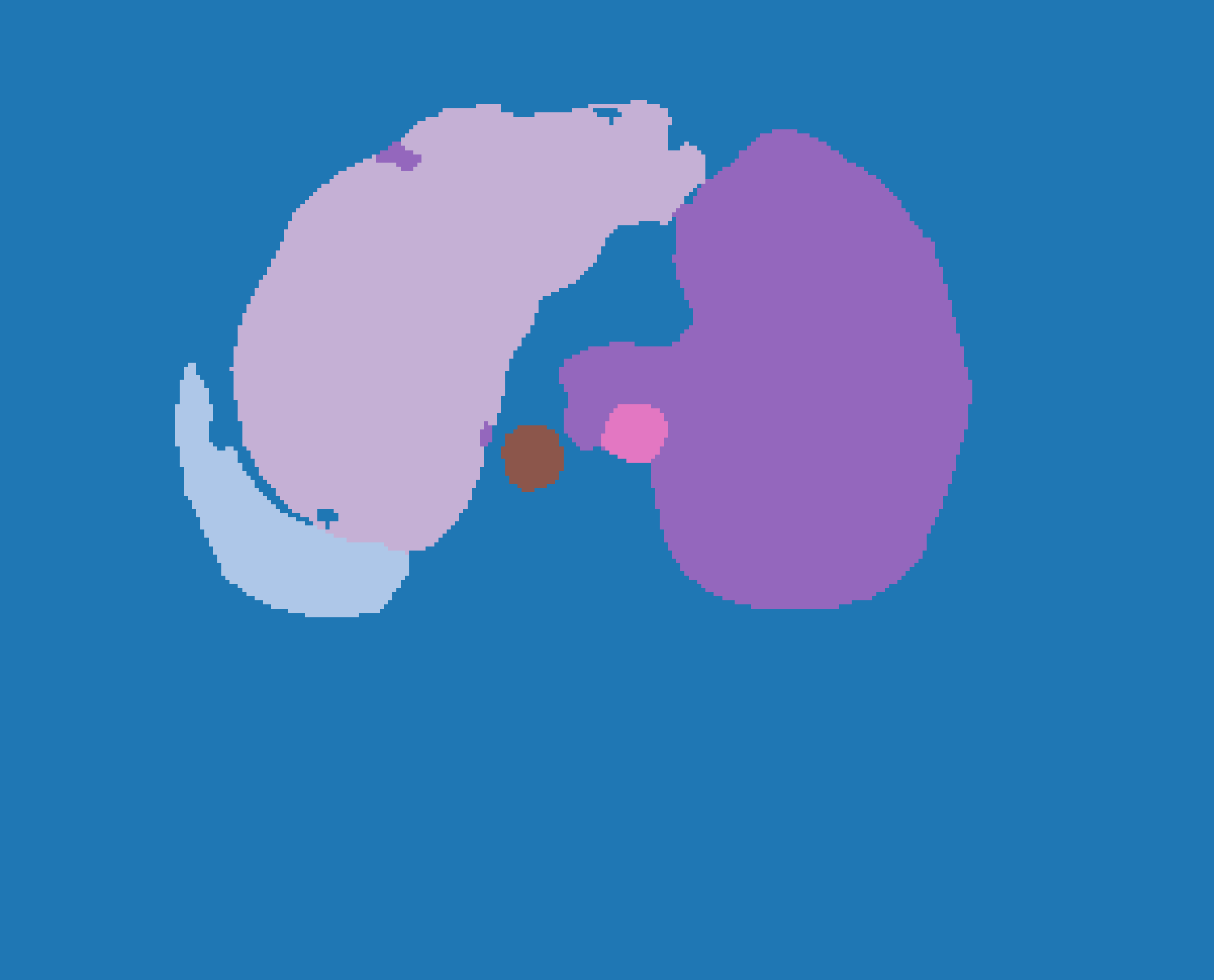} \\
\\[-0.45em]
\scriptsize SGAV-MP error & & \scriptsize FP32 uncertainty & \scriptsize INT4 uncertainty & \scriptsize SGAV-MP uncertainty \\
\includegraphics[width=0.185\textwidth]{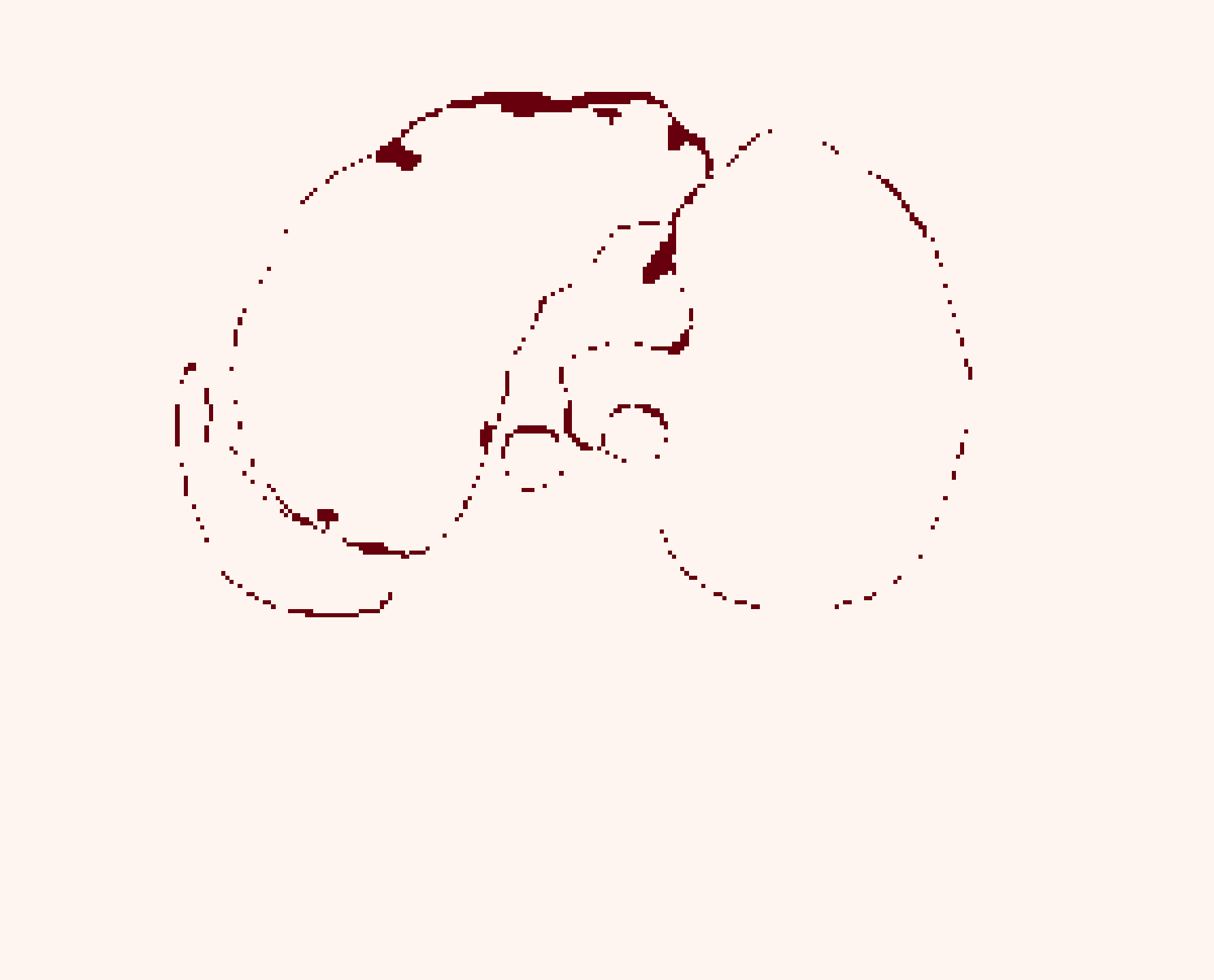} &
&
\includegraphics[width=0.185\textwidth]{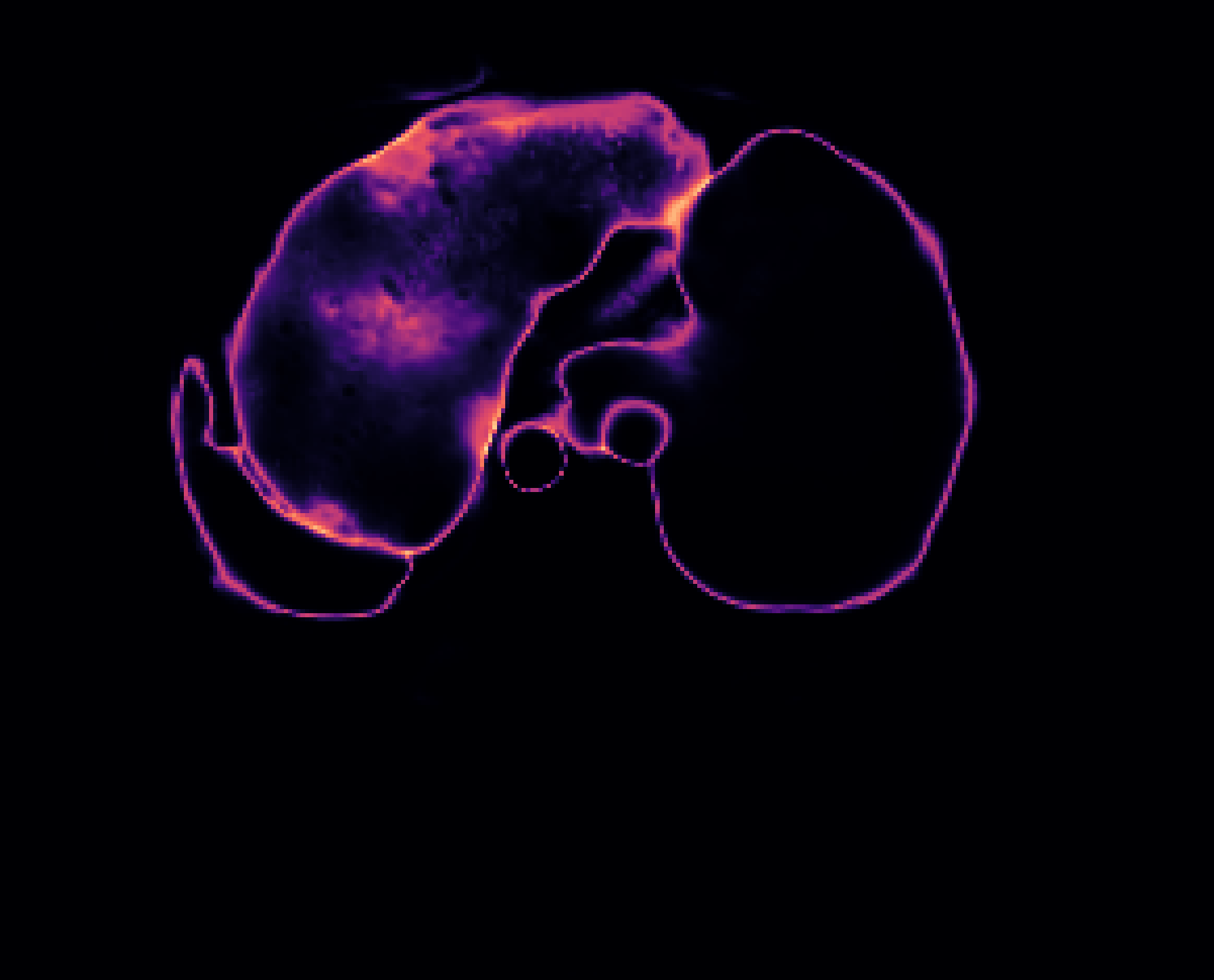} &
\includegraphics[width=0.185\textwidth]{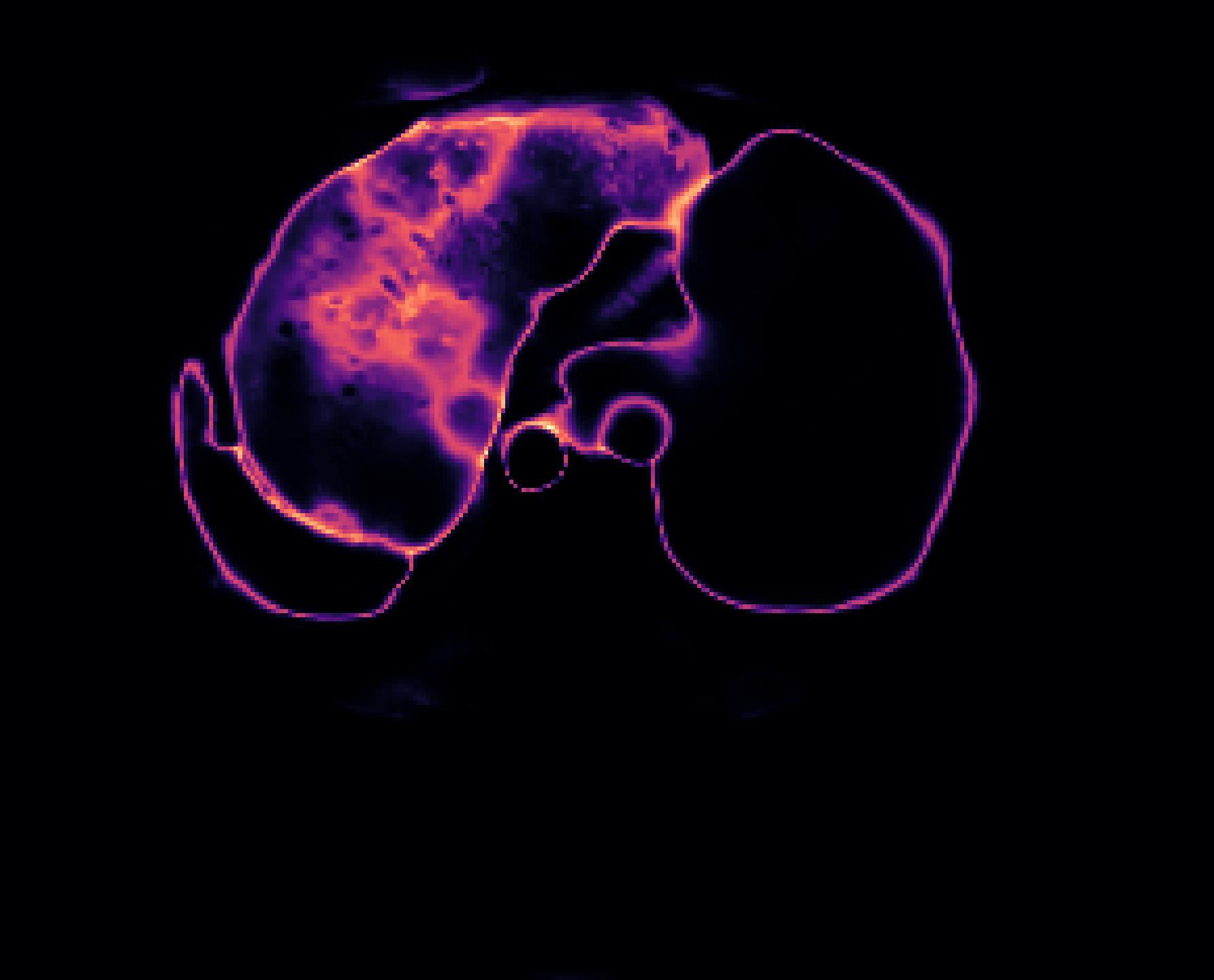} &
\includegraphics[width=0.185\textwidth]{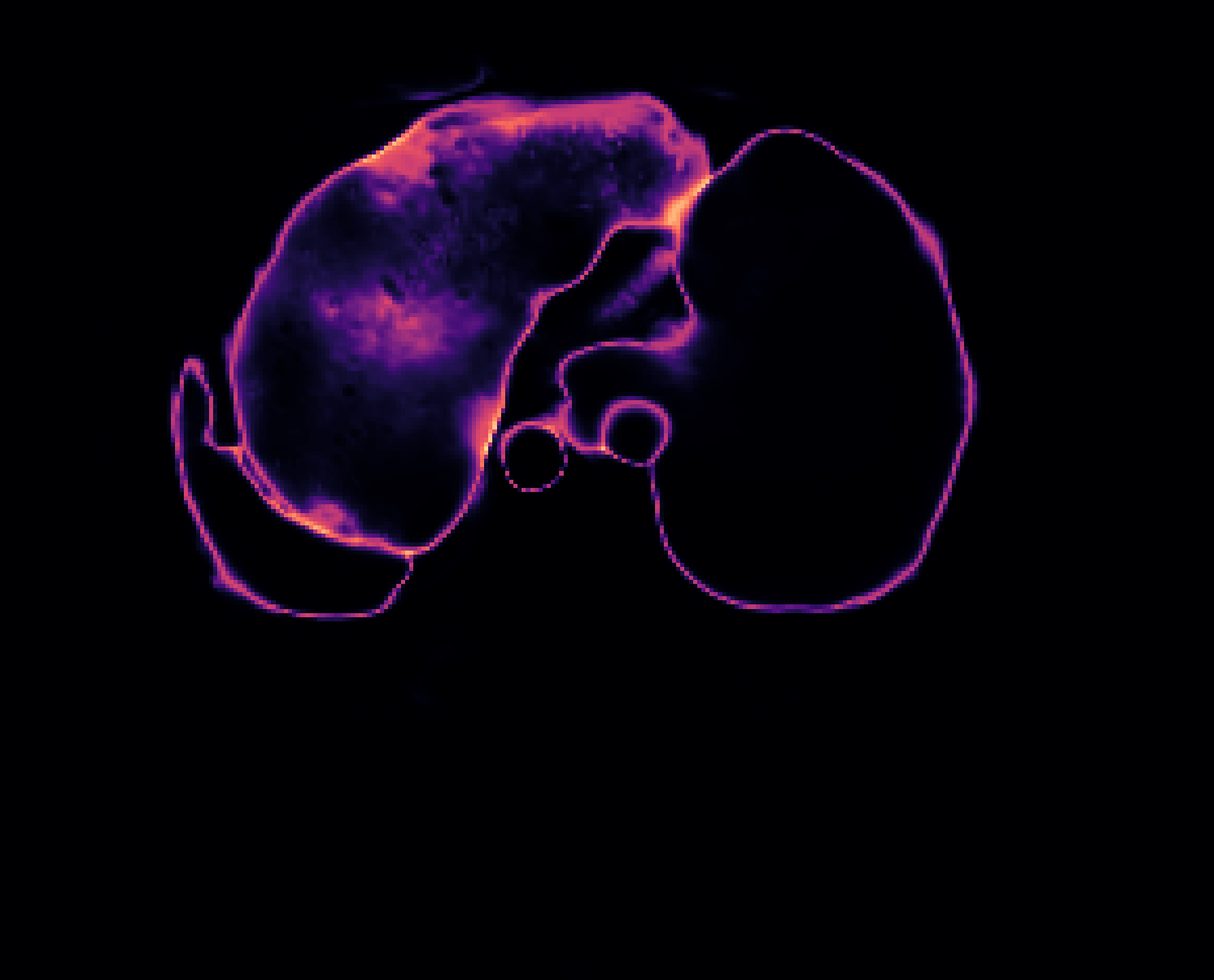}
\end{tabular}
\caption{Additional qualitative AMOS22 example using SegResNet, case 9, slice 217.}
\label{fig:supp_amos22_case9_z217}
\end{figure}

\clearpage
\section{Additional Qualitative Results}
\label{sec:supp_qualitative_results}
Figures~\hyperref[fig:supp_amos22_case4_z124]{\ref*{fig:supp_amos22_case4_z124}},
\hyperref[fig:supp_amos22_case9_z217]{\ref*{fig:supp_amos22_case9_z217}}, and
\hyperref[fig:supp_amos22_case15_z202]{\ref*{fig:supp_amos22_case15_z202}}
provide additional AMOS22 examples comparing the FP32, INT4, and SGAV-MP segmentation masks together with their corresponding uncertainty maps. Across these examples, the SGAV-MP masks remain visually close to the FP32 masks, while the INT4 masks exhibit more visible differences in the predicted segmentation regions. In Figure~\hyperref[fig:supp_amos22_case4_z124]{\ref*{fig:supp_amos22_case4_z124}}, the SGAV-MP prediction follows the FP32 prediction closely relative to the displayed ground truth, whereas the INT4 prediction shows a visibly different region. The uncertainty maps also show that SGAV-MP retains a spatial pattern closer to the FP32 uncertainty map than the INT4 result in the displayed example. Similar visual behavior is present in Figure~\hyperref[fig:supp_amos22_case9_z217]{\ref*{fig:supp_amos22_case9_z217}} and Figure~\hyperref[fig:supp_amos22_case15_z202]{\ref*{fig:supp_amos22_case15_z202}}, where the SGAV-MP segmentation remains visually similar to FP32 and its uncertainty structure is closer to the FP32 map than the corresponding INT4 map. These examples provide visual evidence consistent with the quantitative differences in uncertainty correlation reported for the compression configurations.

Figures~\hyperref[fig:supp_brats_case3_z76]{\ref*{fig:supp_brats_case3_z76}},
\hyperref[fig:supp_brats_case8_z59]{\ref*{fig:supp_brats_case8_z59}}, and
\hyperref[fig:supp_brats_case10_z82]{\ref*{fig:supp_brats_case10_z82}}
provide additional qualitative examples from BraTS. In Figure~\hyperref[fig:supp_brats_case3_z76]{\ref*{fig:supp_brats_case3_z76}}, the SGAV-MP mask is visually close to the FP32 mask, whereas the INT4 mask shows a more noticeable difference in the segmented region. The corresponding INT4 uncertainty map also exhibits a visibly different spatial pattern from the FP32 and SGAV-MP maps. Figure~\hyperref[fig:supp_brats_case8_z59]{\ref*{fig:supp_brats_case8_z59}} similarly shows close visual correspondence between the FP32 and SGAV-MP masks, while the INT4 prediction differs in portions of the segmented region. In Figure~\hyperref[fig:supp_brats_case10_z82]{\ref*{fig:supp_brats_case10_z82}}, the FP32 and SGAV-MP masks again appear visually similar, while the INT4 uncertainty map contains a different spatial distribution of highlighted regions. Across these displayed cases, the qualitative examples therefore show that the SGAV-MP predictions and uncertainty maps remain visually closer to the FP32 outputs than the INT4 outputs. These observations are qualitative and are intended to complement, rather than replace, the quantitative uncertainty and segmentation measurements reported in Tables~\hyperref[tab:supp_3d_full]{\ref*{tab:supp_3d_full}} and~\hyperref[tab:supp_promise_full]{\ref*{tab:supp_promise_full}}.

\begin{figure}[tb!]
\centering
\setlength{\tabcolsep}{1.5pt}
\begin{tabular}{ccccc}
\scriptsize Input & \scriptsize Ground truth & \scriptsize FP32 mask & \scriptsize INT4 mask & \scriptsize SGAV-MP mask \\
\includegraphics[width=0.185\textwidth]{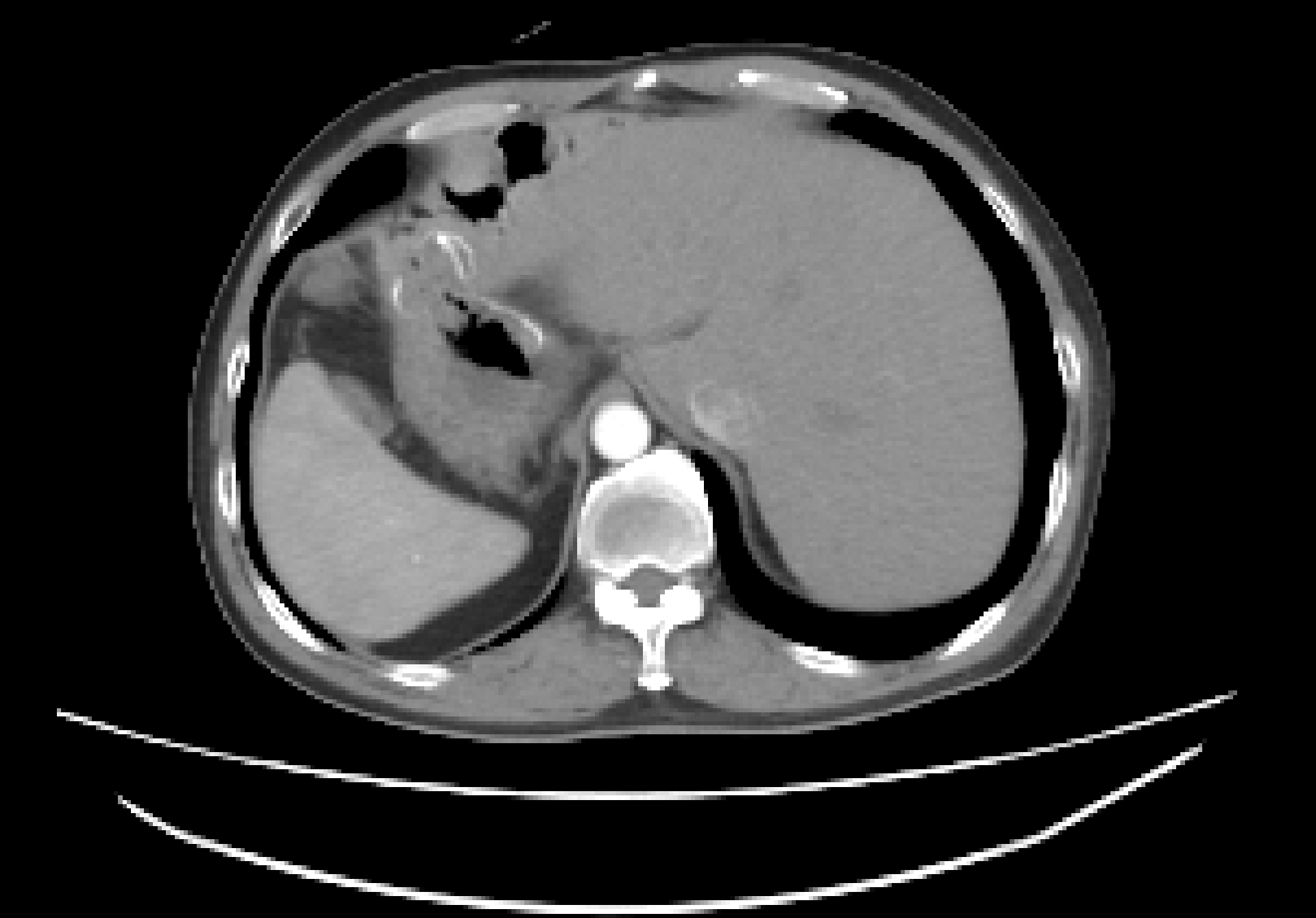} &
\includegraphics[width=0.185\textwidth]{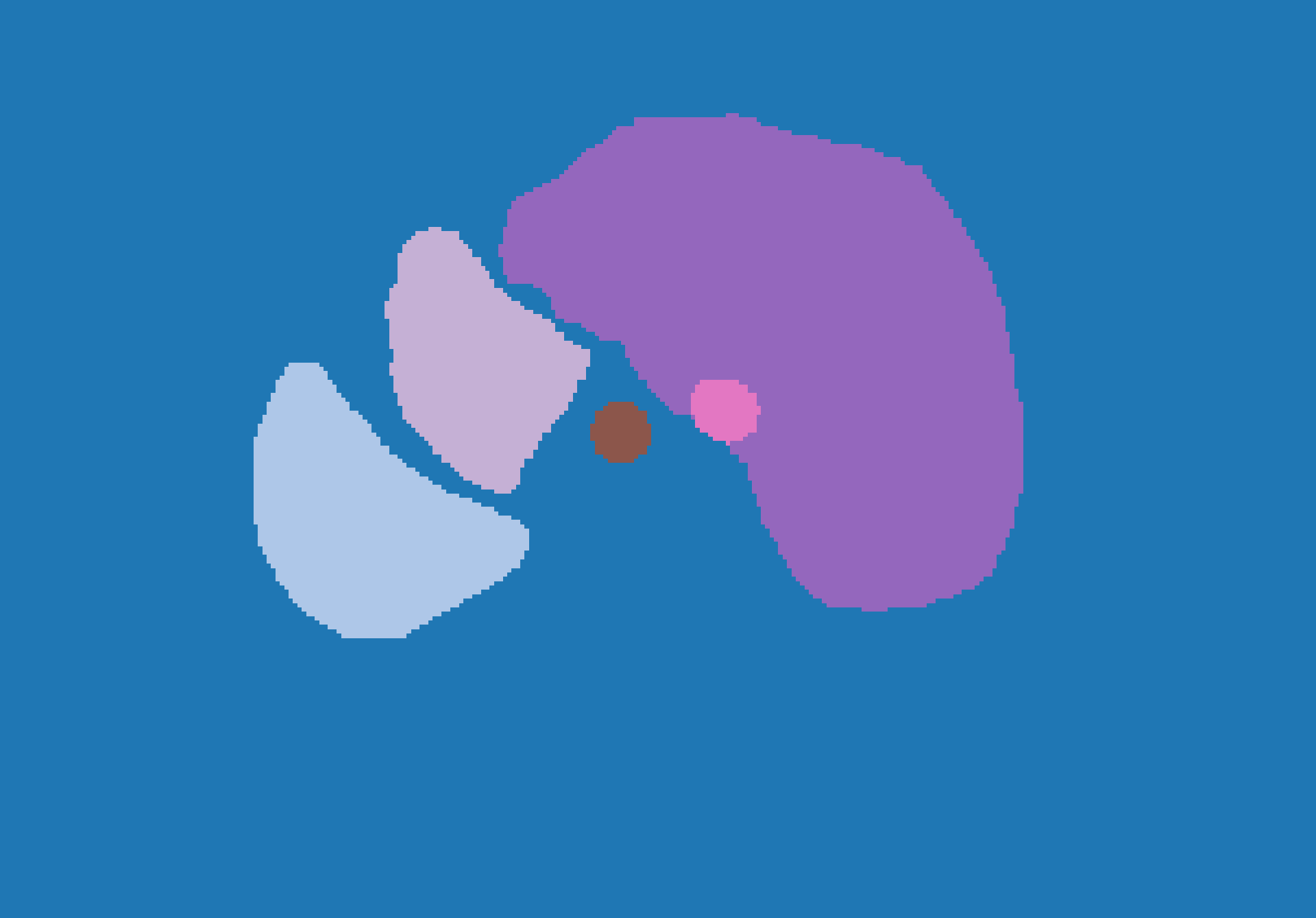} &
\includegraphics[width=0.185\textwidth]{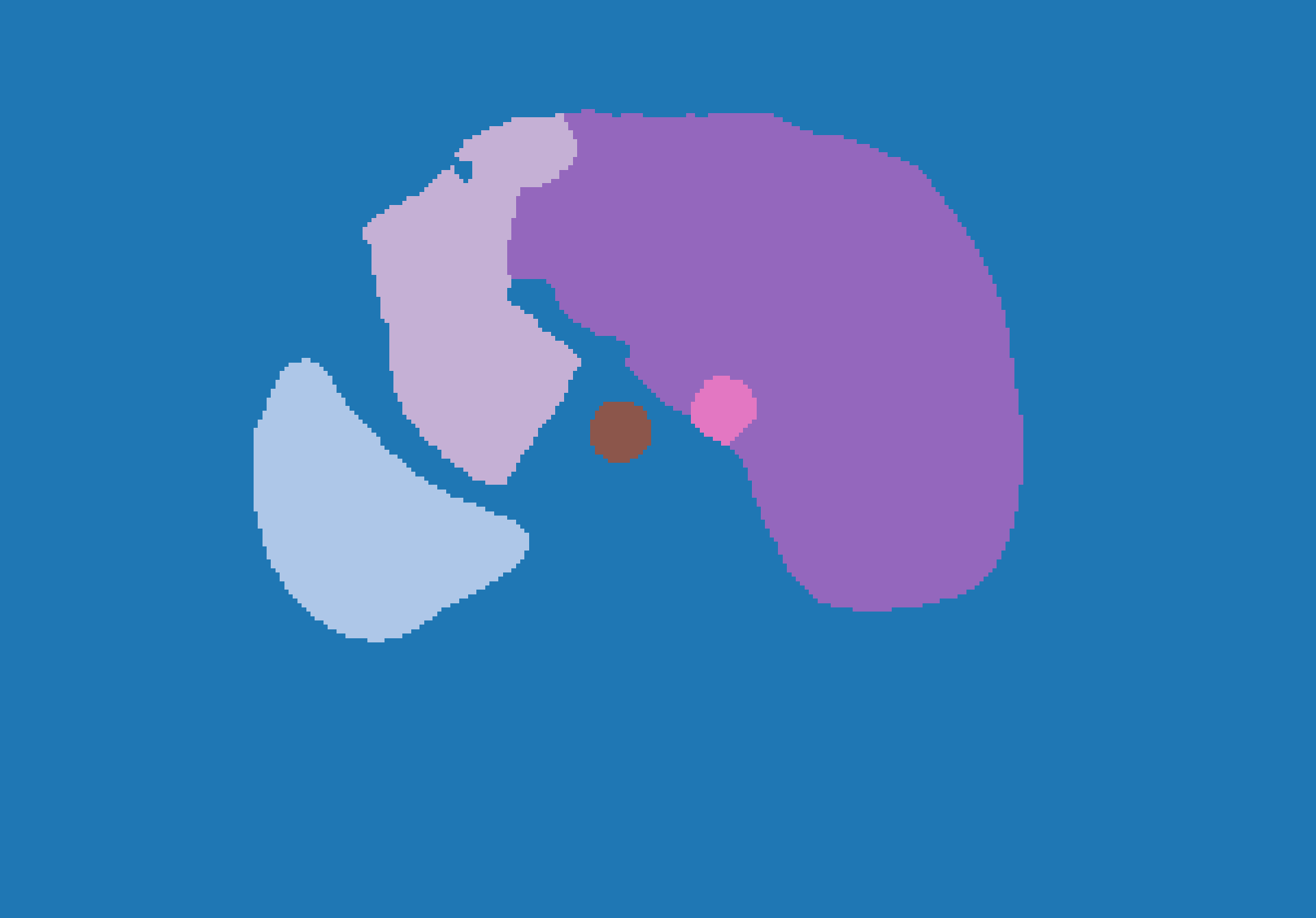} &
\includegraphics[width=0.185\textwidth]{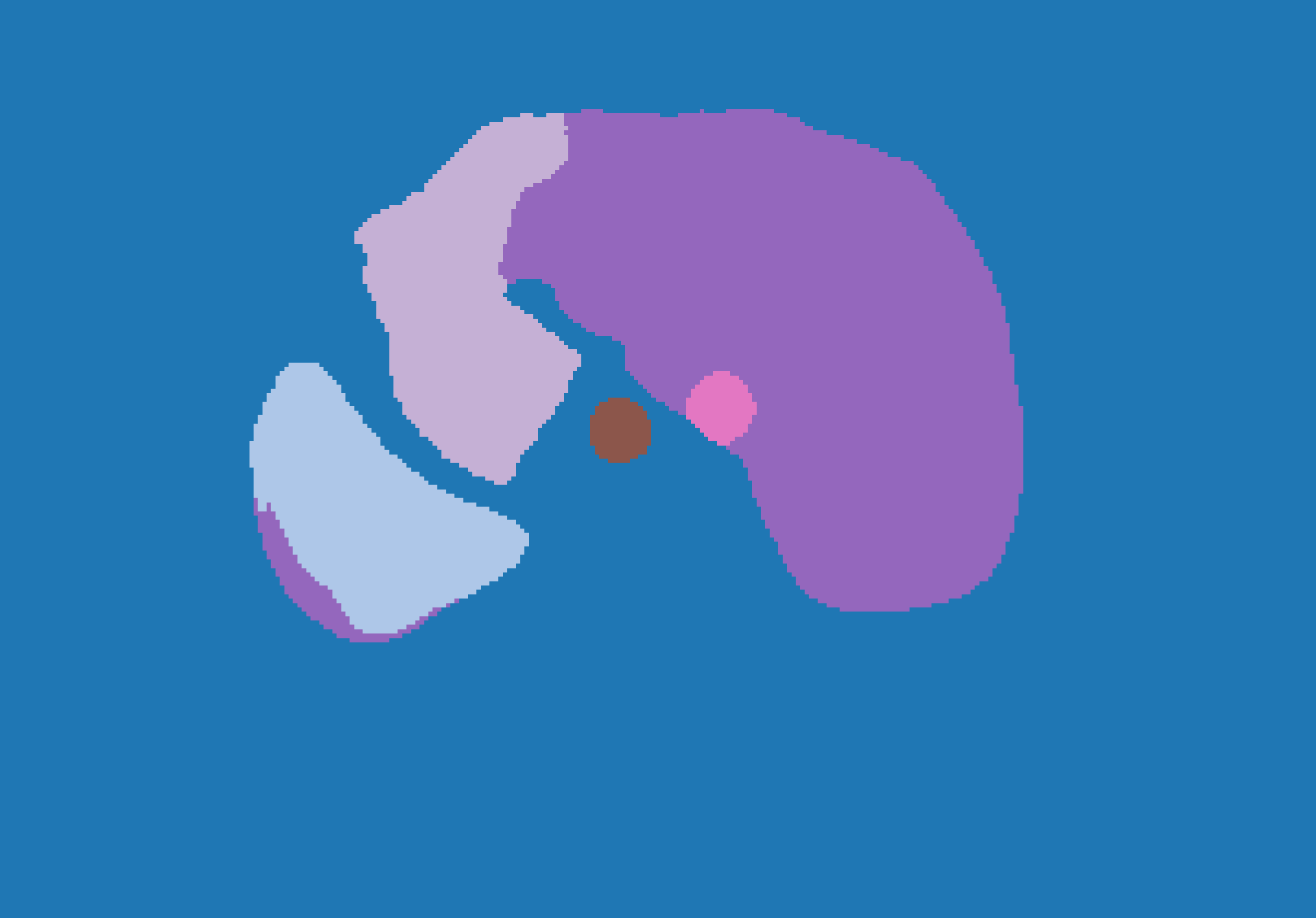} &
\includegraphics[width=0.185\textwidth]{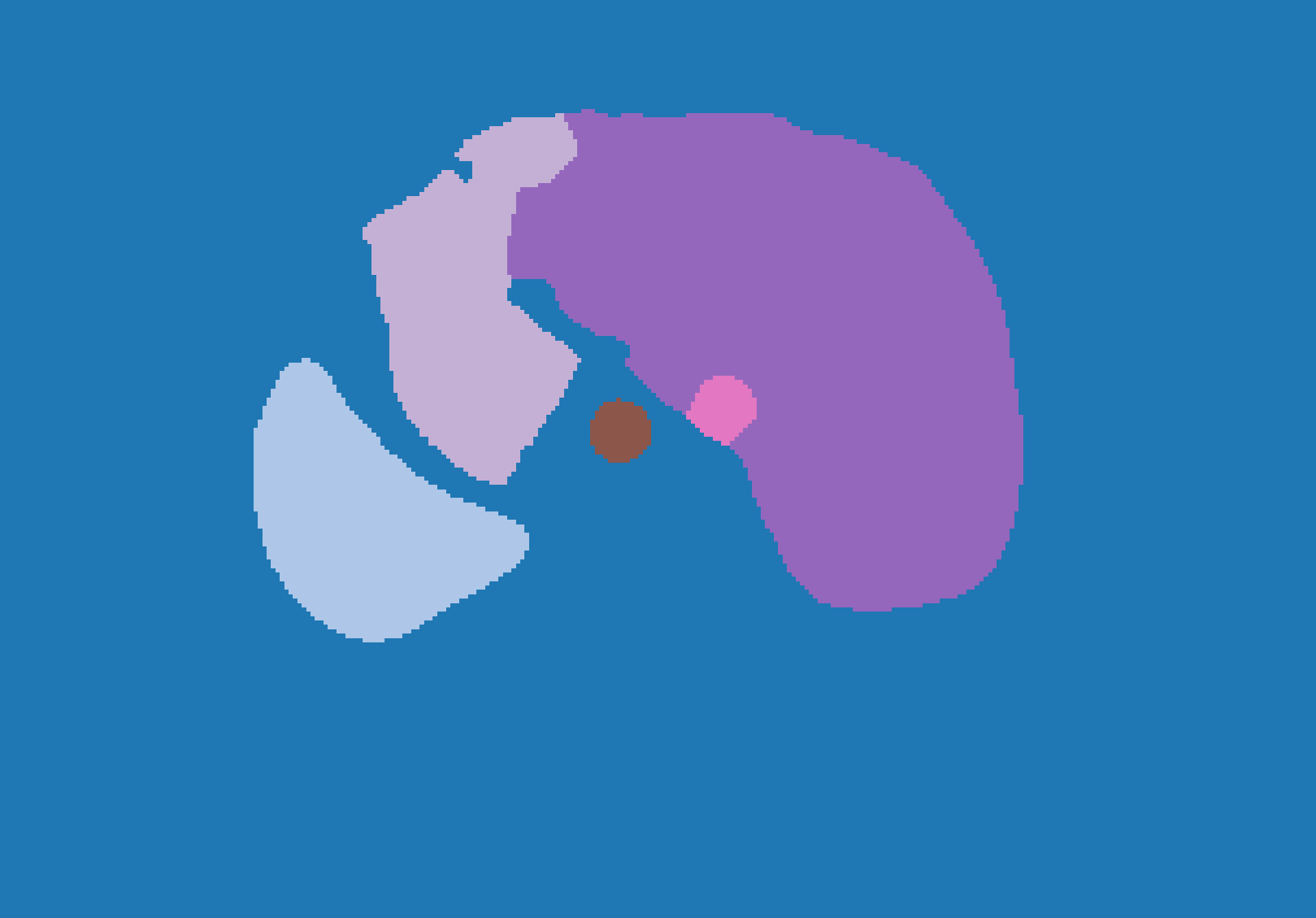} \\
\\[-0.45em]
\scriptsize SGAV-MP error & & \scriptsize FP32 uncertainty & \scriptsize INT4 uncertainty & \scriptsize SGAV-MP uncertainty \\
\includegraphics[width=0.185\textwidth]{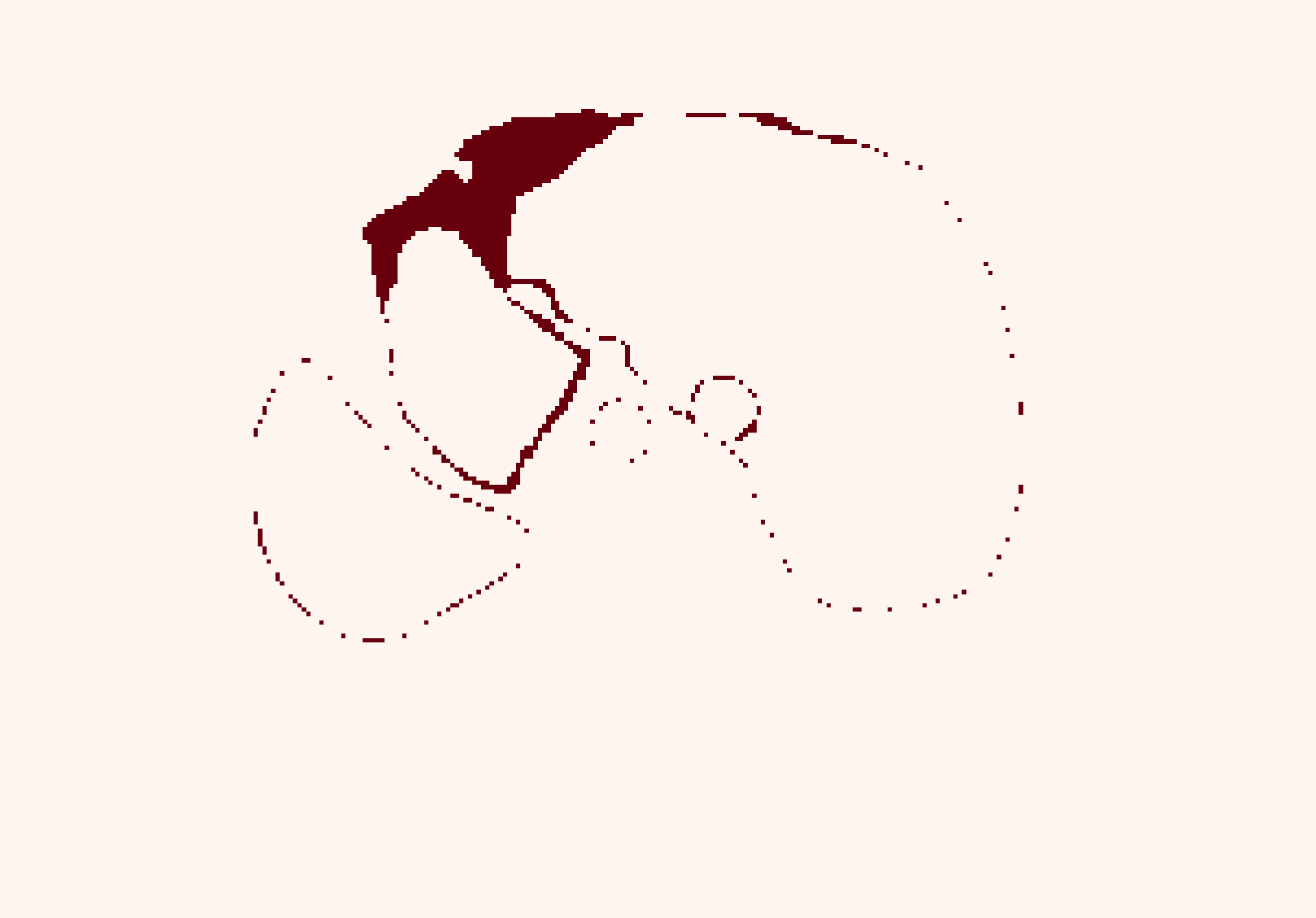} &
&
\includegraphics[width=0.185\textwidth]{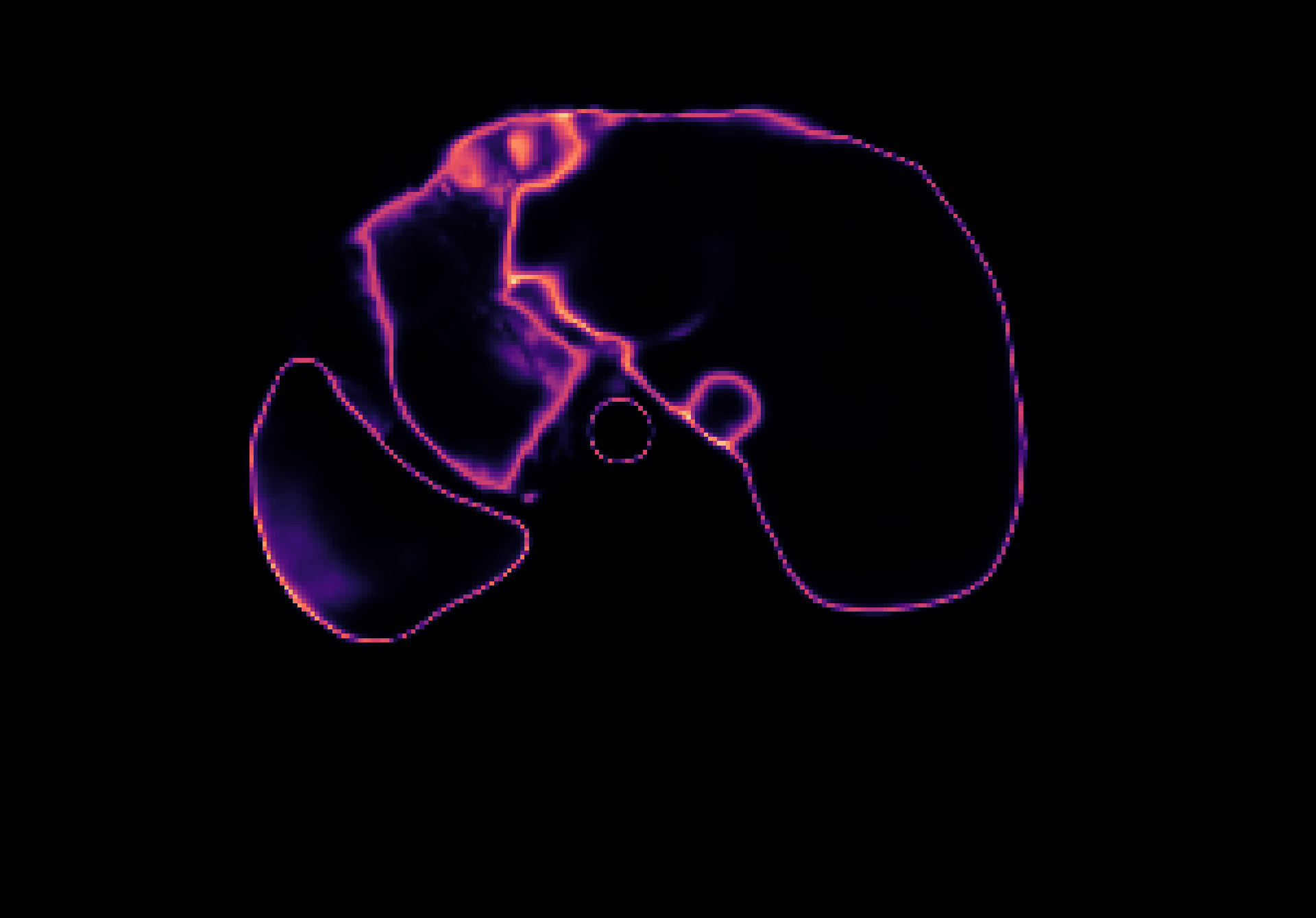} &
\includegraphics[width=0.185\textwidth]{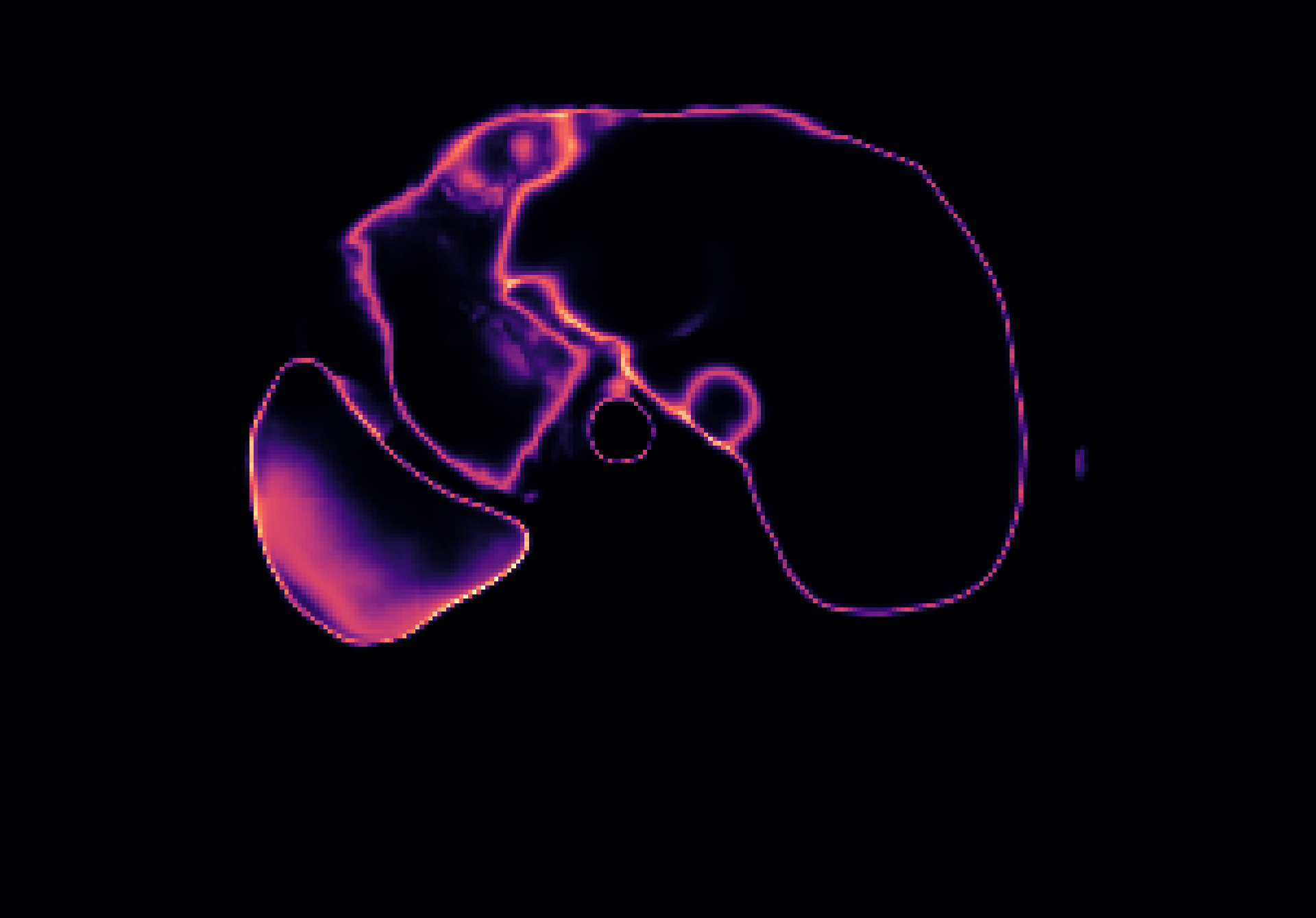} &
\includegraphics[width=0.185\textwidth]{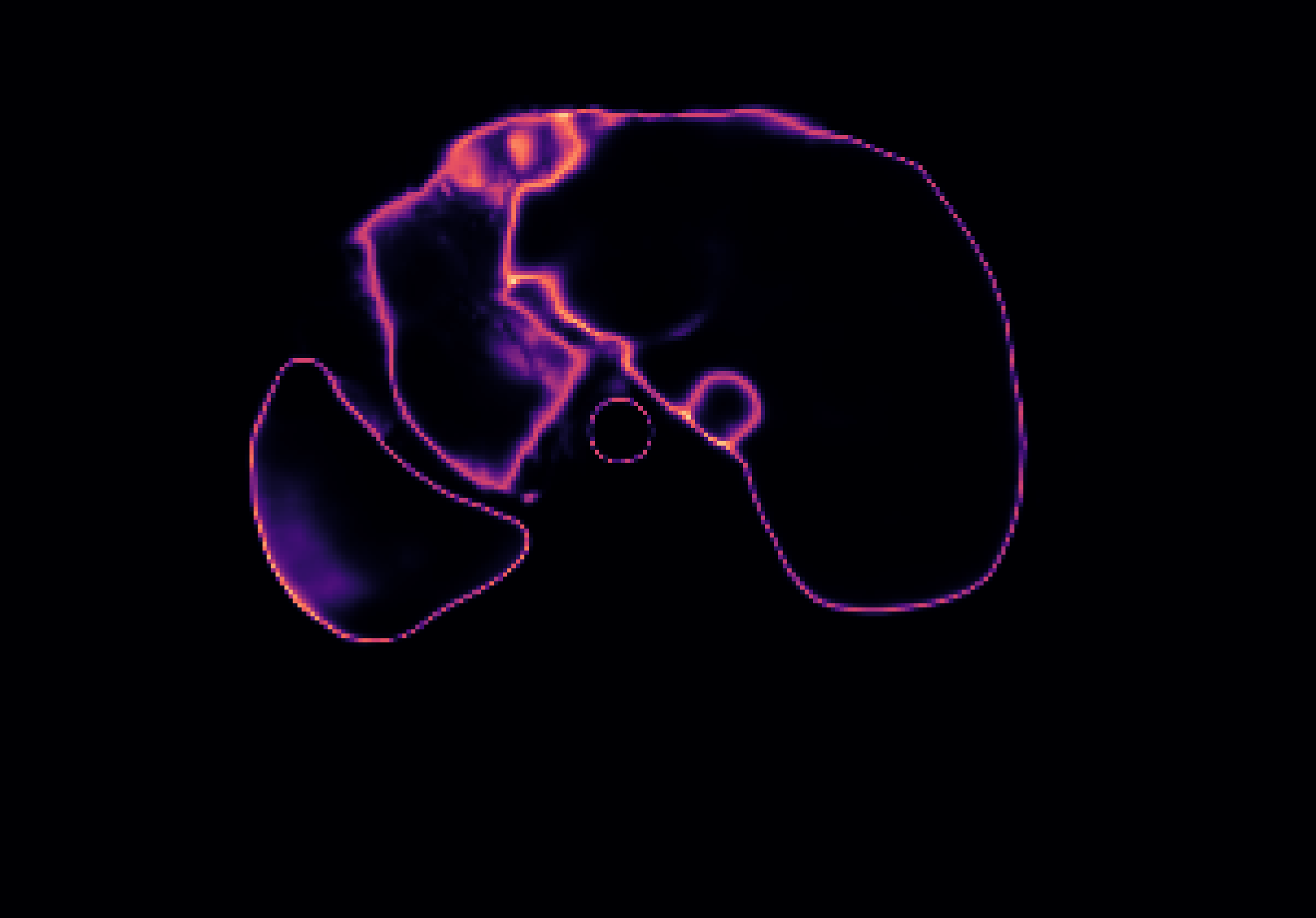}
\end{tabular}
\caption{Additional qualitative AMOS22 example using SegResNet, case 15, slice 202.}
\label{fig:supp_amos22_case15_z202}
\end{figure}

\begin{figure}[tb!]
\centering
\setlength{\tabcolsep}{1.5pt}
\begin{tabular}{ccccc}
\scriptsize Input & \scriptsize Ground truth & \scriptsize FP32 mask & \scriptsize INT4 mask & \scriptsize SGAV-MP mask \\
\includegraphics[width=0.185\textwidth]{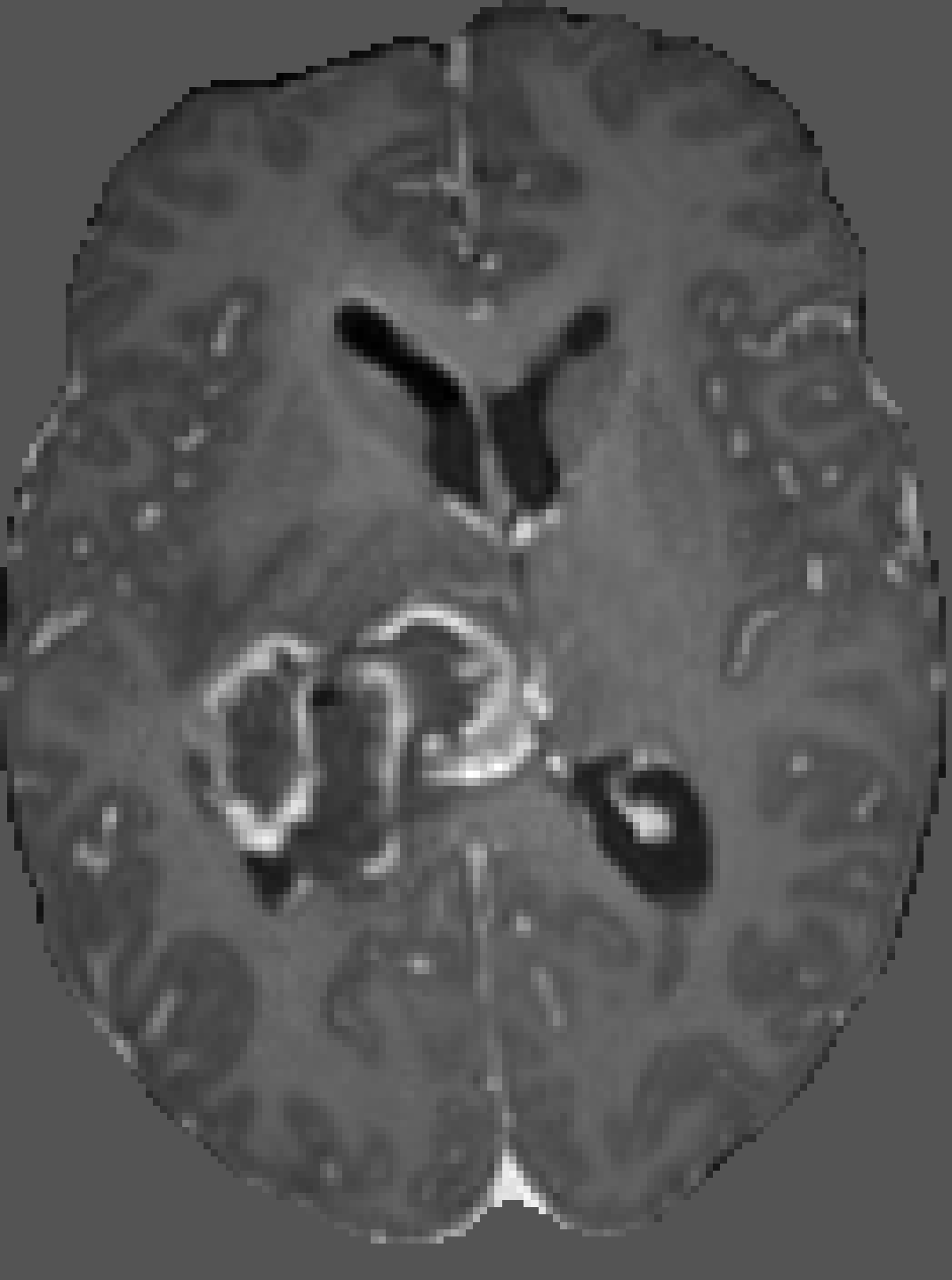} &
\includegraphics[width=0.185\textwidth]{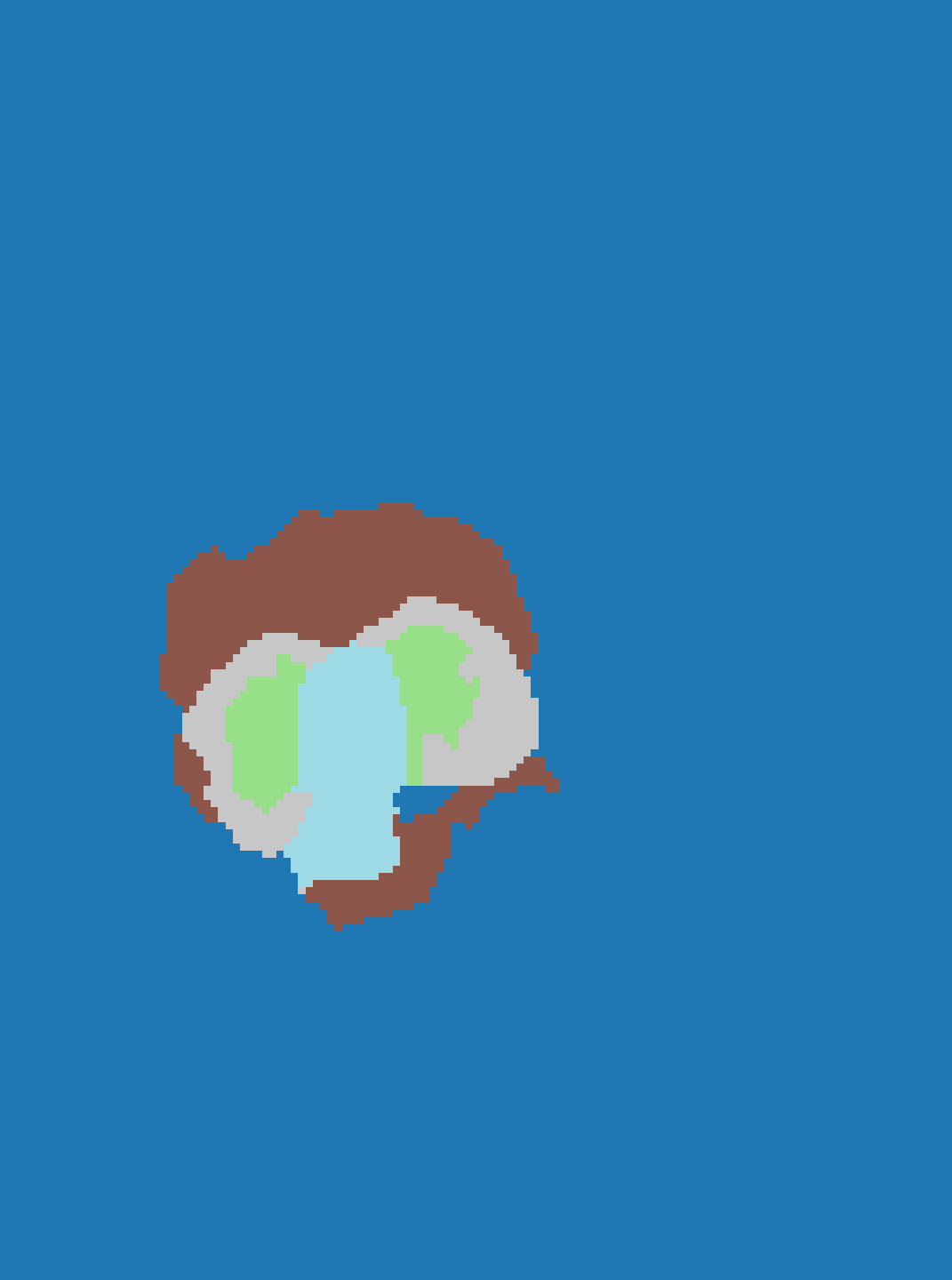} &
\includegraphics[width=0.185\textwidth]{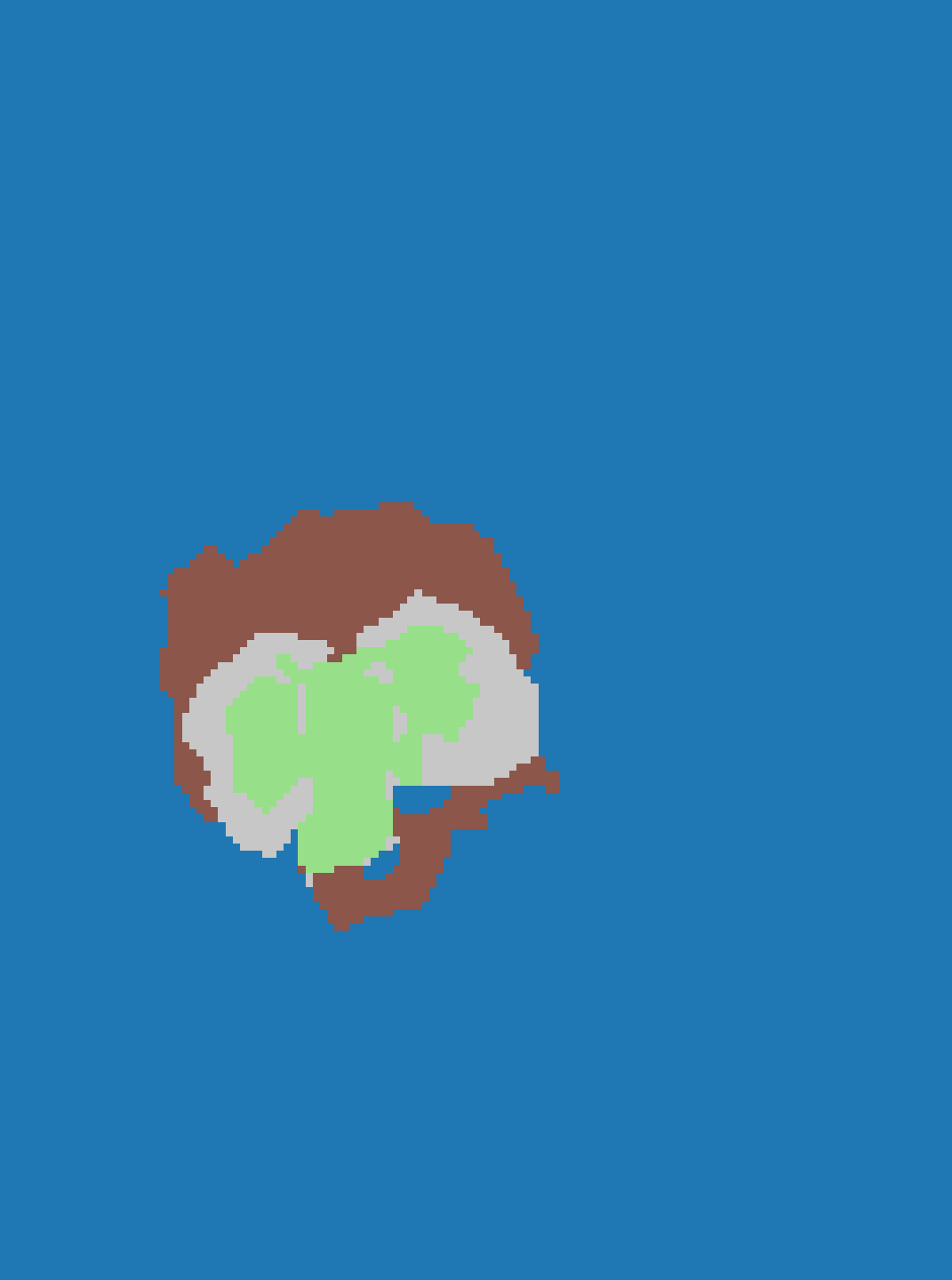} &
\includegraphics[width=0.185\textwidth]{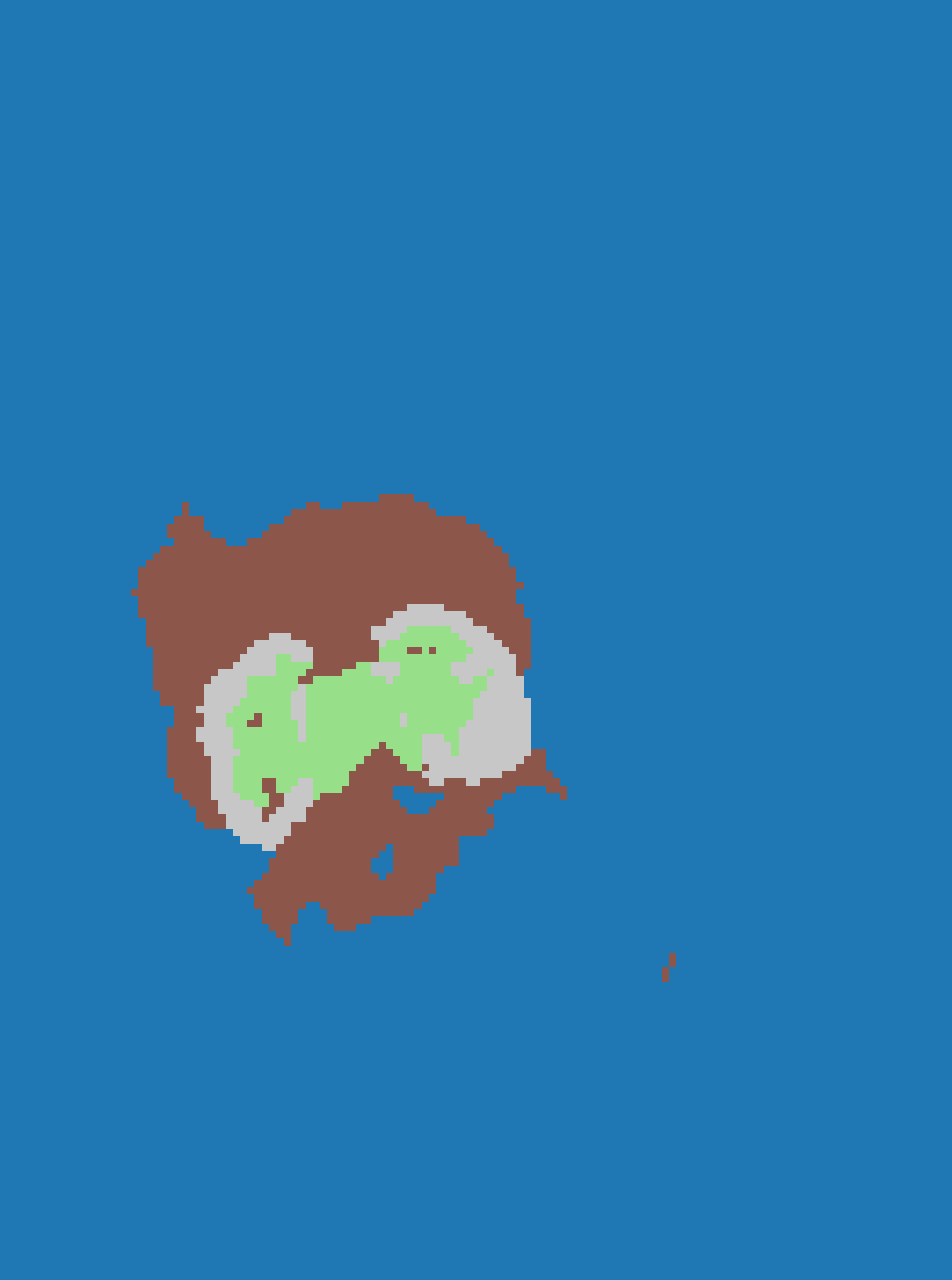} &
\includegraphics[width=0.185\textwidth]{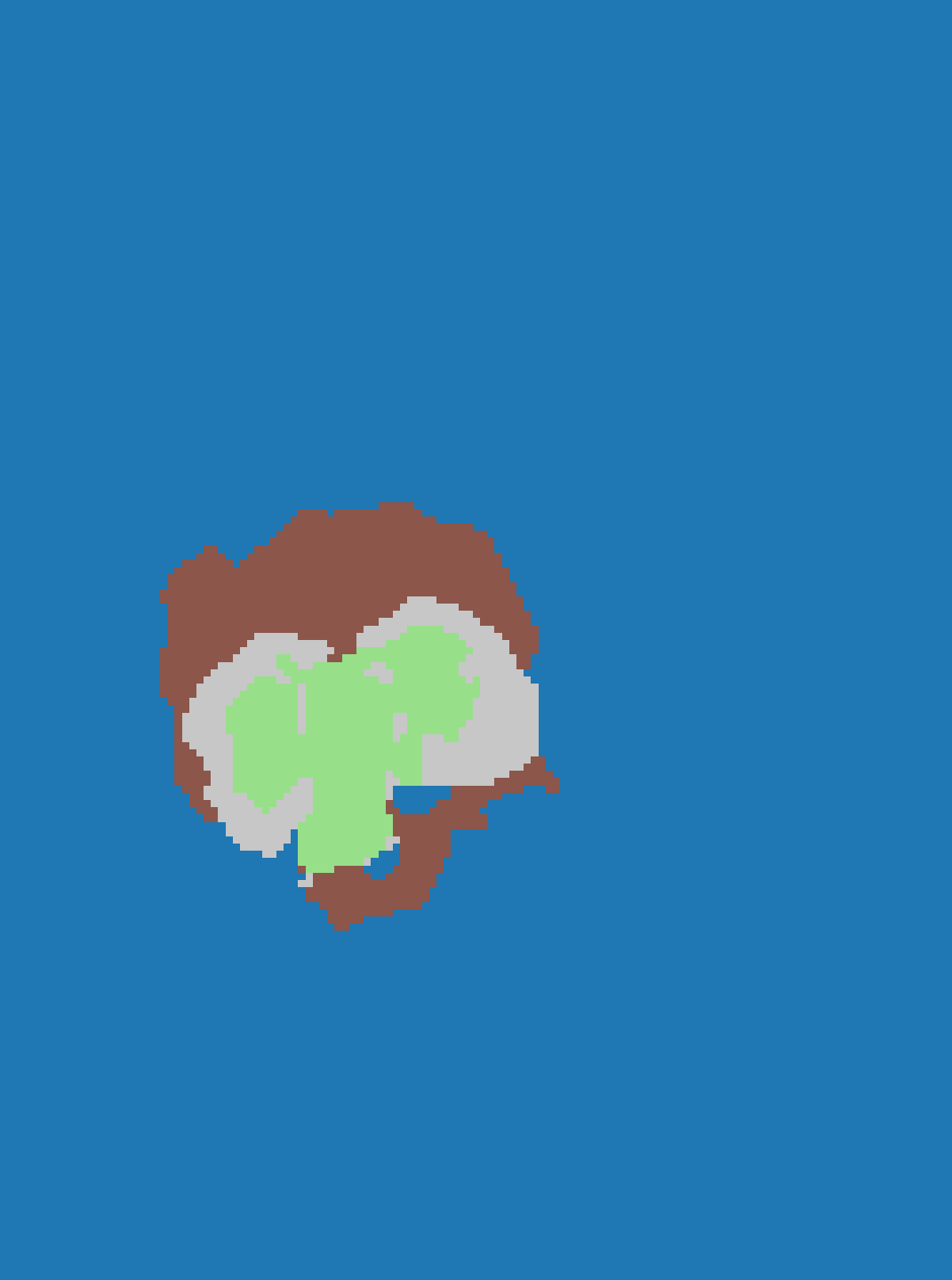} \\
\\[-0.45em]
\scriptsize SGAV-MP error & & \scriptsize FP32 uncertainty & \scriptsize INT4 uncertainty & \scriptsize SGAV-MP uncertainty \\
\includegraphics[width=0.185\textwidth]{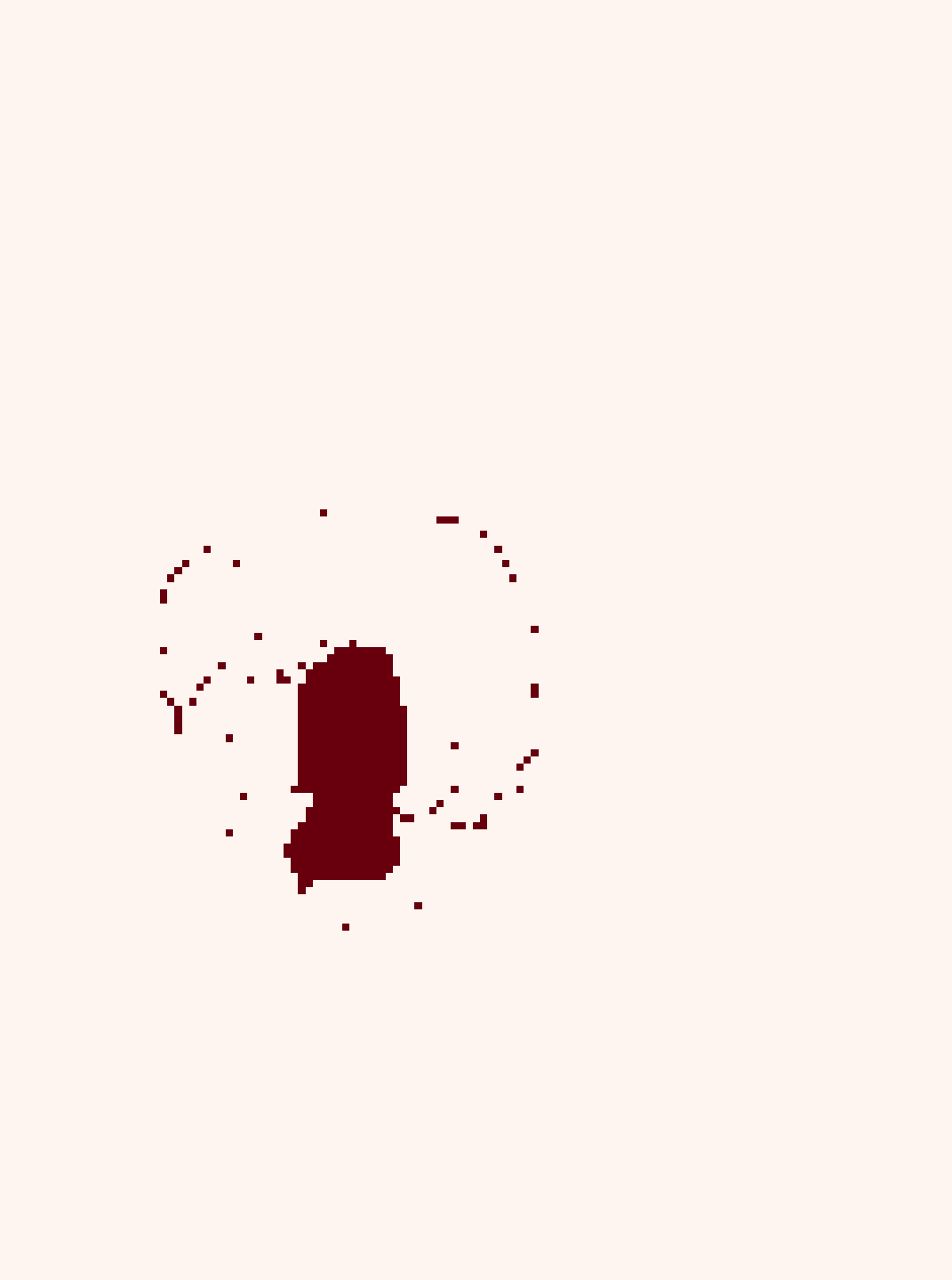} &
&
\includegraphics[width=0.185\textwidth]{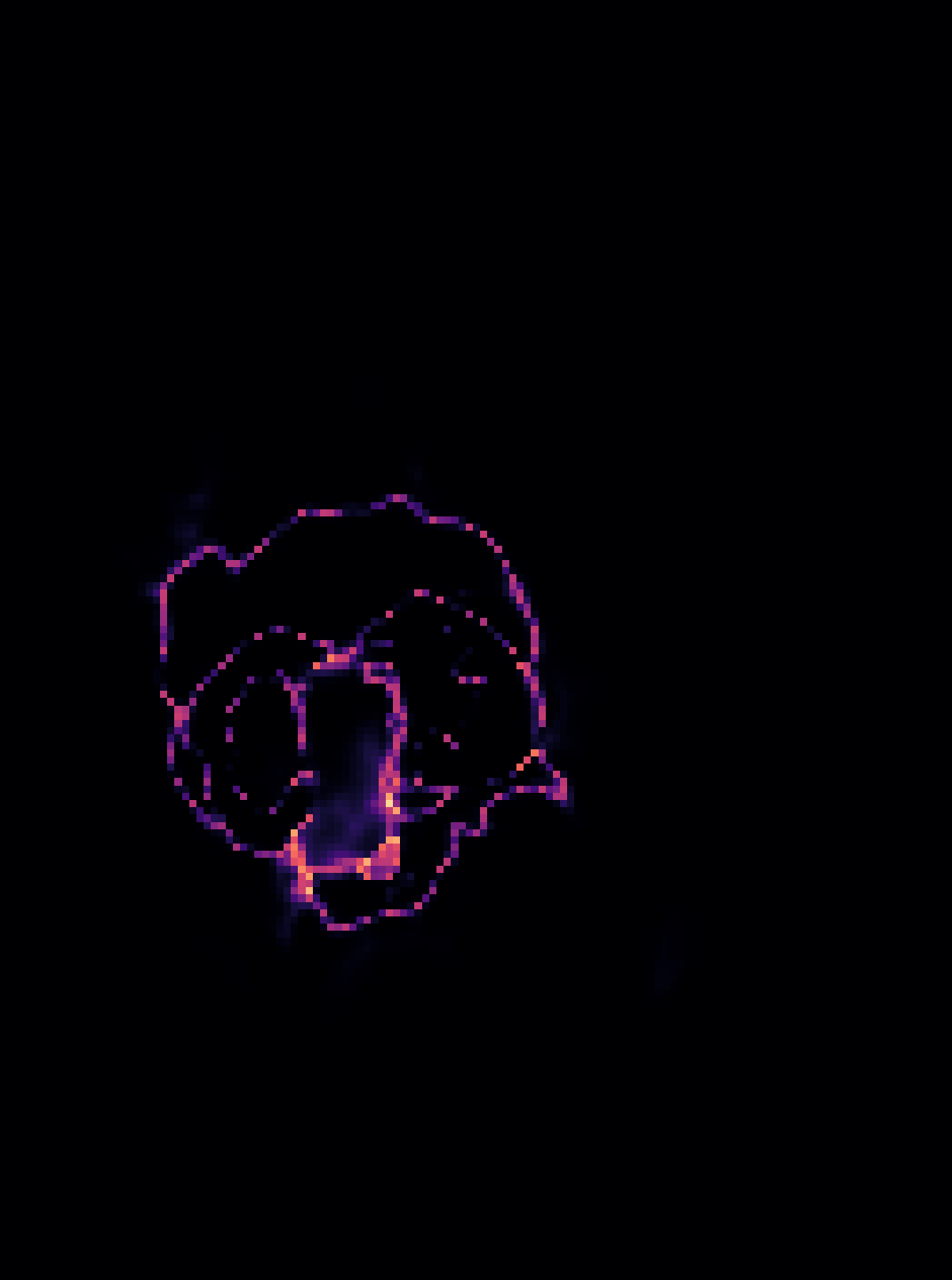} &
\includegraphics[width=0.185\textwidth]{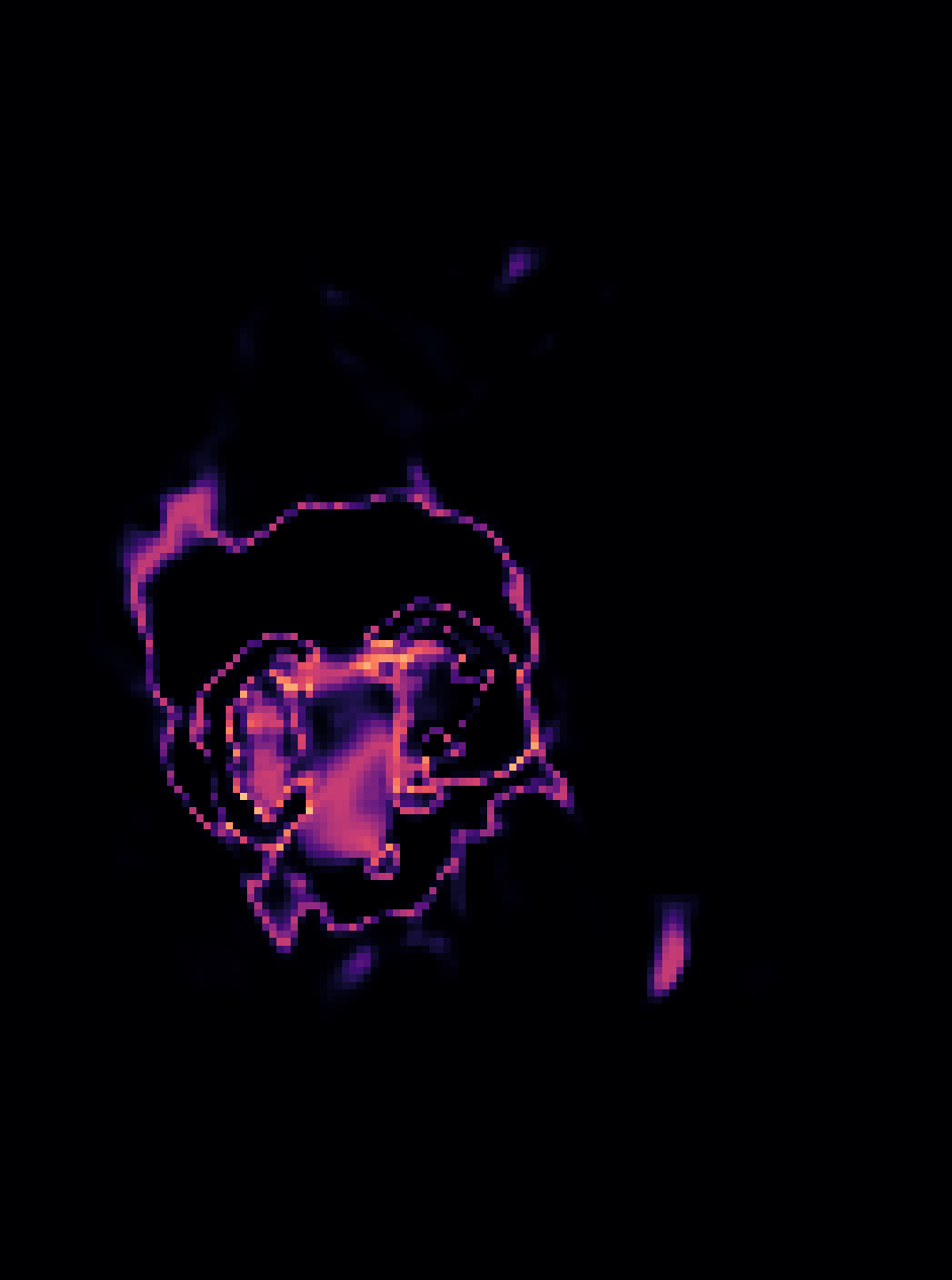} &
\includegraphics[width=0.185\textwidth]{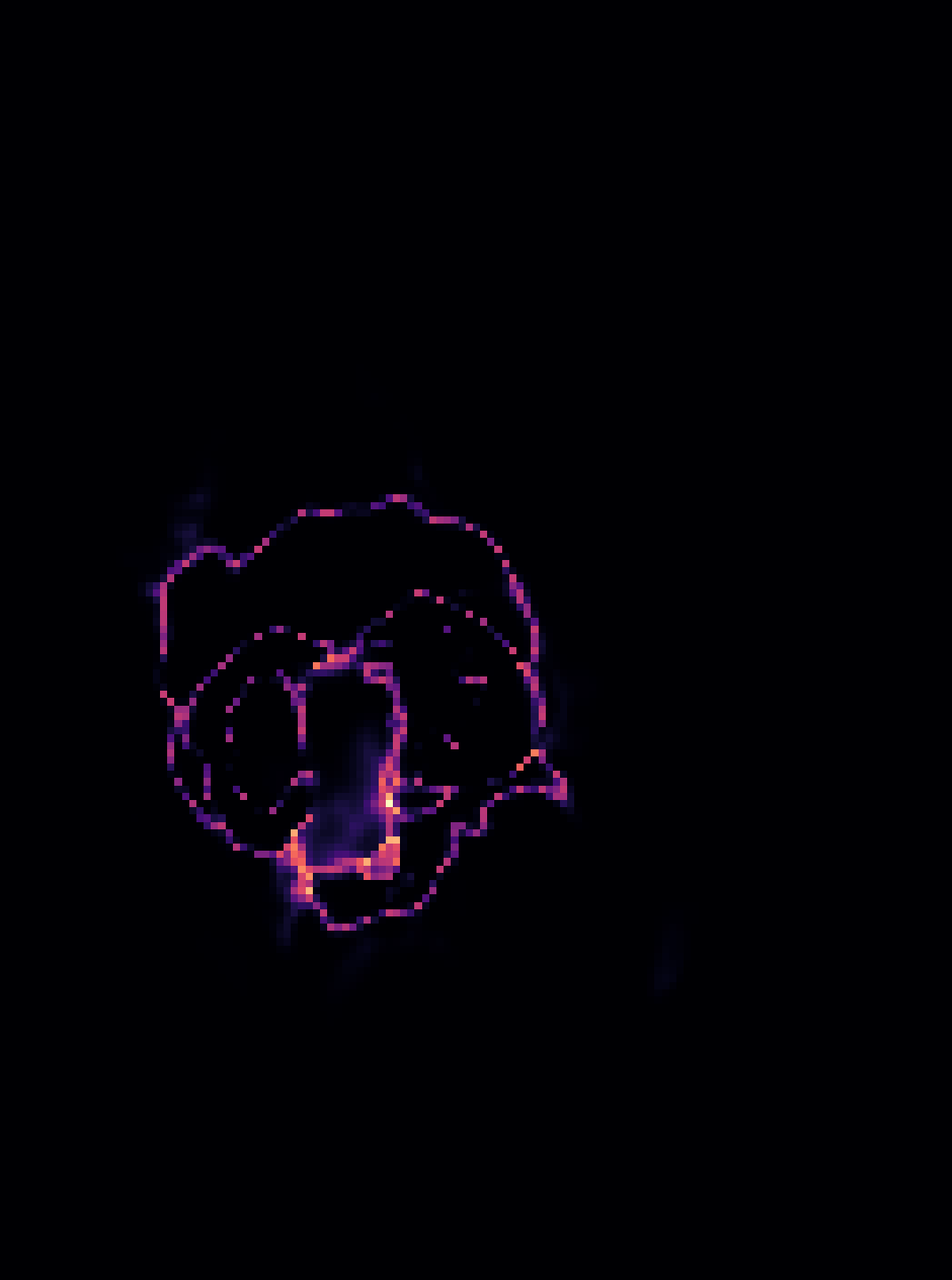}
\end{tabular}
\caption{Additional qualitative BraTS example using SegResNet, case 3, slice 76.}
\label{fig:supp_brats_case3_z76}
\end{figure}

\begin{figure}[tb!]
\centering
\setlength{\tabcolsep}{1.5pt}
\begin{tabular}{ccccc}
\scriptsize Input & \scriptsize Ground truth & \scriptsize FP32 mask & \scriptsize INT4 mask & \scriptsize SGAV-MP mask \\
\includegraphics[width=0.185\textwidth]{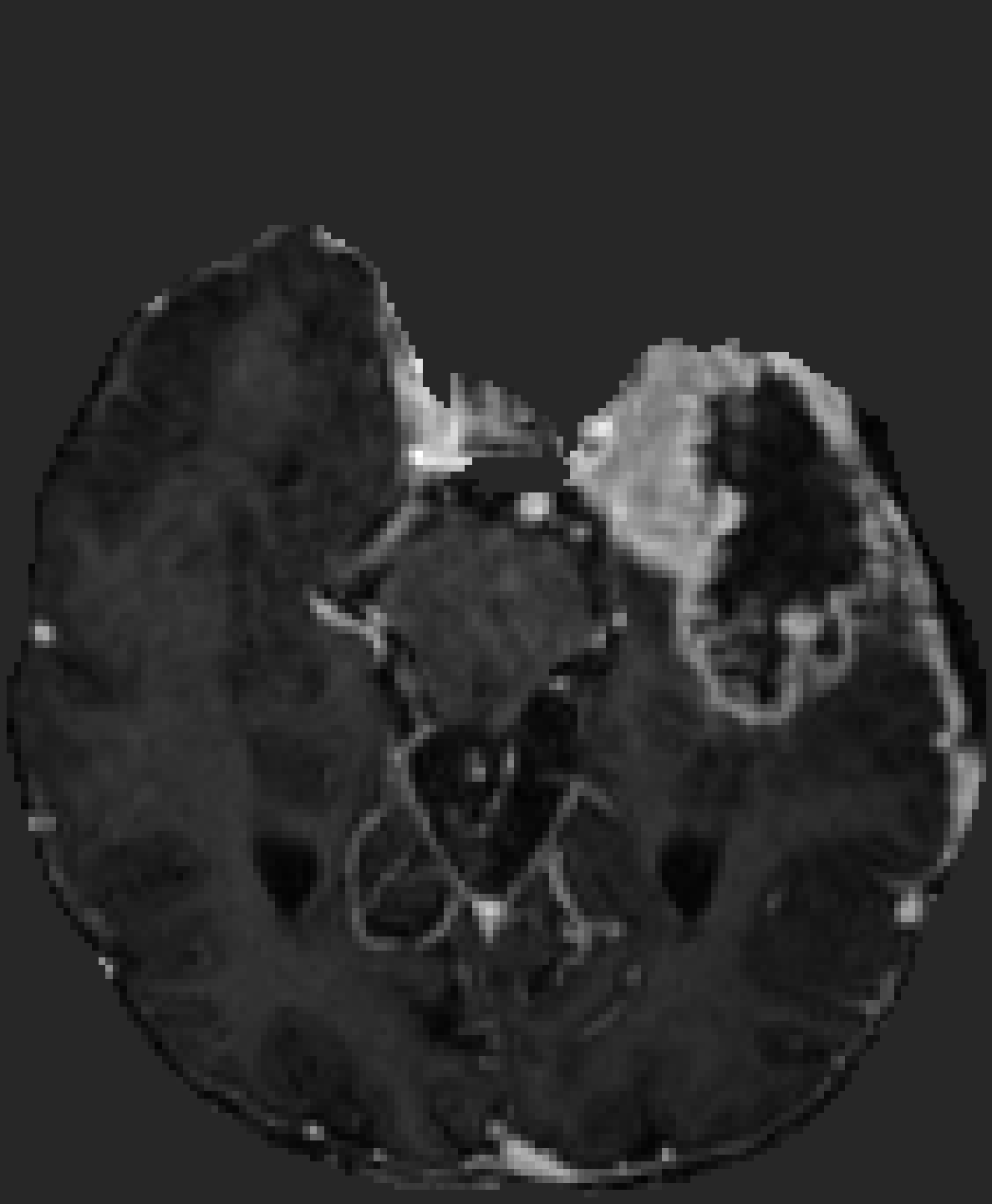} &
\includegraphics[width=0.185\textwidth]{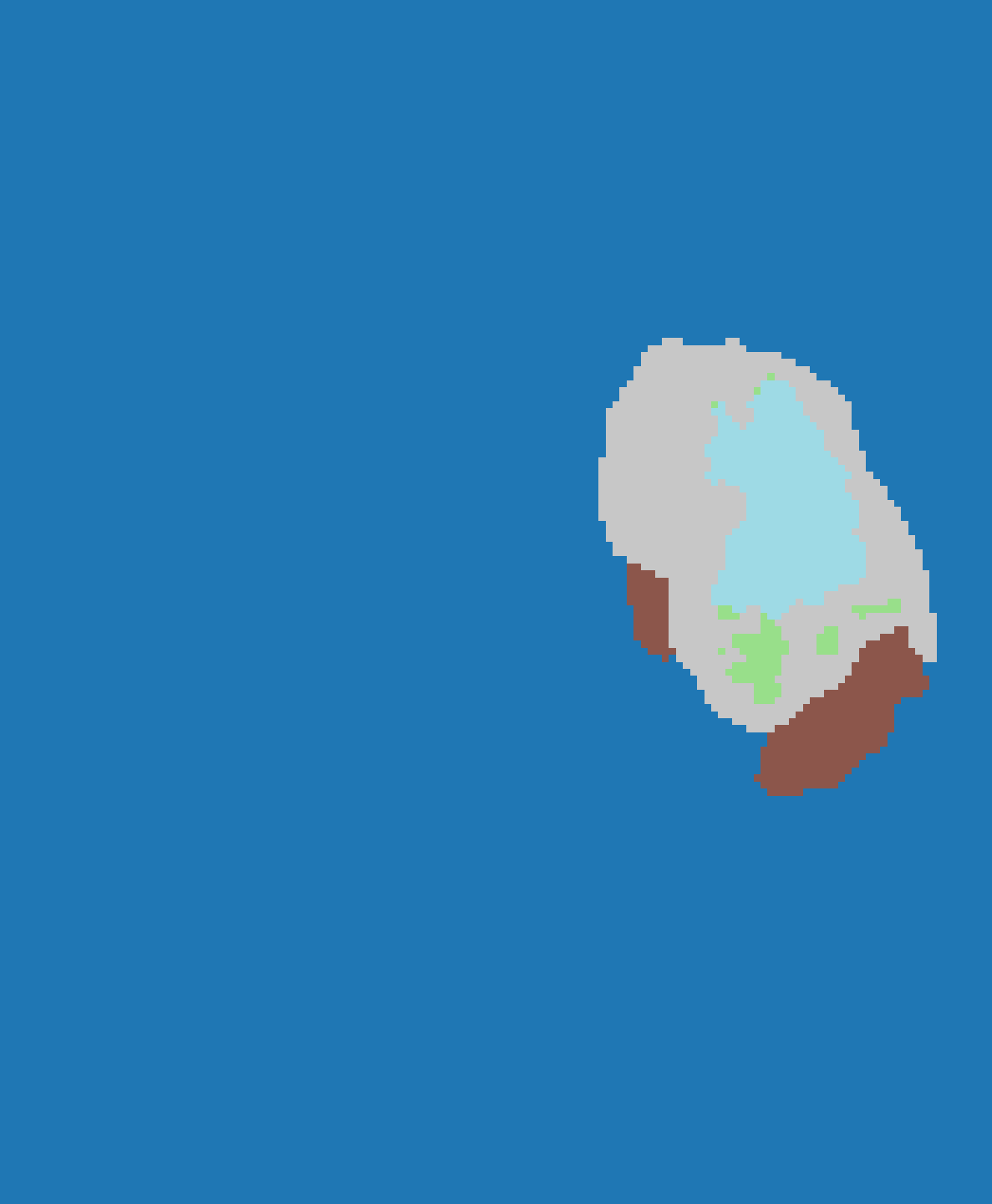} &
\includegraphics[width=0.185\textwidth]{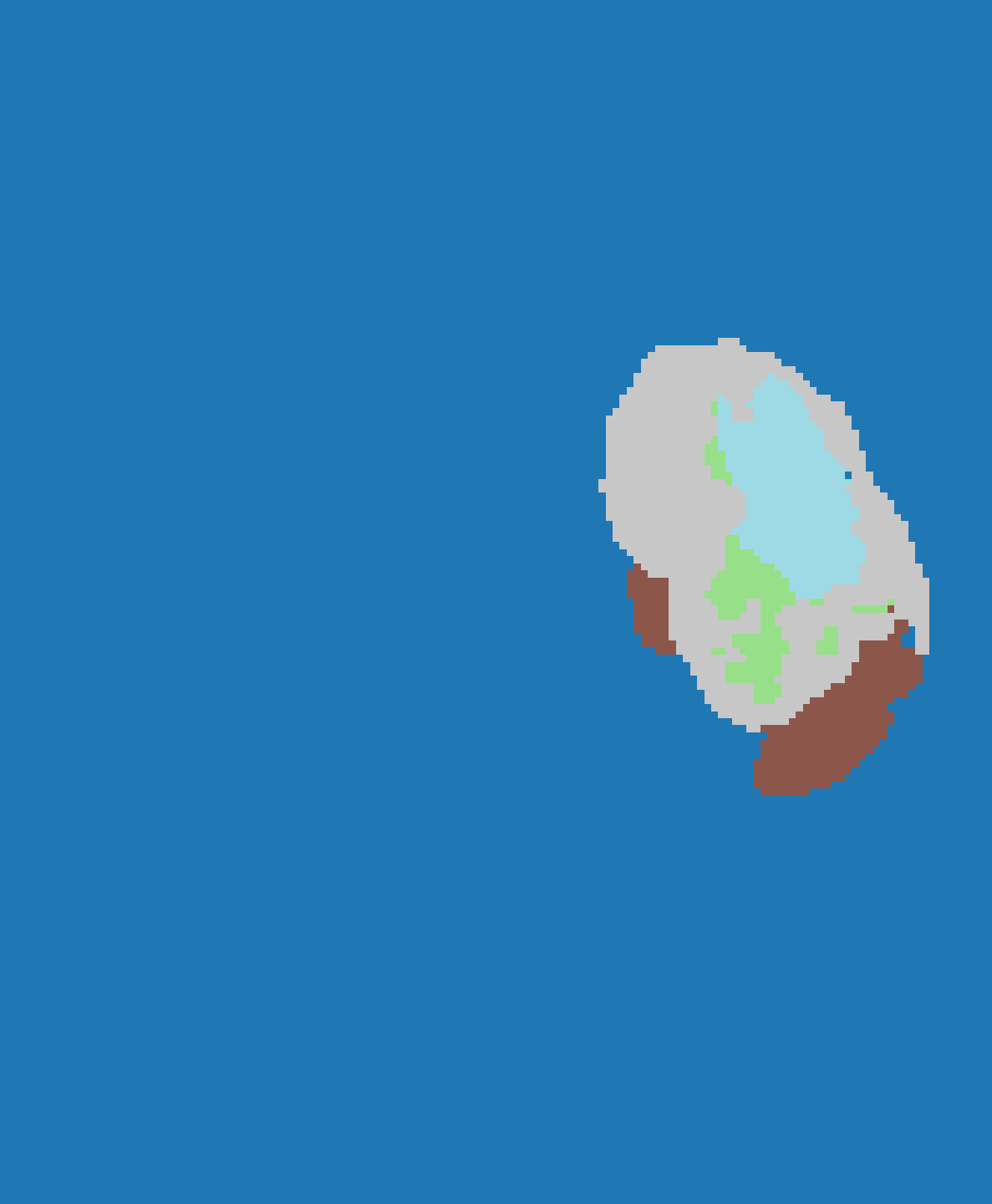} &
\includegraphics[width=0.185\textwidth]{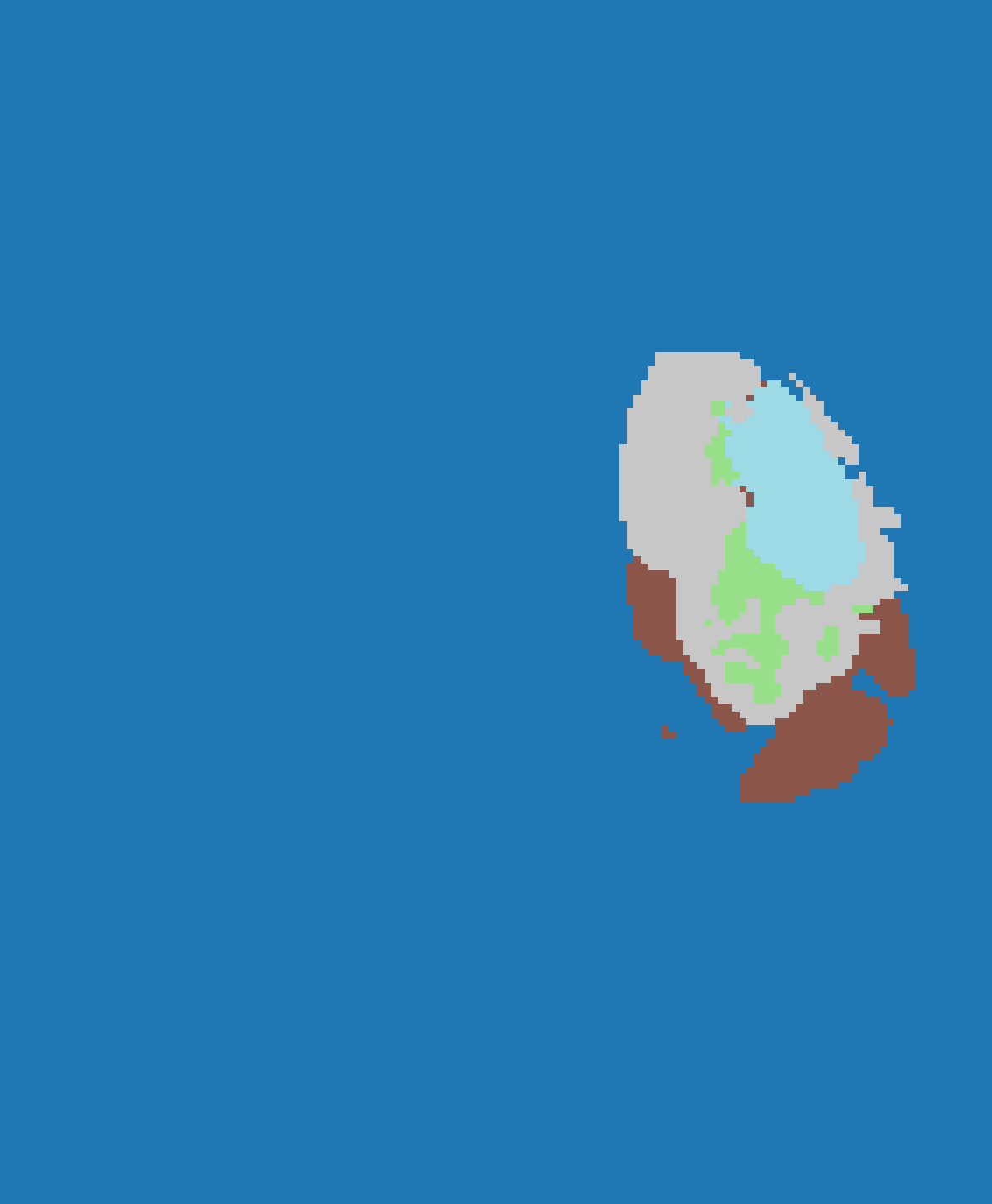} &
\includegraphics[width=0.185\textwidth]{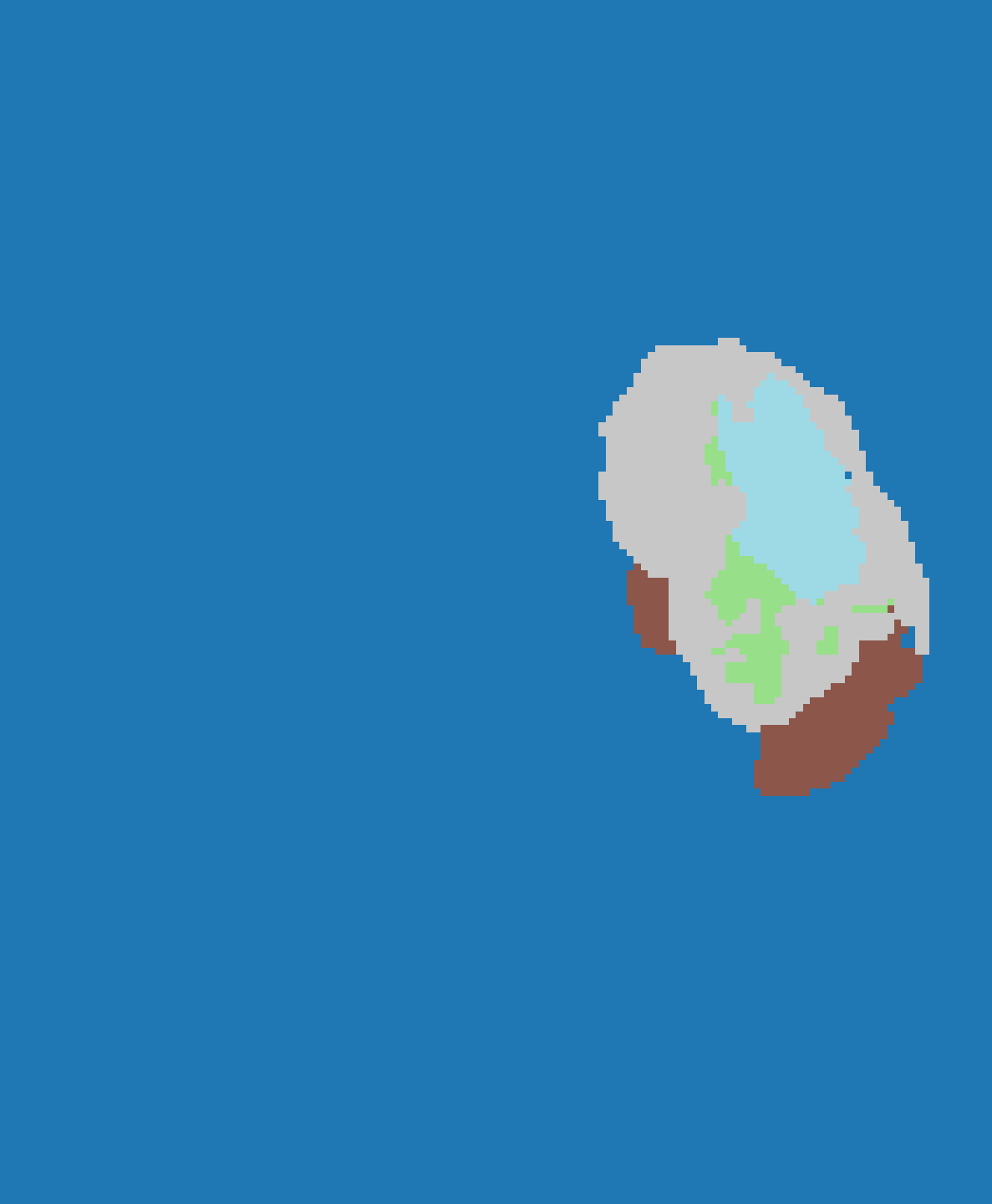} \\
\\[-0.45em]
\scriptsize SGAV-MP error & & \scriptsize FP32 uncertainty & \scriptsize INT4 uncertainty & \scriptsize SGAV-MP uncertainty \\
\includegraphics[width=0.185\textwidth]{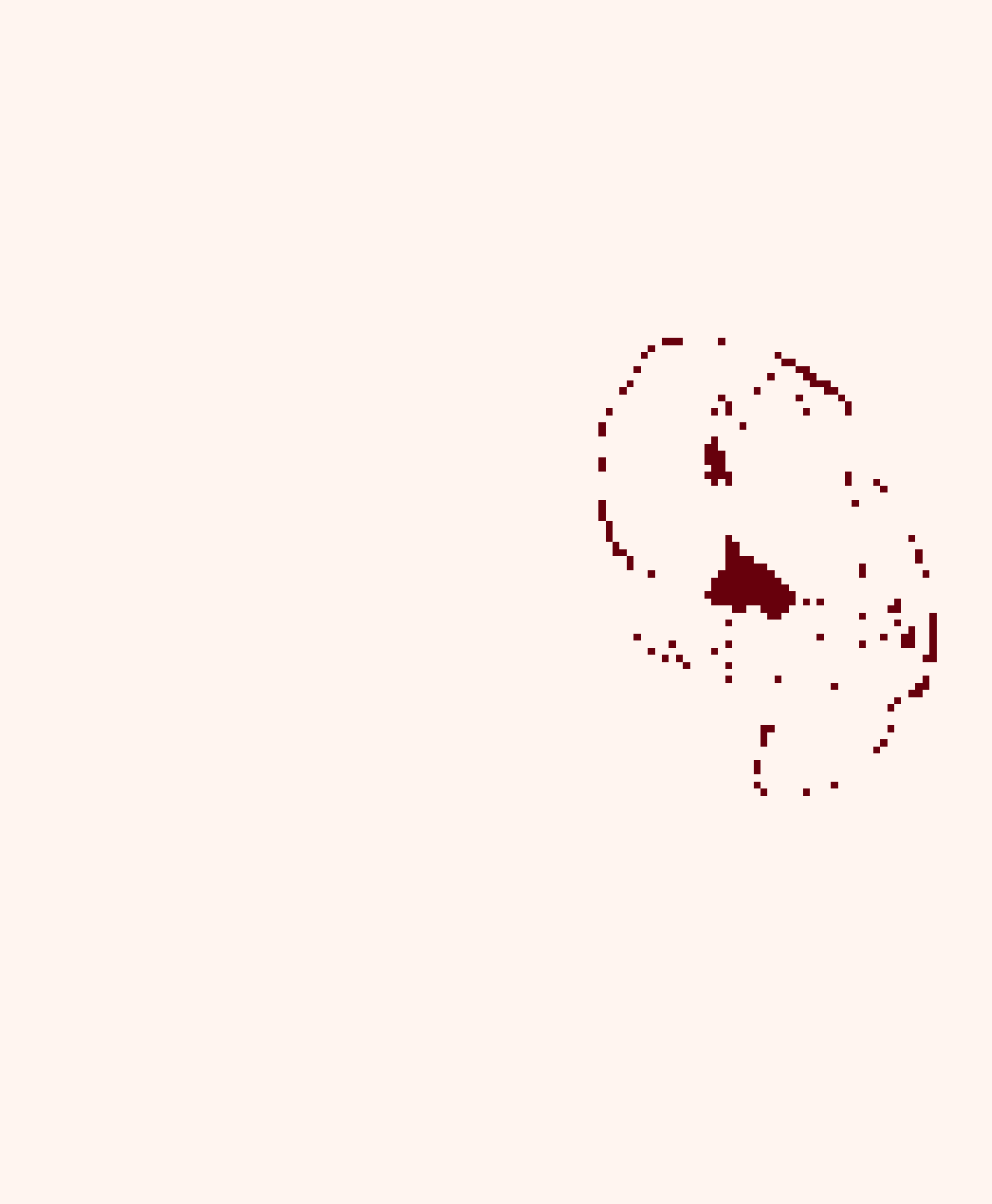} &
&
\includegraphics[width=0.185\textwidth]{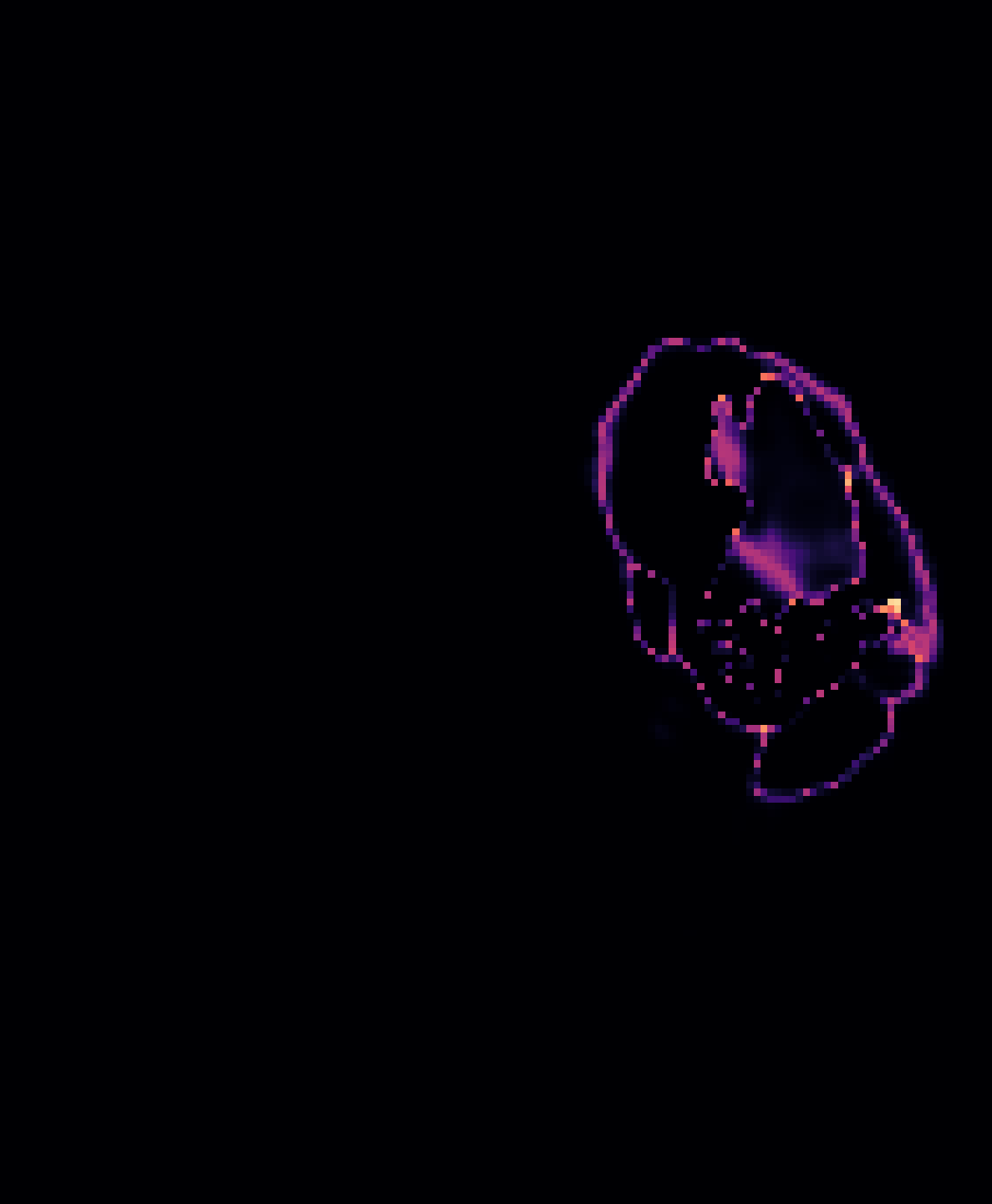} &
\includegraphics[width=0.185\textwidth]{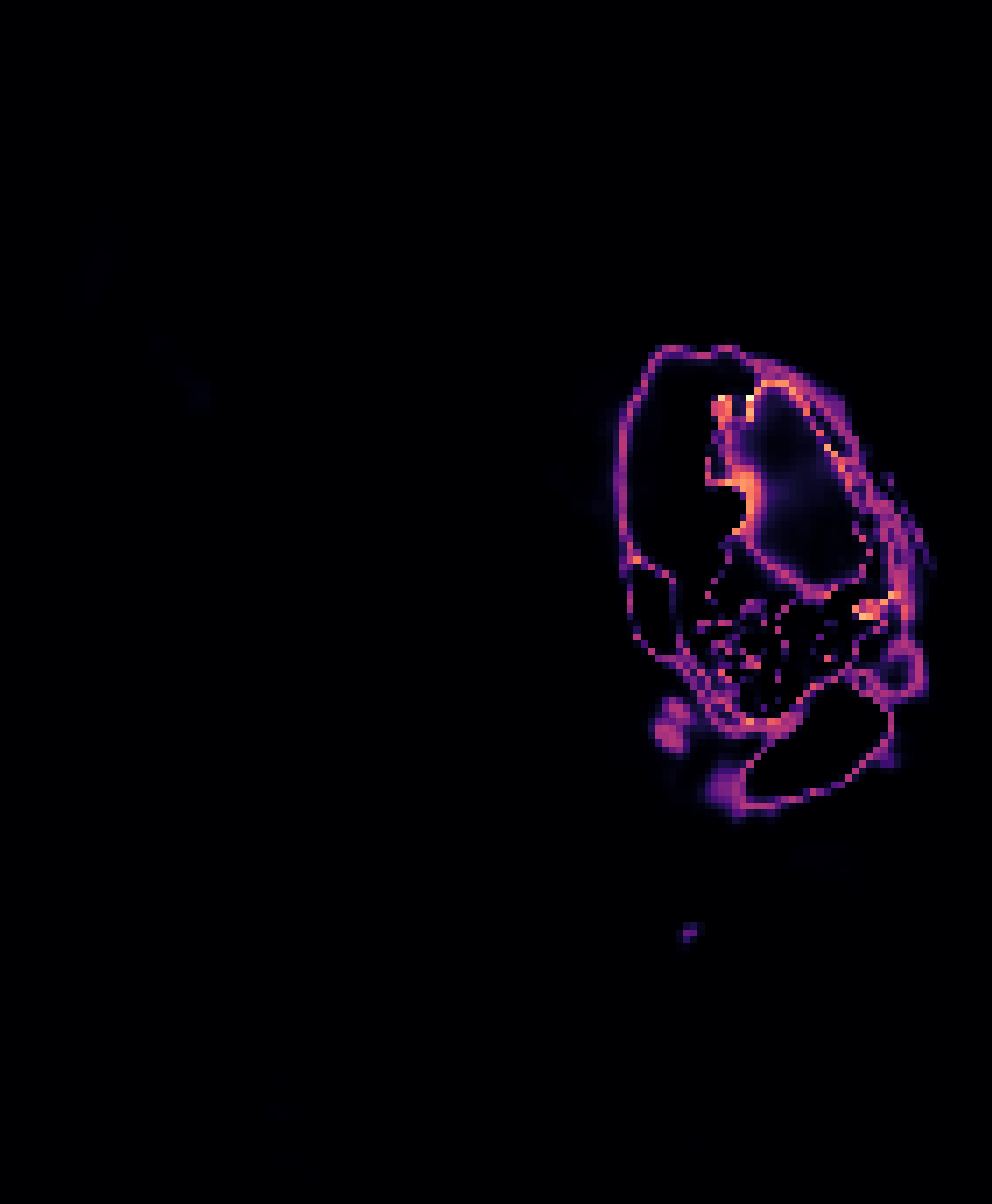} &
\includegraphics[width=0.185\textwidth]{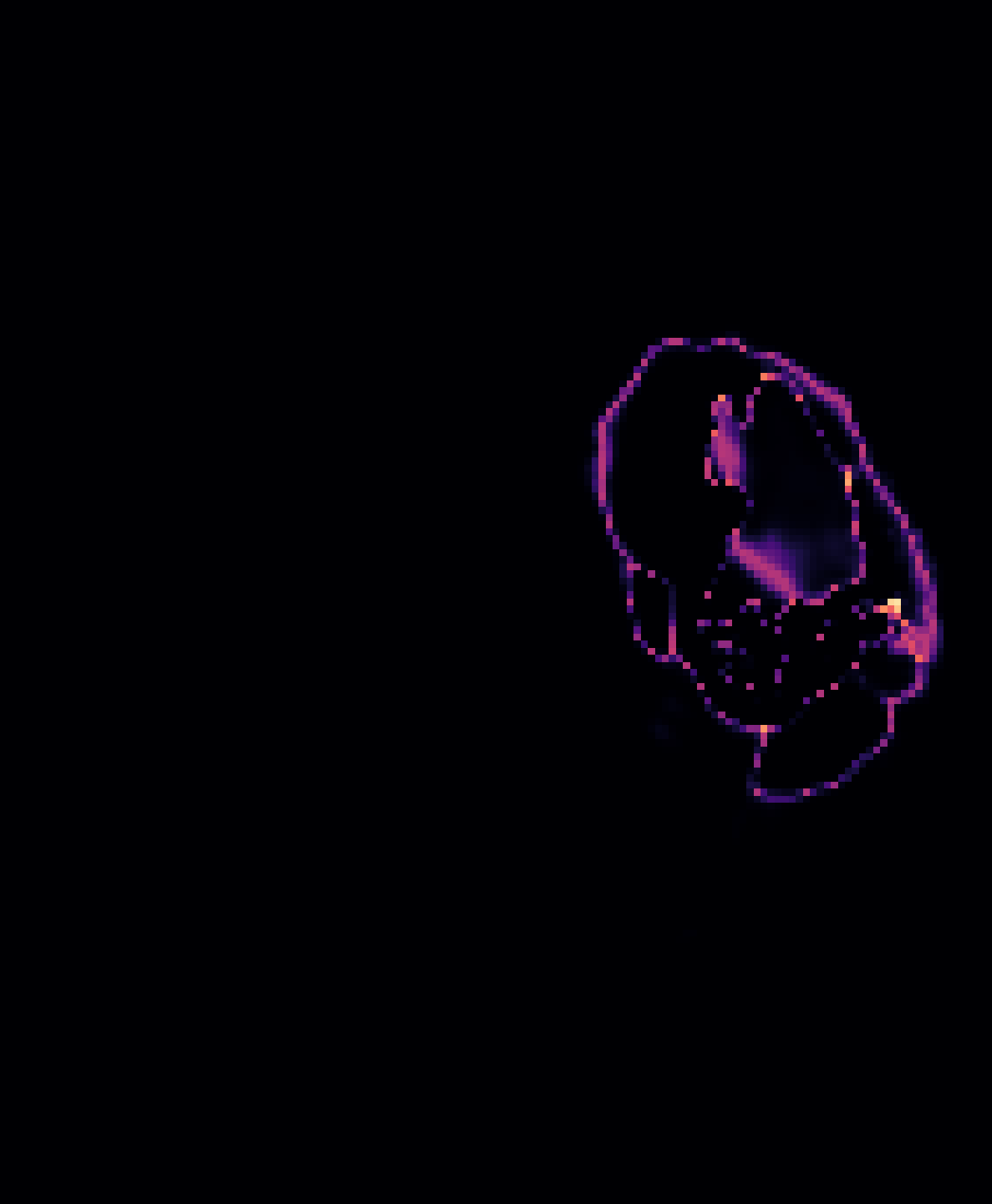}
\end{tabular}
\caption{Additional qualitative BraTS example using SegResNet, case 8, slice 59.}
\label{fig:supp_brats_case8_z59}
\end{figure}

\begin{figure}[tb!]
\centering
\setlength{\tabcolsep}{1.5pt}
\begin{tabular}{ccccc}
\scriptsize Input & \scriptsize Ground truth & \scriptsize FP32 mask & \scriptsize INT4 mask & \scriptsize SGAV-MP mask \\
\includegraphics[width=0.185\textwidth]{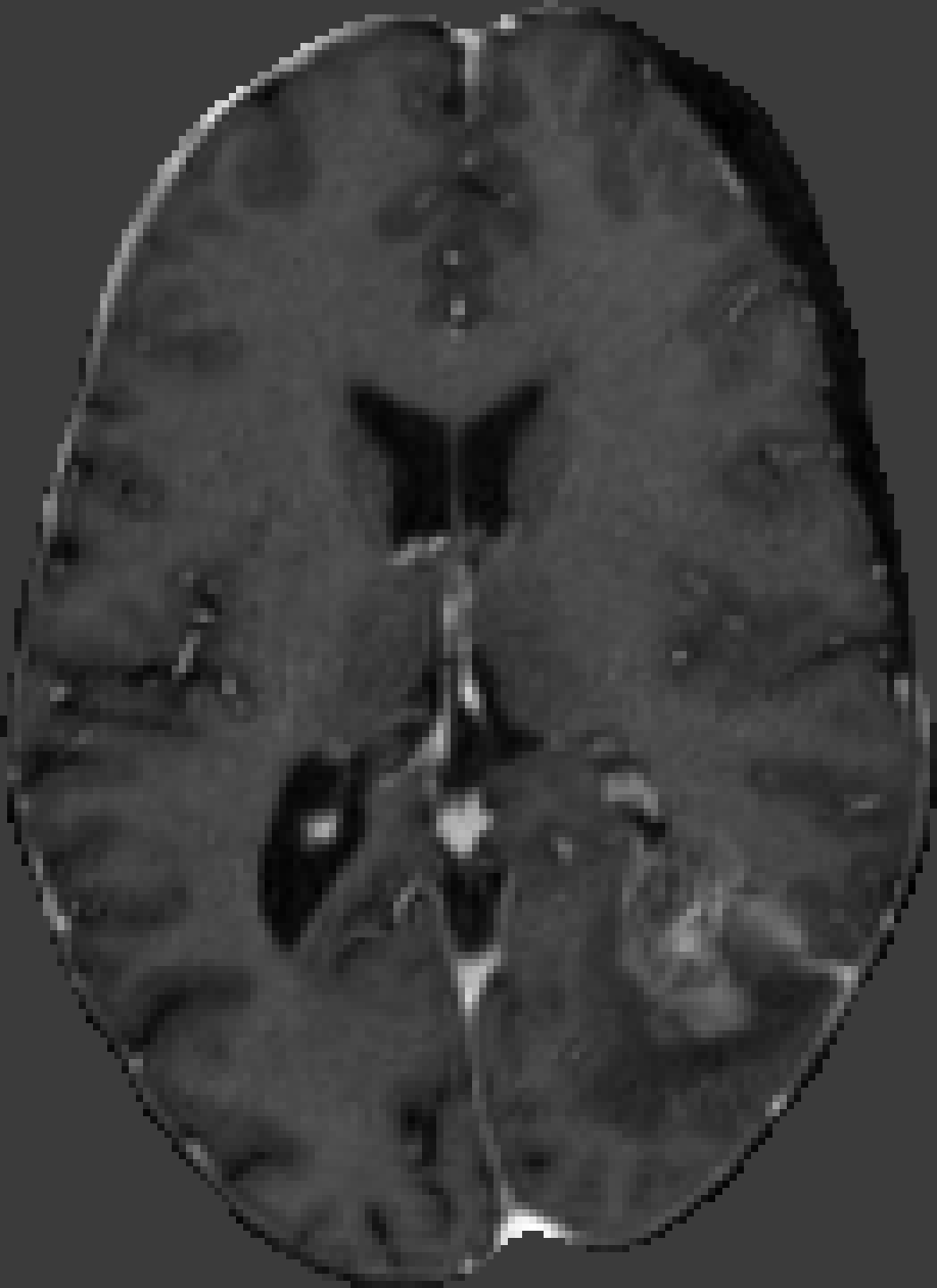} &
\includegraphics[width=0.185\textwidth]{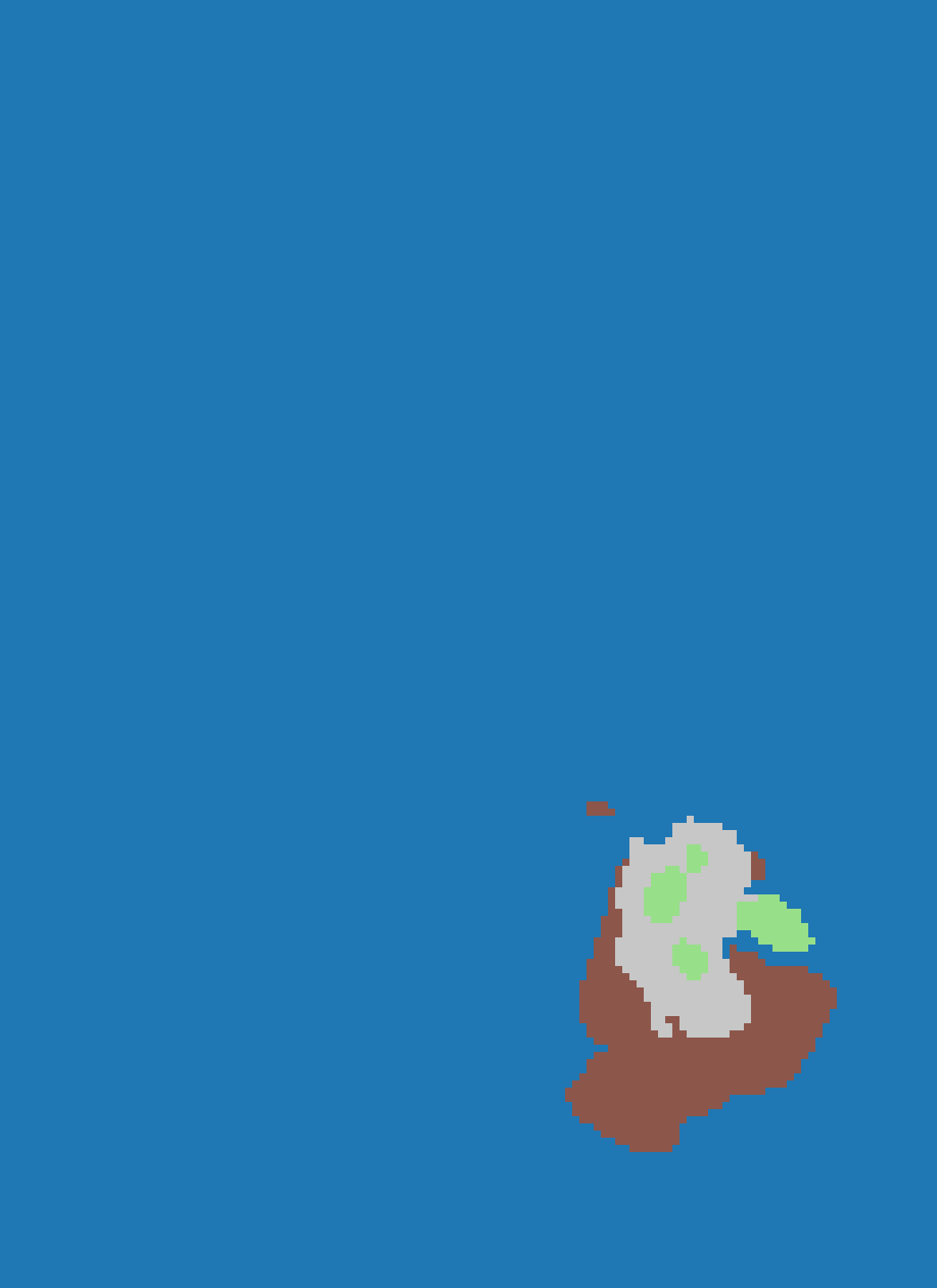} &
\includegraphics[width=0.185\textwidth]{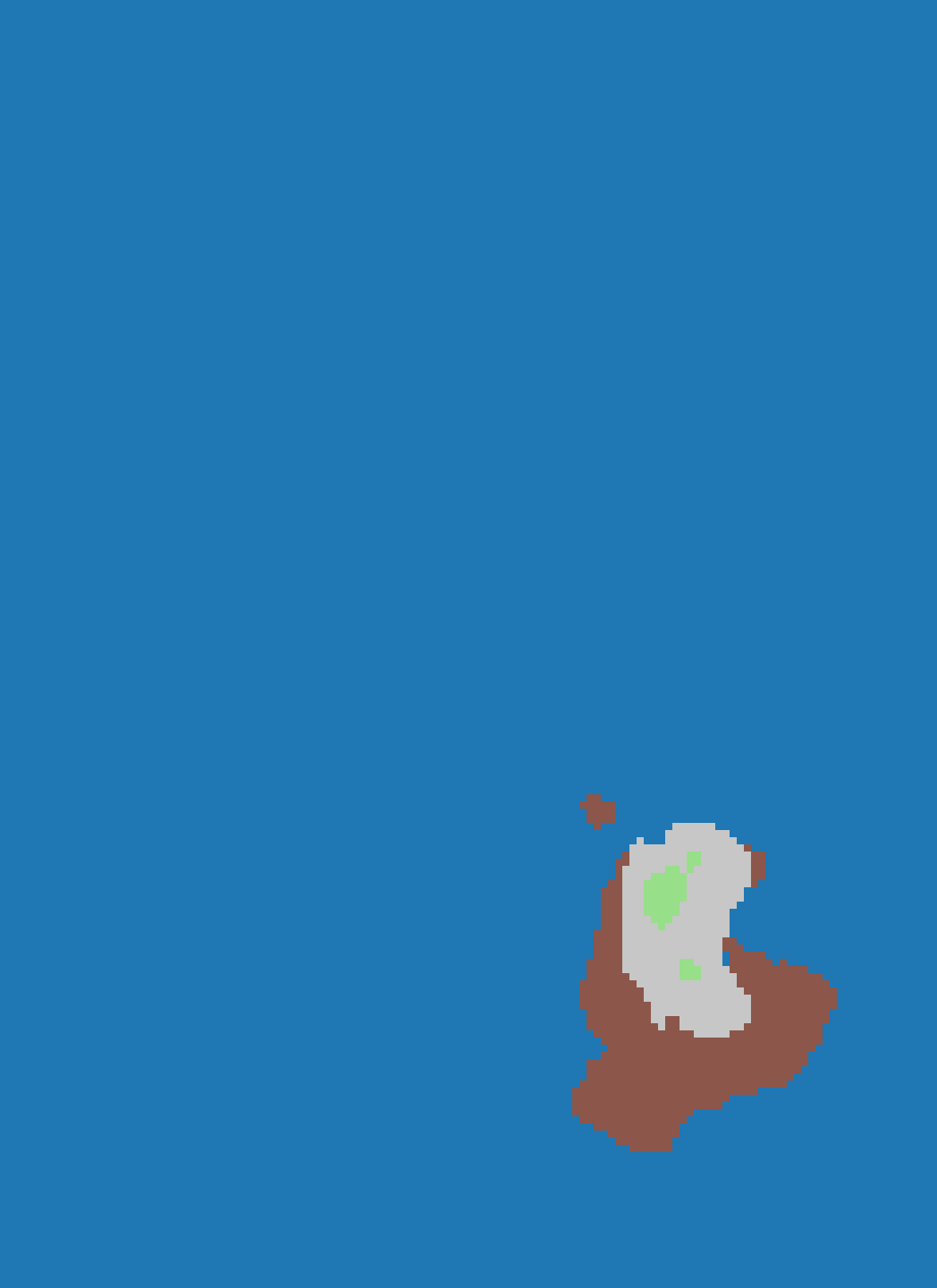} &
\includegraphics[width=0.185\textwidth]{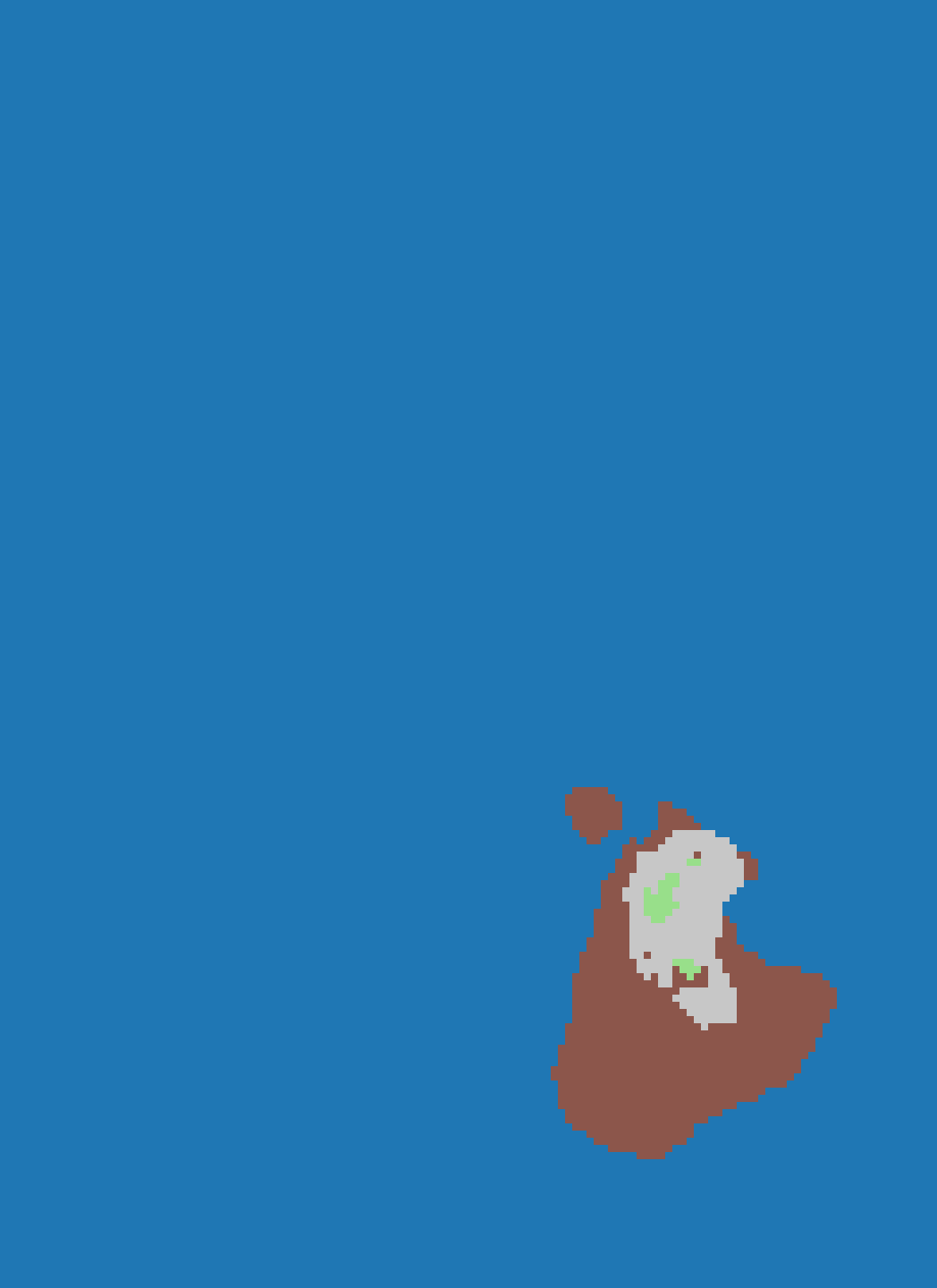} &
\includegraphics[width=0.185\textwidth]{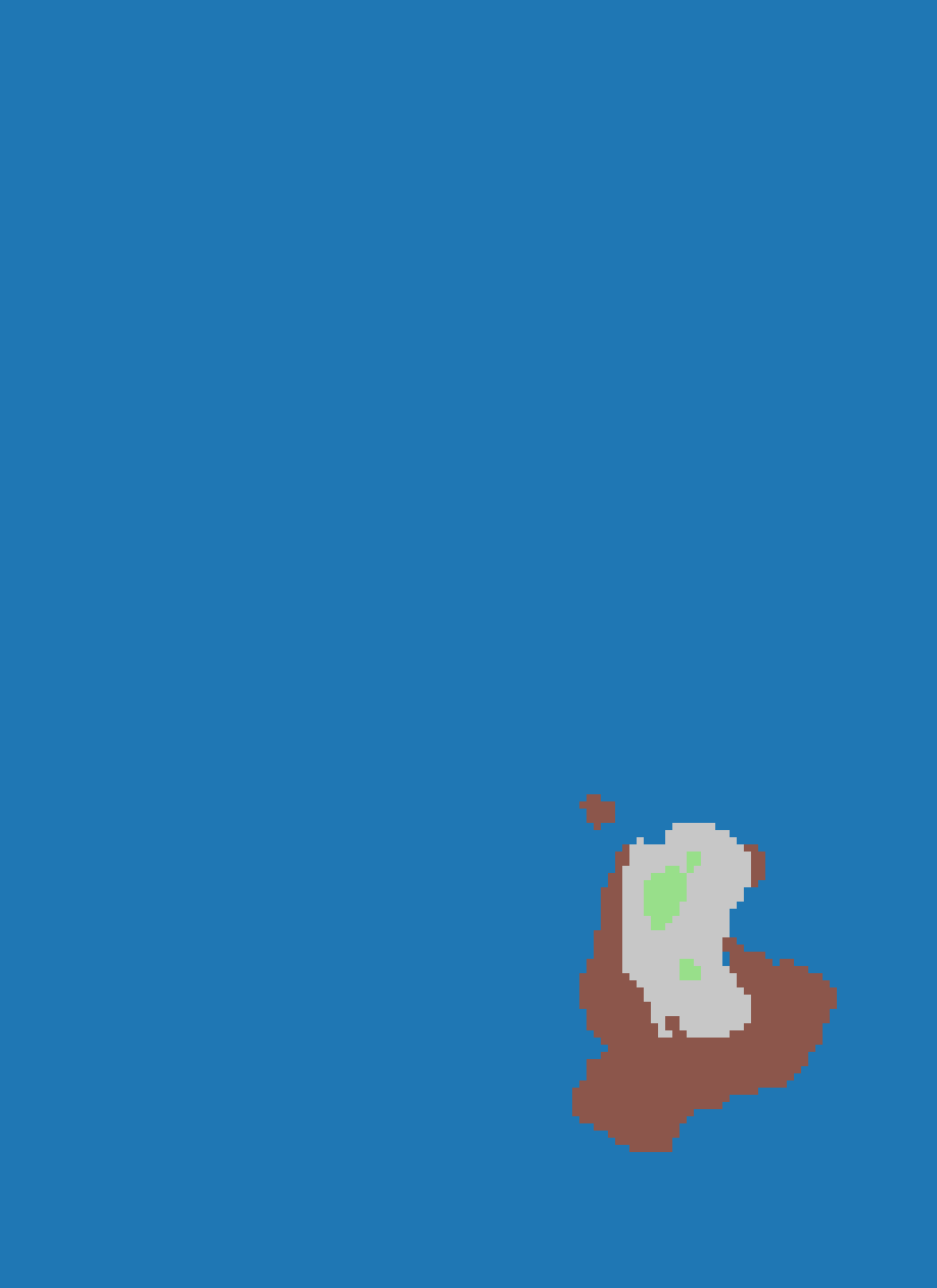} \\
\\[-0.45em]
\scriptsize SGAV-MP error & & \scriptsize FP32 uncertainty & \scriptsize INT4 uncertainty & \scriptsize SGAV-MP uncertainty \\
\includegraphics[width=0.185\textwidth]{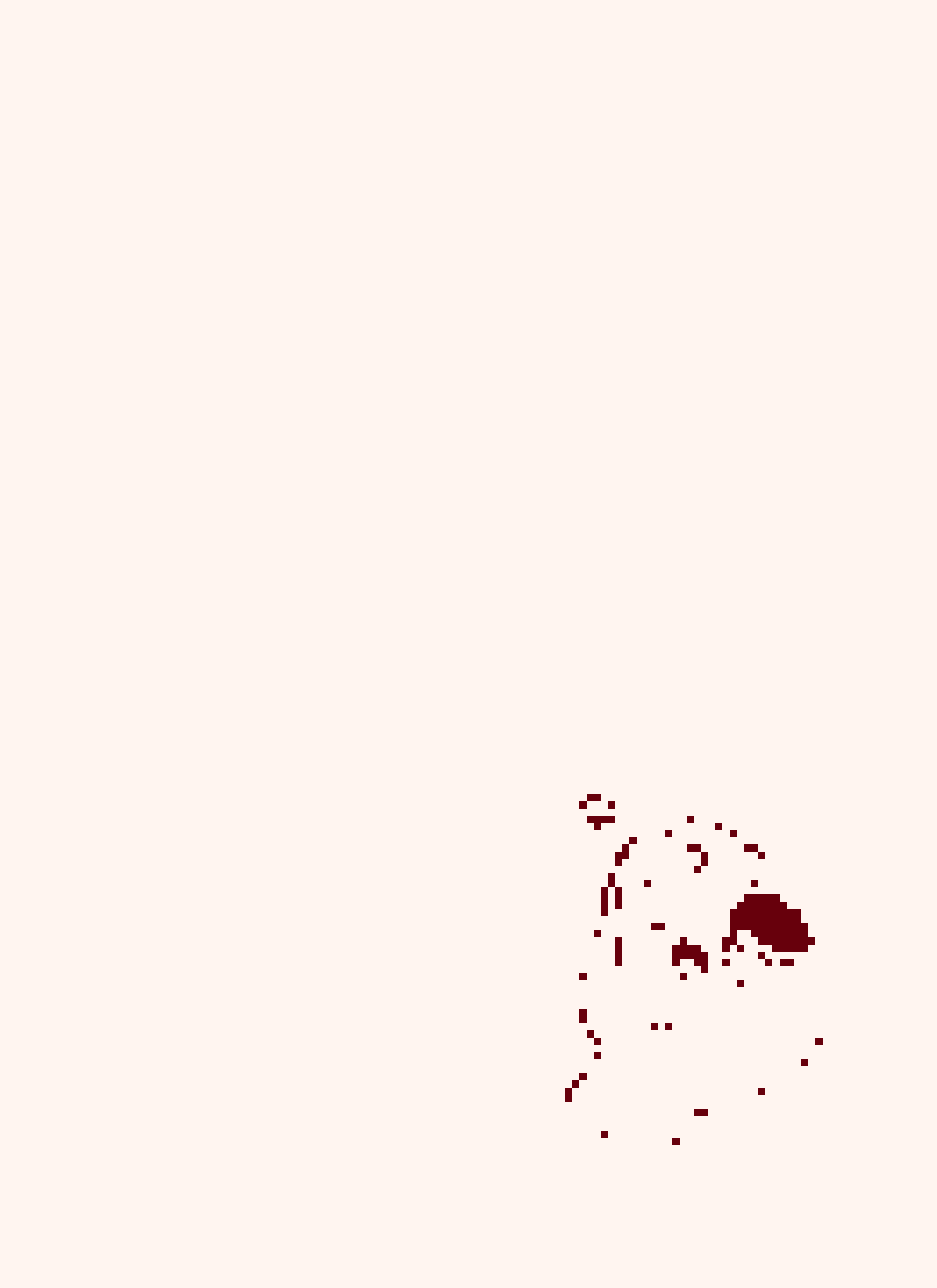} &
&
\includegraphics[width=0.185\textwidth]{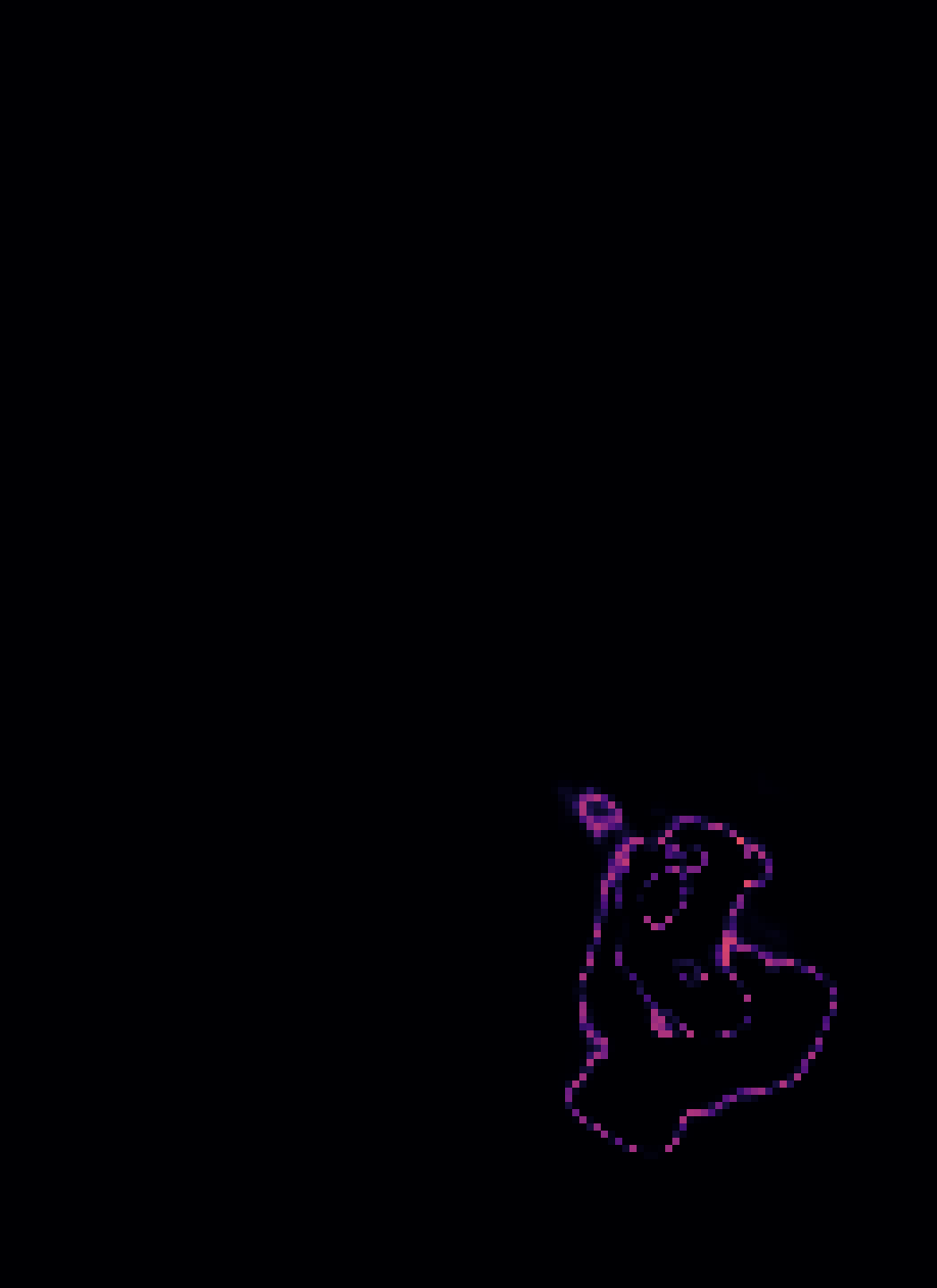} &
\includegraphics[width=0.185\textwidth]{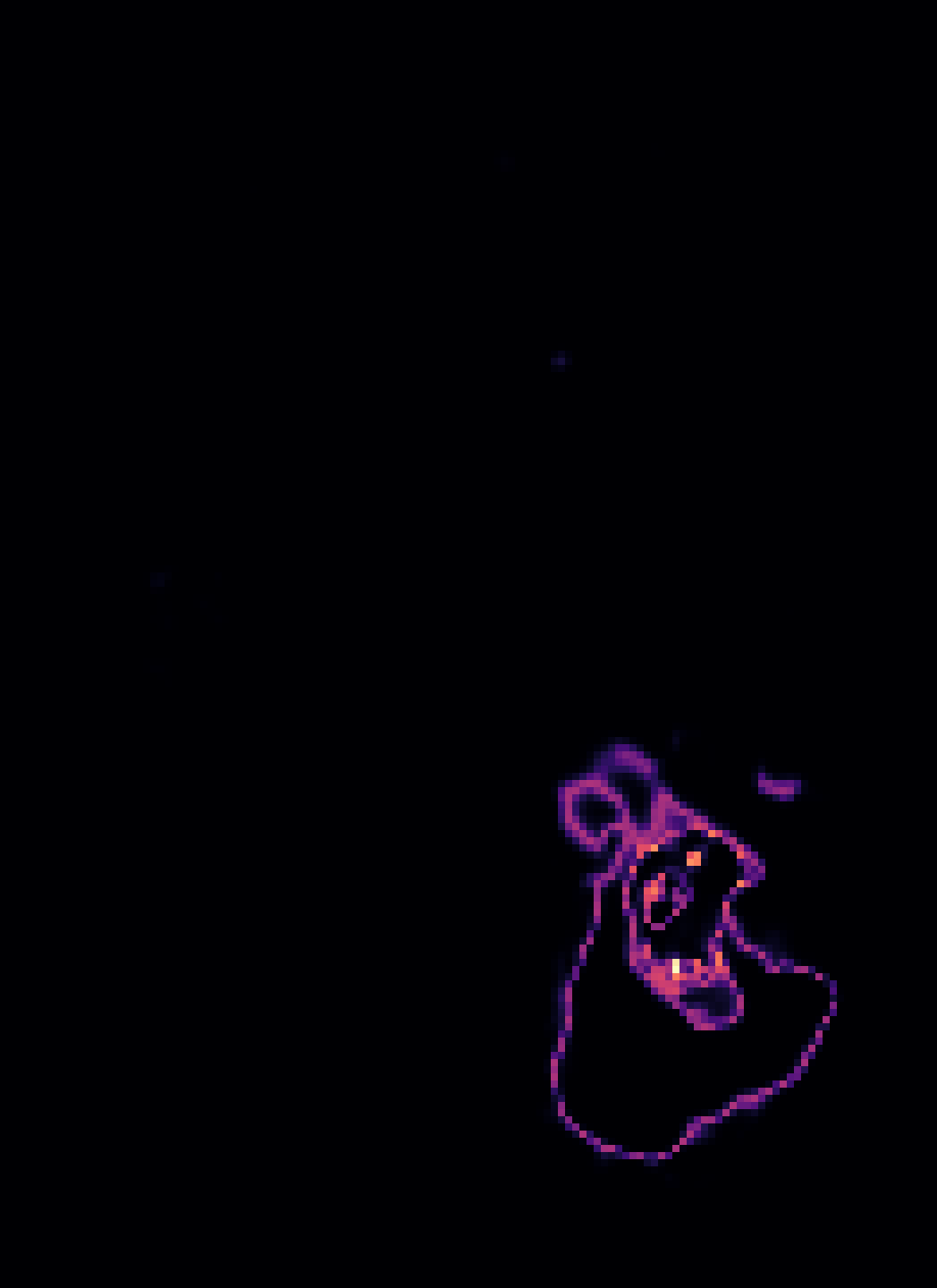} &
\includegraphics[width=0.185\textwidth]{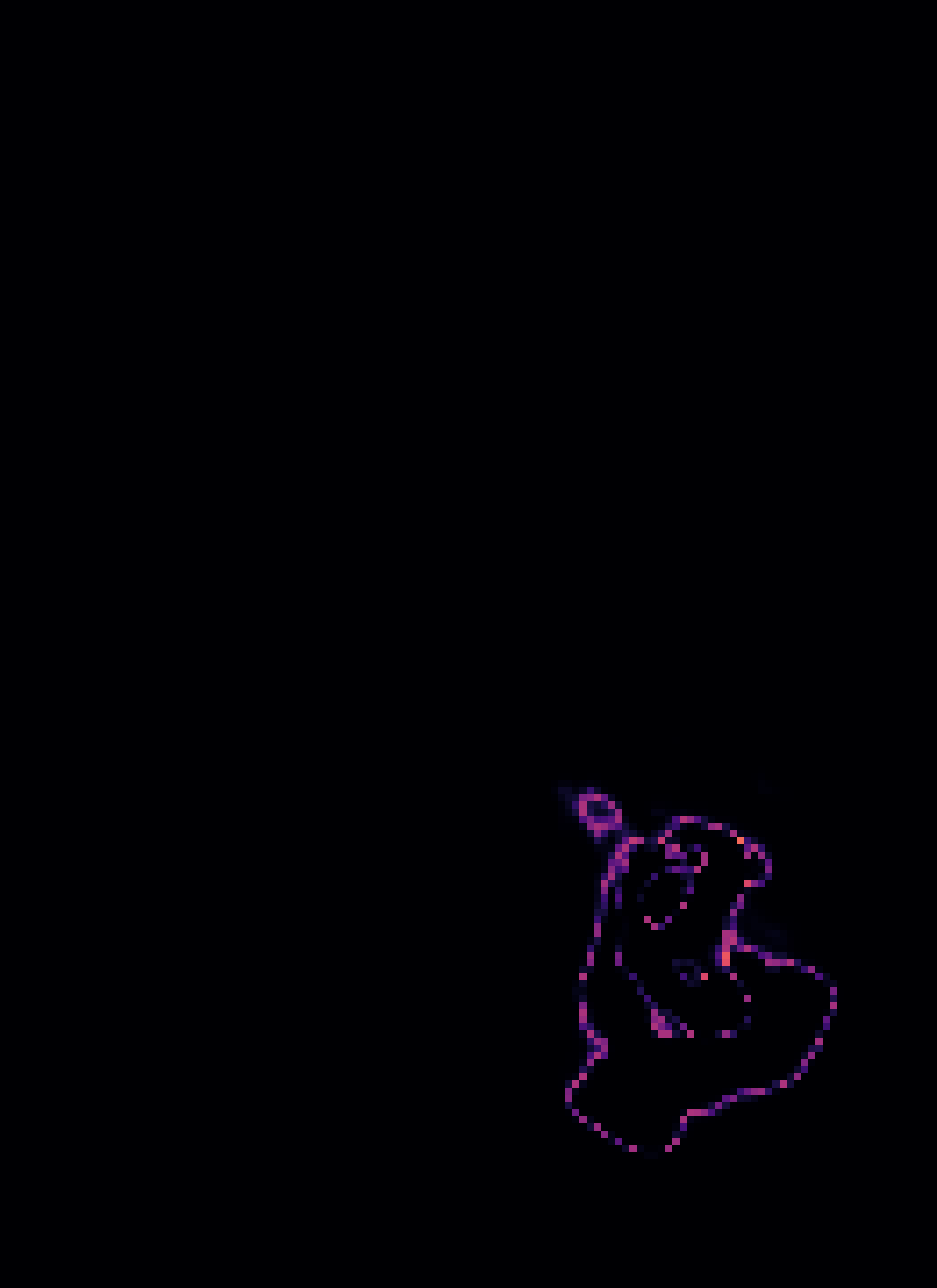}
\end{tabular}
\caption{Additional qualitative BraTS example using SegResNet, case 10, slice 82.}
\label{fig:supp_brats_case10_z82}
\end{figure}
